\documentclass[11pt,letterpaper]{article}
\usepackage[margin=1.1in]{geometry}
\usepackage{amsmath,amssymb,amsfonts,setspace,url,enumitem}
\usepackage{amsthm,hyperref}
\usepackage{tikz}
\usepackage{color}
\usepackage{bigstrut}
\usepackage{multirow}
\usepackage{cleveref} 
\usepackage{ascmac}
\usepackage{mathpazo}

\theoremstyle{plain}
\newtheorem{theorem}{Theorem}
\newtheorem{lemma}{Lemma}

\newtheorem{corollary}{Corollary}

\newtheorem{definition}{Definition}

\newtheorem{proposition}{Proposition}

\theoremstyle{definition}

\newcommand{\norm}[1]{\left\|#1\right\|} 
\newcommand{\smallnorm}[1]{\|#1\|} 
\newcommand{\estnorm}[1]{\widetilde{\left\|#1\right\|}} 
\newcommand{\fnorm}[1]{\norm{#1}_{\mathrm {F}}}    
\newcommand{\set}[2]{\{#1,\ldots,#2\}}
\newcommand{\stat}[1]{\vert#1\vert_{\mathrm{tv}}} 

\newcommand{\ex}[1]{\mathbb{E}\!\left[#1\right]} 
\newcommand{\supp}[1]{\mathrm{supp}(#1)} 
\newcommand{\var}[1]{\mathrm{Var}\left[#1\right]} 
\newcommand{\Prob}[1]{\mathrm{Pr}\left[#1\right]} 
\newcommand{\spec}{\mathrm{Spec}} 
\newcommand{\Tr}{\mathrm{Tr}} 
\newcommand{\ket}[1]{|#1\rangle}

\newcommand{\abs}[1]{\left|#1\right|}
\newcommand{\sample}{\mathsf{Sample}}
\newcommand{\olsi}[1]{\,\overline{\!{#1}\!}\,}

\newcommand{\Comp}{\mathbb{C}}
\newcommand{\Real}{\mathbb{R}}

\providecommand{\Qq}{\mathcal{Q}}

\newcommand{\myp}{\mathfrak{p}}
\newcommand{\myq}{\mathfrak{q}}

\newcommand{\BQP}{\mathsf{BQP}}

\newcommand{\ceil}[1]{\left\lceil #1 \right\rceil}

\newcommand{\poly}{\mathrm{poly}}
\newcommand{\row}[1]{\mathrm{r_{#1}}}
\newcommand{\brow}[1]{\mathrm{{r}_{\bar #1}}}
\newcommand{\browolsi}[1]{\mathrm{{r}_{\olsi #1}}}

\newcommand{\Q}[1]{\mathsf{Q}_{#1}}   
\newcommand{\SQ}[1]{\mathsf{SQ}_{#1}}   
   
\newcommand{\OSQ}[1]{\mathsf{OSQ}_{#1}}
\newcommand{\RSQ}[1]{{\mathsf{RSQ}_{#1}}}
\newcommand{\Samp}[1]{\mathsf{S}_{#1}}

\newcommand{\costq}[1]{\textbf{q}(#1)}
\newcommand{\costs}[1]{\textbf{s}(#1)}
\newcommand{\costsq}[1]{\textbf{sq}(#1)}
\newcommand{\costosq}[1]{\textbf{osq}(#1)}
\newcommand{\costrsq}[1]{\textbf{rsq}(#1)}

\begin{document}

\title{Robust Dequantization of the \\Quantum Singular Value Transformation and \\Quantum Machine Learning Algorithms}
\author{
Fran{\c c}ois Le Gall\\
Nagoya University\\
legall@math.nagoya-u.ac.jp
}

\date{}
\maketitle
\thispagestyle{empty}
\setcounter{page}{1}
\begin{abstract}
Several quantum algorithms for linear algebra problems, and in particular quantum machine learning problems, have been ``dequantized'' in the past few years. These dequantization results typically hold when classical algorithms can access the data via \emph{length-squared sampling}. This assumption, which is standard in the field of randomized linear algebra, means that for a unit-norm vector $u\in\Comp^{n}$, we can sample from the distribution $p_u\colon\{1,\ldots,n\}\to [0,1]$ defined as $p_u(i)=|u(i)|^2$ for each $i\in\{1,\ldots,n\}$. Since this distribution corresponds to the distribution obtained by measuring the quantum state $\ket{u}$ in the computational basis, length-squared sampling access gives a reasonable classical analogue to the kind of quantum access considered in many quantum algorithms for linear algebra problems.

In this work we investigate how robust these dequantization results are. We introduce the notion of \emph{approximate length-squared sampling}, where classical algorithms are only able to sample from a distribution close to the ideal distribution in total variation distance. While quantum algorithms are natively robust against small perturbations, current techniques in dequantization are not. Our main technical contribution is showing how many techniques from randomized linear algebra can be adapted to work under this weaker assumption as well. We then use these techniques to show that the recent low-rank dequantization framework by Chia, Gily\'en, Li, Lin, Tang and Wang (JACM 2022) and the dequantization framework for sparse matrices by Gharibian and Le Gall (STOC 2022), which are both based on the Quantum Singular Value Transformation, can be generalized to the case of approximate length-squared sampling access to the input. We also apply these results to obtain a robust dequantization of many quantum machine learning algorithms, including quantum algorithms for recommendation systems, supervised clustering and low-rank matrix inversion.
\end{abstract}
\newpage
\section{Introduction}

\subsection{Background}
The quantum algorithm for matrix inversion by Harrow, Hassidim and Lloyd \cite{HHL09}, often called simply the HHL algorithm, is a milestone in quantum algorithms for linear algebra. Given a well-conditioned invertible matrix $A\in\Comp^{n\times n}$ and a vector $u\in\Comp^n$, the algorithm prepares in time $\poly(\log n)$ a quantum state proportional to a vector $\hat x\in\Comp^n$ close to the solution of the system of equations $Ax=u$. This quantum state can then be used to compute some partial information about $\hat x$. The HHL algorithm has been very influential for the development of quantum algorithms solving problems related to linear algebra, and has in particular lead to several quantum machine learning algorithms (\cite{Cong+16,Kerenidis+ITCS17,Liu+17,Lloyd+14,Rebentrost+18,Rebentrost+14,Rebentrost+18B,Wiebe+12} for instance). 

An important assumption in the HHL algorithm, and in most of these extensions as well, is that the input vector should be accessible as a quantum state: the quantum computer has access to a few copies of the quantum state $\ket{v}=\sum_{i=1}^n v(i)\ket{i}$ on $O(\log n)$ qubits, for $v=u/\norm{u}$. In order for the HHL algorithm to be useful, it is thus crucial to be able to prepare this quantum state efficiently. One concrete proposal is having the vector $u$ stored in Quantum Random Access Memory (QRAM), in which case
(when using the definition of QRAM in \cite{Kerenidis+ITCS17}) 
one copy of $\ket{v}$ can be created in time $\poly(\log n)$. 
Similar assumptions are needed on how the quantum algorithm can access the matrix $A$. As discussed in~\cite{Aaronson15}, due to all these assumptions it is difficult to directly compare the performance of these quantum algorithms to the performance of classical methods.

A series of works initiated by Tang \cite{TangSTOC19} has been investigating the importance of such assumptions for quantum machine learning. These results have shown that for most quantum machine learning algorithms in the QRAM model, if classical algorithms are allowed to access the input via \emph{length-squared\footnote{Length-squared sampling is also called $\ell_2$-norm importance sampling in the literature.} sampling-and-query access} (which we abbreviate as \emph{sampling-and-query access} in this paper) then the quantum advantage vanishes: classical algorithms can achieve performance similar (i.e., polynomially related) to the performance of quantum algorithms. The concept of (length-squared) sampling-and-query access is a central concept in the literature on linear algebra algorithms based on random sampling \cite{Drineas+04,Drineas+FOCS01,Drineas+SICOMP06,Drineas+06B,Drineas+06C,Frieze+JACM04,Kannan+09, Kannan2017}. 
For a nonzero vector $u\in\Comp^n$, sampling-and-query access to $u$ means that the following two operations can be implemented in $\poly(\log n)$ time: on input $i\in\{1,\ldots,n\}$ we can query the entry $u(i)$; we can get one sample from the distribution $p_u\colon\{1,\ldots,n\}\to [0,1]$ defined as 
$
p_u(i)=|u(i)|^2/\norm{u}^2
$
for each $i\in\{1,\ldots,n\}$. The first operation is standard and similar to the assumption (made in most classical algorithms) that the input is stored in Random Access Memory (RAM). It is the second operation that makes linear algebra algorithms based on random sampling extremely powerful. In the perspective of dequantization, this second operation can be seen as a natural classical analogue of the assumption that the quantum state $\ket{v}$ (with $v=u/\norm{u}$) can be created at cost $\poly(\log n)$, since measuring $\ket{v}$ in the computational basis gives precisely a sample from the distribution $p_u$. 

The main thesis of the line of research on dequantization initiated in \cite{TangSTOC19} is that the performance of quantum machine learning algorithms in the QRAM model should be compared to the performance of classical algorithms with sampling-and-query access to the input. More precisely, the main objective of this approach is to investigate if some quantum advantage persists even when compared to such classical algorithms. The answers obtained so far are unfortunately negative. After Tang's celebrated dequantization \cite{TangSTOC19} of the quantum algorithm for recommendation systems by Kerenidis and Prakash \cite{Kerenidis+ITCS17}, which was one of the main candidates for quantum advantage in machine learning, several works have successively dequantized most known quantum machine learning algorithms in the QRAM model \cite{Chen+19, Chia+JACM22,Chia+20,Chia+MFCS20,Ding+22,Du+20,GG22, Gilyen+18,Jethwani+MFCS20,TangSTOC19}.  The coup de gr\^ace was recently given by Chia, Gily\'en, Li, Lin, Tang and Wang \cite{Chia+JACM22}, who developed a general framework for dequantization applying to almost all known quantum machine learning algorithms, and gave strong evidence for the lack of exponential speedup for quantum machine learning algorithms in the QRAM model.

\subsection{Motivation: Approximate sampling-and-query access}\label{sec:intro2}
In this work, we explore the possibility of dequantizing quantum algorithms in the QRAM model by classical algorithms using a weaker version of sampling-and-query access, which we call $\varepsilon$-approximate (length-squared) sampling-and-query access, where $\varepsilon\in[0,1]$ is a parameter. In $\varepsilon$-approximate sampling-and-query access we weaken the definition as follows: While the query operation is unchanged (on input $i\in\{1,\ldots,n\}$ we can query the entry $u(i)$), for the sampling operation we can only get one sample from a distribution $\tilde p_u\colon\set{1}{n}\to [0,1]$ that is close to $p_u$. More precisely, the only guarantee is that the total variation distance between $\tilde p_u$ and $p_u$ is at most~$\varepsilon$. (Standard sampling-and-query access corresponds to the case $\varepsilon=0$.) This variant corresponds to the setting where query access to $u$ can be easily implemented (e.g., when the input is stored in RAM) but perfect sampling access is not available or simply too costly to be implemented. We stress that in this setting, we do \emph{not} know the value $\tilde p_u(i)$ (and cannot compute it efficiently from the samples and queries).\footnote{If we had access to $\tilde p_u(i)$, we could simply define a vector $\tilde u$ such that $\tilde p_u(i)=|\tilde u(i)|^2/\norm{\tilde u}^2$ and apply the framework from \cite{Chia+JACM22} on $\tilde u$ instead of $u$ (since we could then easily obtain \emph{perfect} sampling-and-query access to~$\tilde u$).}

It happens that this seemingly minor change is problematic for known dequantization techniques. Consider for instance the inner product estimation from \cite{TangSTOC19}, a technique also used in several other ``quantum-inspired'' algorithms, which dequantizes the SWAP test (a basic quantum algorithm that can be used to estimate the inner product between two pure states). Given two unit-norm vectors $u,v\in\Real^n$, where $u$ is given via  sampling-and-query access, the approach typically works as follows: sample an index $i\in\set{1}{n}$ according to the distribution~$p_u$ and output the value $\frac{v(i)}{u(i)}$. Note that the expectation of this value is 
\begin{equation}\label{eq:expectation}
\sum_{i=1}^n p_u(i)\frac{v(i)}{u(i)}=\sum_{i=1}^n \abs{u(i)}^2\frac{v(i)}{u(i)}=\sum_{i=1}^n u(i)v(i)=(u,v).
\end{equation}
It is easy to show that the variance of this estimator is small, and thus taking the mean for a few samples gives a good approximation of the inner product $(u,v)$. Note that this test, however, is not ``robust'' due to the division by $u(i)$. In particular, when only able to sample from $\tilde p_u$ instead of $p_u$, the same strategy does not work since the ratio $\tilde p_u(i)\frac{1}{u(i)}$ can become arbitrarily large, which significantly compromises the quality of the approximation.

Quantum algorithms in the QRAM model, on the other hand, are natively robust: since quantum evolutions are trace-preserving, for any quantum algorithm $\Qq$ and any two quantum states $\ket{\phi}$ and $\ket{\psi}$ such that the trace distance\footnote{We consider the trace distance since it directly gives an upper bound on the total variation distance of the corresponding probability distributions. Using another closeness measure between quantum states (e.g., the fidelity) would lead to another (possibly weaker) upper bound on the total variation distance.}
between $\ket{\phi}$ and $\ket{\psi}$ is at most $\varepsilon$ (which implies that the total variation distance between the distributions $p_{\ket{\phi}}$ and $p_{\ket{\psi}}$ is at most $\varepsilon$), the distance between the output of $\Qq$ on input $\ket{\phi}$ and the output of $\Qq$ on input $\ket{\psi}$ is at most $\varepsilon$. The central motivation of the present work is to understand whether this property of quantum algorithms can be exploited to show a quantum advantage in machine learning. 

To our knowledge, this notion of approximate sampling-and-query access has never been considered in the literature on linear algebra via random sampling.\footnote{The slightly relation notion of oversampling access has been considered in \cite{Drineas+SICOMP06}, and in \cite{Chia+JACM22} in the context of dequantization. 
This notion is conceptually different from ours (our notion only guarantees that the distribution is close to the ideal distribution in total variation distance).}
 The main reason is that in those works, sampling access to the vectors and matrices is usually not given as input. Instead, the algorithm makes one pass over the whole data as a preliminary step in order to collect enough information to be able to implement efficient sampling access in the second phase of the algorithm. The time required by the preliminary step is typically linear in the input size (i.e., the number of nonzero entries of the vectors or matrices), which is enough to get an implementation of perfect sampling-and-query access to the input. In dequantization algorithms, on the other hand, the main goal is typically to construct algorithms with running time exponential faster than the input size, and thus making a pass on the whole data is not possible. In this setting, the implementation of sampling access becomes extremely challenging, and it makes sense to consider the scenario where only approximate sampling access can be implemented.

\subsection{Our results}\label{sec:intro3}
In this work we show that for essentially all known quantum algorithms that have been dequantized in the past few years, robust dequantization is also possible: there exist classical algorithms matching the performance of these quantum algorithms (up to polynomial factors)  even when having only approximate sampling-and-query access to the input. In particular, this gives evidence for the lack of exponential speedup for quantum machine learning algorithms in the QRAM model even when classical algorithms have only approximate sampling-and-query access to the input.\vspace{-2mm}

\paragraph{Approximate oversampling and closeness properties.} Our first contribution is to develop a framework for working with vectors and matrices given by approximate sampling-and-query access. Following the approach developed in \cite{Chia+JACM22} (which used the concept of oversampling-and-query access), we develop our framework by introducing the concept of \emph{approximate} over-sampling-and-query access. In this setting, instead of being able to generate samples from a distribution $\tilde p_u$ close to $p_u$, we can only generate samples from a distribution $\tilde q_u\colon\set{1}{n}\to[0,1]$ such that $\tilde q_u(i)\ge \frac{1}{\phi} \tilde p_u(i)$ holds for all $i\in\set{1}{n}$, for some $\phi\ge 1$ called the oversampling parameter. (Note that taking $\phi=1$ gives the usual approximate sampling-and-query access.) We prove the following two properties for approximate oversampling-and-query access.
\begin{itemize}
\item
Closeness property: given approximate oversampling-and-query access to a small number of vectors, we can implement efficiently approximate oversampling-and-query access to a linear combination of these vectors (with a slightly worse oversampling parameter).
\item
Conversion property:
given approximate oversampling-and-query access to a vector, we can implement efficiently  approximate sampling-and-query access with an overhead in the complexity corresponding to the oversampling parameter. 
\end{itemize}

These two properties generalize the two properties shown (for the setting of perfect sampling-and-query access) in \cite{Chia+JACM22}. 
The closeness property, which usually does not hold for approximate sampling-and-query access, is the reason why the notion of approximate oversampling-and-query access is very convenient to deal with. In addition, the conversion property shows that working with oversampling-and-query access is meaningful, since at the end of the computation it is possible to recover sampling-and-query access from the output (with an overhead in the complexity corresponding to the oversampling parameter).\vspace{-2mm}

\paragraph{Matrix multiplication.} Our second contribution is to develop a robust version of one of the central techniques in the area of randomized numerical linear algebra: matrix multiplication based on importance matrix sketches. Given two matrices $X\in\Comp^{m\times n}$ and $Y\in\Comp^{m\times n'}$, the approach by Drineas, Kannan and Mahoney \cite{Drineas+SICOMP06}, on which the dequantization framework by Chia, Gily\'en, Li, Lin, Tang and Wang \cite{Chia+JACM22} is also based, creates a matrix $\Sigma\in\Real^{r\times m}$ with only one nonzero entry per row, and estimates the matrix product $X^\dagger Y$ by the matrix product $X^\dagger \Sigma^\dagger \Sigma Y$. Note that $X^\dagger \Sigma^\dagger$ is an $n\times r$ matrix obtained by selecting $r$ columns of $X^\dagger$ (i.e., $r$ rows of $X$) and renormalizing them, while $ \Sigma Y$ is an $r\times n'$ matrix obtained by selecting $r$ rows of $Y$ and renormalizing them. The matrix product $X^\dagger \Sigma^\dagger \Sigma Y$ can thus be considered as a ``sketch'' of the matrix product $X^\dagger Y$. The quality of the approximation, i.e., the Frobenius norm
\[
\fnorm{X^\dagger \Sigma^\dagger \Sigma Y-X^\dagger Y}
\]
depends naturally on the choice of $\Sigma$. Drineas, Kannan and Mahoney \cite{Drineas+SICOMP06} analyzed the optimal choice for~$\Sigma$, and showed that this $\Sigma$ can be constructed using (length-squared) sampling-and-query access to the matrices $X$ and $Y$.\footnote{Sampling-and-query access to a matrix is defined using the concept of sampling-and-query access to a vector: we say that we have sampling-and-query access to a matrix $A\in\Comp^{m\times n}$ if we have sampling-and-query access to each row of $A$ and we have sampling-and-query access to the vector of row norms of the matrices $A$, i.e., the vector $(\norm{A(1,\cdot)},\ldots,\norm{A(m,\cdot)})$. This notion has been extended to oversampling-and-query access to a matrix in the literature \cite{Chia+JACM22,Drineas+SICOMP06}. In this work we further extend it to $\varepsilon$-approximate (over)sampling-and-query access to a matrix (see Definitions \ref{def:qs-access-mat} and \ref{def:qos-access-mat} in Section \ref{sub:def2}).} In the dequantization setting, this result has been used in particular to obtain a sketch of the singular value transformation of a matrix \cite[Theorem~5.1]{Chia+JACM22}.

In our setting, where we only have approximate (over)sampling-and-query access to the matrices $X$ and $Y$, several challenges appear when trying to adapt this technique. First, constructing efficiently unbiased estimators for matrix multiplication seems hopeless since the probability distribution from which we can sample is biased. Second, and more importantly, it is even unclear whether the approach from \cite{Chia+JACM22, Drineas+SICOMP06} can be adapted, for the same reasons as mentioned in Section \ref{sec:intro2} for the inner product (observe that inner product is a special case of matrix product). 

Our main technical contribution is to show that a robust version of matrix multiplication is possible, with a small bias depending on the parameter $\varepsilon$. Here are  informal versions of our two results related to matrix multiplication (we refer to Sections \ref{sub:MM1} and \ref{sub:MM} for the formal statements):

\begin{theorem}[Informal version]\label{th:MM1}
Given two matrices $X\in \Comp^{m\times n}$ and $Y\in\Comp^{m\times n'}$, assume that we have $\varepsilon$-approximate sampling-and-query access to $X$, for any $\varepsilon\in[0,1]$. For any $\xi>0$ and any $\delta\in(0,1]$, we can construct efficiently a matrix sketch $\Sigma\in \Real^{r\times m}$ with $r=\Theta\left(\frac{1}{\delta\xi^2}\right)$ such that 
\[
\Prob{
\fnorm{X^\dagger \Sigma^\dagger \Sigma Y-X^\dagger Y}\le \left(2\sqrt{2\varepsilon}+\xi\right)\fnorm{X}\fnorm{Y}}\ge 1-\delta.
\]
\end{theorem}
\begin{theorem}[Informal version]\label{th:MM}
Given two matrices $X\in \Comp^{m\times n}$ and $Y\in\Comp^{m\times n'}$, assume that we have $\varepsilon$-approximate sampling-and-query access to both $X$ and $Y$, for any $\varepsilon\in[0,1]$. For any $\xi>0$ and any $\delta\in(0,1]$, we can construct efficiently a matrix sketch $\Sigma\in \Real^{r\times m}$ with $r=\Theta\left(\frac{\log(1/\delta)}{\xi^2}\right)$ such that 
\[
\Prob{
\fnorm{X^\dagger \Sigma^\dagger \Sigma Y-X^\dagger Y}\le \left(2\varepsilon+\xi\right)\fnorm{X}\fnorm{Y}}\ge 1-\delta.
\]
\end{theorem}
Theorem \ref{th:MM} gives a smaller bias and an exponentially better dependence on $\delta$ than Theorem~\ref{th:MM1}. On the other hand, it requires $\varepsilon$-approximate sampling-and-query access to both $X$ and~$Y$. We note that the dependence in $\delta$ in Theorem \ref{th:MM1} is likely to be inevitable since even with perfect sampling-and-quantum access, it is not known how to achieve a better dependence on~$\delta$ when given access to only one matrix \cite{Drineas+SICOMP06}. 

We then use these techniques to show how to compute a sketch of the singular value transformation of a matrix, similarly to what \cite{Chia+JACM22} has done in the perfect sampling-and-query access setting. Here is the informal version of our result (see Section \ref{sub:SVT} for the formal version):

\begin{theorem}[Informal version]\label{th:SVT}
Let $X\in\Comp^{m\times n}$ be a matrix and $f\colon\Real\to \Comp$ be a smooth function. Assume that we have $\varepsilon$-approximate sampling-and-query access to $X$. For any parameters $\delta\in(0,1]$ and $\gamma>0$, we can efficiently construct two matrices $\Sigma\in \Comp^{r\times m}$ and $C=\Comp^{r\times c}$ with $r,c=\poly(\log(1/\delta),1/\gamma,\fnorm{X})$ such that if $\varepsilon$ is small enough then the inequality
\[
\fnorm{X^\dagger \Sigma^\dagger \bar{f}(CC^\dagger)\Sigma X -f(X^\dagger X)}\le \gamma
\]
holds with probability at least $1-\delta$, where $\bar f(x):=f(x)/x$.
\end{theorem}

Theorem \ref{th:SVT} shows that in order to get a good ``sketch" of a singular value transformation $f(X^\dagger X)$ of the matrix $X^\dagger X$ it is enough to compute $\bar f(C C^\dagger)$ for the much smaller matrix $CC^\dagger\in\Comp^{r\times r}$. This technique is the basis for most of our advanced applications.\vspace{-2mm}

\paragraph{Applications.}

By applying all the above techniques, we are able to robustly dequantize (i.e., dequantize even when the input is given via $\varepsilon$-approximate sampling-and-query access for $\varepsilon>0$) most\footnote{Due to space constraints, in this paper we only present our robust dequantization results for a few important quantum algorithms. Our framework does apply to more quantum algorithms. Actually, we have not been able to find an example of quantum algorithm that can be dequantized but not robustly dequantized.} of the quantum algorithms that have been dequantized in the literature.  Below we briefly discuss all our applications and refer to Section \ref{sec:app} for details. Simplified statements of the complexities of our results and their comparison with prior works can be found in Table~\ref{table:results}.\footnote{In this paper, the notation $\tilde O(\cdot)$ removes factors polylogarithmic in the dimension of the matrices and vectors. }

\begin{table}
\begin{center}
\begin{tabular}{|l|l|l|l|l|}
  \cline{2-4}
  \multicolumn{1}{c|}{\multirow{2}{*}{}}&
  \multicolumn{1}{c|}{\multirow{2}{*}{Quantum}}&
  \multicolumn{1}{c|}{\multirow{2}{85pt}{Classical  (perfect sampling)}} &
  \multicolumn{1}{c|}{\multirow{2}{112pt}{Classical ($\varepsilon$-approxi- mate sampling)}}\\
  \multicolumn{1}{c|}{}&&&\\ \hline
Inner product & $\tilde O(1/\xi^2)$  & $\tilde O(1/\xi^2)$ \cite{TangSTOC19}  & $\tilde O(1/\xi^2)$ \:\:\:\:\:Props.~\ref{th:IP},\:\ref{prop:IP} \\\hline
 \multirow{2}{50pt}{Low-rank QSVT} & 
 \multirow{2}{*}{$\tilde O\left(\frac{d}{\smallnorm{\myp(\sqrt{A^\dagger A})b}}\right)$ \cite{Gilyen+STOC19}} &
 \multirow{2}{*}{$\tilde O\left(\frac{d^{16}}{\smallnorm{\myp(\sqrt{A^\dagger A})b}^6}\right)$ \cite{Chia+JACM22}} & 
 \multirow{2}{*}{$\tilde O\left(\frac{d^{16}}{\smallnorm{\myp(\sqrt{A^\dagger A})b}^6}\right)$ \:Th.~\ref{th:deq-evenSVT}} \\
&&&\\\hline
  \multirow{2}{50pt}{Sparse QSVT} & 
\multirow{2}{*}{$\poly(s,d)$\:\:\:\cite{Gilyen+STOC19}}& 
\multirow{2}{*}{$\tilde O(s^{d})$ \:\:\:\:\:\:\:\cite{GG22}} & 
\multirow{2}{*}{$\tilde O(s^{d})$ \:\:\:\:\:\:\:\:\:\:\:\:\:\:\:Th.~\ref{th:sparseQSVT}} \\
&&&\\\hline
\multirow{2}{90pt}{Sparse matrix inversion} & \multirow{2}{*}{$\poly(s,\kappa,1/\eta)$\:\cite{Gilyen+STOC19,HHL09}} & \multirow{2}{*}{$\tilde O(s^{\kappa\log(\kappa/\eta)})$ \:\:\:\:\cite{GG22}} & \multirow{2}{*}{$\tilde O(s^{\kappa\log(\kappa/\eta)})$\: \:Th.~\ref{th:sparseQSVT}}\\ 
&&&\\\hline
\multirow{2}{50pt}{Supervised clustering} & 
\multirow{2}{*}{$\tilde O(1/\eta)$\:\:\:\:\:\:\,\cite{Lloyd+13}} & 
\multirow{2}{*}{$\tilde O(1/\eta^2)$ \,\,\cite{Chia+JACM22,Tang21}}&  
\multirow{2}{*}{$\tilde O(1/\eta^2)$ \:\:\:\:\:\:\:\:\,\:Th.~\ref{th:clus}} \\
&&&\\\hline
   \multirow{2}{4pt}{Recommendation systems} & \multirow{2}{*}{$\tilde O(\sqrt{K})$ \:\:\:\:\:\:\cite{Kerenidis+ITCS17}} & $\bigstrut\tilde O(K^{12})$ \,\,\,\:\:\cite{TangSTOC19}& \multirow{2}{*}{$\tilde O(K^8)$ \:\:\:\:\:\:\:\:\:\:\:\:\:\,\:Th.~\ref{th:Recom}} \\
   &&$\tilde O(K^8)$ \:\:\:\:\:\:\cite{Chia+JACM22}&\\ \hline
 \multirow{2}{90pt}{Low-rank matrix inversion} & \multirow{2}{*}{$\tilde O(\sqrt{K})$ \:\:\:\:\:\:\cite{Gilyen+STOC19,Gilyen+18}} & \multirow{2}{*}{$\tilde O(K^{14})$ \:\:\:\:\cite{Chia+JACM22}} & \multirow{2}{*}{$\tilde O(K^{14})$\:\:\:\:\:\:\:\:\:\:\:\: \:Th.~\ref{th:inv}}\\ 
&&&\\\hline
\end{tabular}
\end{center}\caption{Overview of our applications and comparison with prior works (the recent works \cite{Bakshi+23,Gilyen+20, Shao+22} are omitted from this table --- we briefly discuss them at the very end of Section~\ref{sec:intro3}). The complexities reported in the last column of this table assume that the parameter $\varepsilon$ is small enough (the precise requirements on $\varepsilon$ are given in the formal statements of the results).}\label{table:results}\vspace{5mm}
\end{table}

\begin{itemize}
\item
{\bf Inner product (Section~\ref{sec:warmingup} and Section~\ref{sub:IP}).} In the quantum setting, the SWAP test can be used to estimate the absolute value of the inner product between two pure states with additive error $\xi$ using $O(1/\xi^2)$ copies of the states.\footnote{If  quantum circuits constructing the two quantum states are available, then the complexity can be further reduced to $O(1/\xi)$ using amplitude estimation.} The dequantization technique from \cite{Chia+JACM22,TangSTOC19} estimates the inner product of two vectors $u$ and $v$ with error $\xi\norm{u}\norm{v}$ in time $O(1/\xi^2)$ given sampling-and-query access to one of the vectors and query access to the other, which matches the complexity of the SWAP test for unit-norm vectors. By applying Theorem \ref{th:MM1}, we immediately obtain a classical algorithm that estimates the inner product with error $(2\sqrt{2\varepsilon}+\xi)\norm{u}\norm{v}$ in time $O(1/\xi^2)$ when given $\varepsilon$-approximate sampling-and-query access to one of the vectors and query access to the other (Proposition \ref{th:IP} in Section~\ref{sec:warmingup} and Proposition \ref{prop:IP} in Section~\ref{sub:IP}). The term $2\sqrt{2\varepsilon}$ here is a bias (coming from the bias in Theorem \ref{th:MM1}) due to the fact that we are sampling from a biased distribution.\footnote{Such a bias is inevitable when estimating the inner product using samples from a biased distribution. Note that a similar bias would appear for quantum algorithms based on the SWAP test in the analogue setting where the input states are at (trace) distance $\varepsilon$ from the ideal states.} This result shows that when $\varepsilon$ is small enough, i.e., when we are sampling from a distribution close enough to the ideal distribution in the total variation distance, we can obtain a performance close to the performance of algorithms having access to the ideal distribution.

\item
{\bf ``Low-rank'' Quantum Singular Value transformation (Section~\ref{sub:app2}).}
Chia, Gily\'en, Li, Lin, Tang and Wang \cite{Chia+JACM22} have shown how to dequantize in the low-rank regime the Quantum Singular Value Transformation (QSVT), which is a recent powerful paradigm developed by Gily\'en, Su, Low and Wiebe \cite{Gilyen+STOC19} that can be used to recover most known quantum algorithms and construct new ones (we refer to \cite{Martyn+21} for a good survey). Given an even polynomial $\myp$ of degree $d$, a matrix $A\in\Comp^{m\times n}$ and a vector~$b$ with $\norm{b}=1$, the QSVT is a technique to construct a good approximation of a quantum state proportional to the vector $\myp(\sqrt{A^\dagger A})b$. In the low-rank regime, which corresponds to the assumption $\fnorm{A}=1$ (or more generally $\fnorm{A}\le 1$), the complexity of the approach is $\tilde O(d/\smallnorm{\myp(\sqrt{A^\dagger A})b})$. Chia, Gily\'en, Li, Lin, Tang and Wang~\cite{Chia+JACM22} showed how to dequantize this technique by proving that given sampling-and-query access to $A$ and $b$, it is possible to implement with high probability sampling-and-query access to a vector close to $\myp(\sqrt{A^\dagger A})b$ in time $\tilde O(d^{16}/\smallnorm{\myp(\sqrt{A^\dagger A})b}^6)$. By adapting the methodology from \cite{Chia+JACM22}, we show how to achieve the same running time when given only $\varepsilon$-approximate sampling-and-query access to $A$ and $b$ for $\varepsilon>0$, when $\varepsilon$ is small enough.

\item
{\bf Sparse Quantum Singular Value Transformation (Section~\ref{sub:app8}).}
Gharibian and Le Gall \cite{GG22} have considered the problem of estimating the value $v^\dagger\myp(\sqrt{A^\dagger A})u$ for a matrix $A\in\Comp^{m\times n}$ such that $\norm{A}= 1$ (or more generally $\norm{A}\le 1$) and two vectors $u,v\in\Comp^{n}$. Note that the condition $\norm{A}\le 1$ is significantly weaker than the condition $\fnorm{A}\le 1$ since the inequality $\norm{A}\le\fnorm{A}$ always holds but $\fnorm{A}$ can be as large as $\sqrt{n}\norm{A}$. On the other hand, with this weaker assumption the problem can only be solved efficiently in the classical setting for special types of matrices, such as $s$-sparse matrices (i.e., matrices containing at most $s$ nonzero entries per row and column). The main result from \cite{GG22} indeed shows that when given sampling-and-query access to an $s$-sparse matrix $A$, a good estimate of the value $v^\dagger\myp(\sqrt{A^\dagger A})u$ can be computed in $\tilde O(s^{d})$ time, which matches (up to polynomial factors) the complexity of quantum algorithms based on the QSVT for constant $s$ and constant~$d$ (Ref.~\cite{GG22} also shows that the problem becomes $\BQP$-hard for  $d=\poly(n)$ and~$s$ constant). By combining our techniques for robust estimation of the inner product with the main technical result from \cite{GG22}, we show that this dequantization result even holds when considering $\SQ{\varepsilon}(v)$ for $\varepsilon>0$, when $\varepsilon$ is small enough.

Using this result, we can generalize the other dequantization results in \cite[Section 4]{GG22}, and in particular obtain an efficient classical algorithm computing a  constant-precision estimation of the ground energy of a local Hamiltonians when given $\varepsilon$-approximate sampling-and-query access to a guiding vector $u$ that has constant overlap with the ground state (Theorem 1 in \cite{GG22} showed this result only for the case $\varepsilon=0$). Similarly, the implications to the Quantum PCP conjecture and the No Low-Energy Samplable States (NLSS) conjecture discussed in \cite[Section 5]{GG22} also generalize to the case of approximate sampling-and-query access to the witness states.

While not considered in \cite{GG22}, the same approach can be used for sparse matrix inversion as well. As discussed in \cite{Gilyen+STOC19} (see also \cite{Martyn+21}), the QSVT can be used for matrix inversion by choosing a polynomial $\myp$ that approximates the function $f(x)=1/x$ on the interval $[-1,1]\setminus(-1/\kappa,1/\kappa)$, where $\kappa$ is the condition number. The analysis from \cite{Gilyen+STOC19} shows that a good enough approximation can be obtained by a polynomial of degree $d=O(\kappa\log(\kappa/\eta))$, where $\eta$ is the needed precision. The dequantization technique from~\cite{GG22} then leads to a classical algorithm with complexity $\tilde O(s^{\kappa\log(\kappa/\eta)})$  given exact sampling-and-query access to the input. Our result shows that this complexity holds even when considering approximate sampling-and-query access.

\item
{\bf Supervised clustering (Section~\ref{sub:app3}).} Lloyd, Mohseni, and Rebentrost \cite{Lloyd+13} developed a quantum algorithm for the supervised clustering problem, an important problem in machine learning. From a computational perspective, this quantum algorithm simply computes the inner product between two vectors representing appropriately the data, which can be done in $\tilde O\left(1/\eta\right)$ time using the SWAP test, where $\eta$ represents the precision parameter. Prior dequantization works \cite{Chia+JACM22,Tang21} showed how to construct classical algorithms solving this problem in time $O\left(1/\eta^2\right)$ given sampling-and-query access to the input. By applying our algorithm for robust estimation of the inner product, we immediately obtain the same complexity when given only $\varepsilon$-approximate sampling-and-query access for $\varepsilon>0$.

\item
{\bf Recommendation systems (Section~\ref{sub:app4}).}
In recommendation systems, the main computational task reduces to sampling from a row of a low-rank matrix close to the input matrix $A$ (which corresponds to the users' preference matrix). A key parameter for expressing the complexity of quantum algorithms for recommendation systems \cite{Kerenidis+ITCS17} and their dequantization \cite{Chia+JACM22,TangSTOC19} is the ratio $K=\fnorm{A}^2/\sigma^2$, where $\sigma$ is a threshold value used for the specification of the problem. Kerenidis and Prakash \cite{Kerenidis+ITCS17} have shown how to solve the problem in time $\tilde O(\sqrt{K})$ on a quantum computer when~$A$ is given in QRAM. Ref.~\cite{Chia+JACM22} gave a classical algorithm solving the problem in time $\tilde O(K^{16})$ when $A$ is given via sampling-and-query access, which matches the complexity of the quantum algorithm up to a (large) polynomial factor and improved the complexity of Tang's original algorithm~\cite{TangSTOC19}. By applying our framework, we show how to obtain the same complexity when given only $\varepsilon$-approximate sampling-and-query access to the input for $\varepsilon>0$, when~$\varepsilon$ is small enough.

\item
{\bf ``Low-rank'' matrix inversion (Section~\ref{sub:app5}).}
While the matrix inversion problem solved by the HHL quantum algorithm is $\BQP$-hard \cite{HHL09} and thus unlikely to be dequantized, several low-rank versions of matrix inversion (or, more precisely, solving a system of linear equations) have been studied in recent works on dequantization \cite{Chia+20,Chia+JACM22,Gilyen+18}. Ref.~\cite{Chia+JACM22}, in particular, showed that classical algorithms can solve linear equations in time $\tilde O(K^{14})$ given sampling-and-query access to the matrix and the vector specifying the system of linear equations. Since quantum algorithms solve this problem in time $\tilde O(\sqrt{K})$ when~$A$ is given in QRAM \cite{Gilyen+STOC19,Gilyen+18}, the classical algorithms match this complexity up to a (large) polynomial factor. By applying our framework, we show how to obtain the same complexity when given only $\varepsilon$-approximate sampling-and-query access for $\varepsilon>0$, when $\varepsilon$ is small enough. 
\end{itemize}

All these results give strong evidence for the lack of exponential speedup for quantum machine learning algorithms in the QRAM model (and additionally for the framework from~\cite{GG22}) even when the classical algorithms only have approximate sampling-and-query access to the input. The main conceptual message of this paper is thus that no quantum advantage seems to arise in machine learning (at least for the problems we consider) from the ability of quantum algorithms to natively robustly manipulate classical information encoded as quantum states. 

\paragraph{More recent quantum-inspired algorithms.} Very recently, Bakshi and Tang \cite{Bakshi+23} have constructed quantum-inspired classical algorithms for the Quantum Singular Value transformation and its applications with improved bounds compared with \cite{Chia+JACM22}. For low-rank matrix inversion, recent works by Gily{\'{e}}n, Song and Tang \cite{Gilyen+20} and Montanaro and Shao \cite{Shao+22} have given quantum-inspired classical algorithms with improved bounds over \cite{Chia+JACM22} as well. 
These improved algorithms, which use iterative methods and new classical techniques, assume perfect sampling-and-query access to the data, as in \cite{Chia+JACM22}. It would be interesting to obtain similar improvements in the case of approximate sampling-and-query access. 
\subsection{Overview of our techniques}
Our work adapts the frameworks developed by Chia, Gily\'en, Li, Lin, Tang and Wang \cite{Chia+JACM22} and Gharibian and Le Gall \cite{GG22} to the case of approximate sampling-and-query access to the input. In particular, we need to adapt essentially all the definitions and technical lemmas from \cite{Chia+JACM22}, as well as their proofs. From a technical perspective, the three main contributions are the three techniques described below. \vspace{-2mm}

\paragraph{Robust estimation of the inner product.}
We describe the main idea behind our robust estimation of the inner product. Consider two unit-norm vectors $u,v\in\Real^n$, where $u$ is given via $\varepsilon$-approximate sampling-and-query access and $v$ is given by query access (i.e., on an input $i\in\set{1}{n}$ we can query $v(i)$). A straightforward adaptation of the method described in Section \ref{sec:intro2} would be: sample an index $i\in\set{1}{n}$ according to the distribution~$\tilde p_u$ and output the value $v(i)/u(i)$. Comparing with Equation~(\ref{eq:expectation}), we see that indices $i$ such that $u(i)$ is very small but $\tilde p_u(i)$ is large lead to large bias. This suggests the followings strategy: use instead an estimator of the type
\[
\begin{cases}
0&\textrm{ if } i\in \Gamma,\\
\frac{v(i)}{u(i)}& \textrm{ if }i\notin \Gamma,
\end{cases}
\]
for a set $\Gamma=\{i\in\set{1}{n}\:|\:\abs{u(i)}\le\theta(i)\}$, where $\theta(i)$ is a well chosen threshold function. Ideally, we would like to choose $\theta(i)$ based on the value of $\tilde p_u(i)$. This is unfortunately not possible since we do not have access to the value $\tilde p_u(i)$. Our key discovery is that choosing $\theta(i)$ based on the value of $v(i)$ does work. Taking such as estimate leads to a biased estimator where the bias comes from two parts: the first part coming from the fact that we are sampling from $\tilde p_u$ instead of $p_u$ and the second part coming from the fact that we disregarding the contribution of the indices in the set $\Gamma$. When setting $\theta(i)=\gamma \abs{v(i)}$ for some constant $\gamma$, we discover that it is possible to derive ``complementary'' upper bounds on these two biases: the upper bound on the first bias is proportional to $\gamma$ while the upper bound on the second bias is inverse-proportional to $\gamma$. We then take the optimal value for $\gamma$ and set $\theta(i)$ accordingly. Concretely, we take $\theta(i)=\sqrt{2\varepsilon}\abs{v(i)}$ (with additional normalization factors when $u$ and $v$ are not unit-norm), which gives a bias of $\sqrt{2\varepsilon}+\sqrt{2\varepsilon}=2\sqrt{2\varepsilon}$. We refer to Section~\ref{sec:warmingup} for a fairly self-contained presentation of the details of our technique for robust estimation of the inner product.\footnote{The inner product estimation algorithm of Section \ref{sec:warmingup} is actually a special case of Theorem \ref{th:MM1} in Section \ref{sub:MM1}. We present it in Section \ref{sec:warmingup} as a ``warm-up'' since for this special case the presentation is significantly lighter (in particular, we do not need the notion of matrix sketch introduced in in Section \ref{sub:MM1}).}\vspace{-2mm}

\paragraph{Extension to robust estimation of matrix multiplication.}
Our second main technical contribution is extending this strategy to compute a general matrix product $XY$. We adapt the notion of matrix sketch, which is one of the central concepts in linear algebra algorithms based on random sampling \cite{Drineas+04,Drineas+FOCS01,Drineas+SICOMP06,Drineas+06B,Drineas+06C,Frieze+JACM04,Kannan+09, Kannan2017}, to the setting of approximate sampling (see Definition~\ref{def:S} in Section \ref{sub:MM1}). This is the key step that enables us to adapt the matrix multiplication sketching technique by Drineas, Kannan and Mahoney \cite{Drineas+SICOMP06} to the setting where only $\varepsilon$-approximate sampling-and-query access to one matrix is given, for $\varepsilon>0$. When given approximate sampling access only to one of the matrices ($X$ or $Y$), we apply a technique similar to our robust evaluation of the inner product (indeed, we obtain the $2\sqrt{2\varepsilon}$ term in the bias of the sketch in Theorem~\ref{th:MM1}). When given approximate sampling access to both matrices $X$ and $Y$, we adapt the technique based on joint-matrix sampling from \cite{Drineas+SICOMP06}, which was also used in~\cite{Chia+JACM22}. The first step is again to adapt the definition to the setting of approximate sampling (Definition~\ref{definition:mixed} in Section \ref{sub:MM}). Using the inequality of arithmetic and geometric means and a finer analysis of the success probability based on MacDiarmid's bounded difference inequality, this technique leads to an algorithm (Theorem \ref{th:MM}) with exponentially better dependence in the error parameter~$\delta$.\vspace{-2mm}

\paragraph{Robust rejection sampling.}
Our third main technical contribution, crucial to dequantize quantum algorithms for the low-rank quantum singular value transformation and recommendation systems, is the conversion technique mentioned in Section \ref{sec:intro3}, which converts approximate oversampling-and-query access into approximate sampling-and-query. A similar conversion has been used in the perfect setting in \cite{Chia+JACM22,TangSTOC19}, based on the standard technique in probabilistic algorithms called rejection sampling (see, e.g., \cite{Devroye1986} for a detailed description of this technique). Our main technical contribution is to develop a robust version of rejection sampling. 

Standard rejection sampling is a method to generate samples from a (hard) distribution $p_1\colon\set{1}{n}\to [0,1]$ given the ability to generate samples from another (easier) distribution $p_2\colon\set{1}{n}\to [0,1]$. The method works as follows: sample $j\in\set{1}{n}$ from $p_2$; output it with probability 
$
\frac{p_1(j)}{mp_2(j)}
$
for a large value $m$ (large enough such that $p_1(j)\le m p_2(j)$ holds for all $j\in\set{1}{m}$) and otherwise report ``failure''. It is easy to see that the probability of outputting $j$ is 
$p_1(j)/m$
and the probability that the process does not fail is $1/m$. In consequence, conditioned on the event that the process does not fail, the probability that sample $j$ is output is precisely $p_1(j)$. Repeating the process $\Theta(m)$ times will then output a sample drawn from the distribution $p_1$ with high probability. In Section \ref{sub:over}, we show how to generalize this technique to the case where we can only generate samples from a probability distribution $\tilde p_2$ that is close to $p_2$ in total variation distance. At first glance, it seems that the same difficulty as when trying to extend Equation~(\ref{eq:expectation}) to the approximate sampling setting arises. We show that the condition $p_1(j)\le m p_2(j)$ is actually strong enough to control the bias, which enables us to sample with high probability from a probability distribution close to $p_1$ in total variation distance (see Proposition \ref{prop:oversampling} in Section \ref{sub:over}).
\section{Preliminaries}\label{sec:prelim}
In this section we present the definitions needed for this work. In Section \ref{sub:notations} we give general definitions and notations. In Section \ref{sub:def1} we define several notions of access to a vector and in particular introduce the concept of approximate sampling-and-query access, which is one of the main conceptual contributions of this work. In Section \ref{sub:def2}, we extend these definitions to matrices.

\subsection{Notations and general definitions}\label{sub:notations}

\paragraph{General notations.}
In this paper we assume that arithmetic operations in $\Comp$ (i.e., addition, subtraction, multiplication and division of two scalars) can be done at unit cost.\footnote{All our upper bounds can easily be converted into upper bounds for the bit-complexity of the problem as well if the notions of query access and sampling access to the input are redefined appropriately to take in consideration the bit-complexity of the input vectors and matrices.}  For any complex number $x$, we write $x^\ast$ its complex conjugate.
We denote by $\Real_{\ge 0}$ the set of non-negative real numbers. 
We say that a polynomial $\myp\in\Comp[x]$ is even if $\myp(x)=\myp(-x)$ for all $x\in \Comp$.
Given a function $f\colon\Real\to \Comp$, some set $\Lambda\subseteq \Real$ and a positive real number $L$, we say that $f$ is $L$-Lipschitz on $\Lambda$ if the inequality $\abs{f(x)-f(y)}\le L\abs{x-y}$ holds for any $x,y\in \Lambda$. 


\paragraph{Matrices and vectors.}

For a matrix $A\in\Comp^{m\times n}$, we write its entries as $A(i,j)$ for $(i,j)\in\set{1}{m}\times \set{1}{n}$. For each $i\in\set{1}{m}$, we denote the $i$-th row of $A$ as $A(i,\cdot)$ . For any $j\in\set{1}{n}$, we denote the $j$-th column of $A$ as $A(\cdot,j)$. We write $A^\ast$ for its conjugate and $A^\dagger$ for its transpose-conjugate, i.e., the matrices such that $A^\ast(i,j)=A(i,j)^*$ and $A^\dagger(i,j)=A(j,i)^*$ for all $(i,j)\in\set{1}{m}\times \set{1}{n}$. We use similar notations for vectors.

For a vector $u\in \Comp^{n}$, we denote its $\ell_2$ norm as $\norm{u}$. For a matrix $A\in\Comp^{m\times n}$, we denote its spectral norm as $\norm{A}$, and denote its Frobenius norm as $\fnorm{A}$, i.e.,
\[
\fnorm{A}=\sqrt{\sum_{i=1}^m\sum_{j=1}^n\abs{A(i,j)}^2}.
\]
We have $\norm{A}\le \fnorm{A}\le \sqrt{\min(m,n)}\norm{A}$.\footnote{In this paper we systematically use the upper bound $\norm{A}\le \fnorm{A}$ to state the bounds on the complexity of our algorithms only in term of $\fnorm{A}$. It is possible (and done in prior works \cite{Chia+JACM22}) to give a finer analysis of the statements of upper bounds using both $\norm{A}$ and $\fnorm{A}$ but for clarity we abstain from doing this since this gives more complicated bounds.}

Given a Hermitian matrix $H\in\Comp^{n\times n}$, we order its eigenvalues in non-increasing order and denote them $\lambda_1(H)\ge \lambda_2(H)\ge \cdots \ge \lambda_n(H)$. We write
\[
\spec(H)=\{\lambda_1(H), \lambda_2(H), \ldots, \lambda_n(H)\}
\] 
for the spectrum of $H$.

In Section \ref{sub:SVT} we will need the following two lemmas.
 
\begin{lemma}[Hoﬀman-Wielandt theorem \cite{HW53}]\label{lemma:HW}
For Hermitian matrices $X,Y\in\Comp^{m\times m}$, we have
\[
\sum_{i=1}^m\abs{\lambda_i(X)-\lambda_i(Y)}^2\le \fnorm{X-Y}^2.
\]
\end{lemma}

\begin{lemma}[Corollary 2.3 in \cite{Gil10}]\label{lemma:Lipschitz}
For two Hermitian matrices $X,Y\in\Comp^{m\times m}$ and a function $f\colon \Real\to\Comp$ that is $L$-Lipschitz on $\spec(X)\cup\spec(Y)$, we have
\[
\fnorm{f(X)-f(Y)}\le L\fnorm{X-Y}.
\]
\end{lemma}

For a positive semidefinite Hermitian matrix $H\in\Comp^{n\times n}$, consider its spectral decomposition
\[
    H=\sum_{i=1}^{n}\lambda_i(H) v_iv_i^\dagger.
\]
where the $\{v_i\}_i$ are orthonormal vectors in $\Comp^n$.
For any function $f\colon\Real_{\ge 0}\to\Comp$, the matrix 
\[
f(H)=\sum_{i=1}^{n}f(\lambda_i(H)) v_iv_i^\dagger
\]
is called the \emph{singular value transformation} of $H$ associated to $f$ (which in this case coincides with the eigenvalue transformation of $H$ associated to $f$). We use this definition to define the concept, which we will consider in Sections \ref{sub:app2} and \ref{sub:app8}, of (even) polynomial transformation of an arbitrary (i.e., non-Hermitian) matrix. For an even polynomial $\myp\in\Comp[x]$ and a matrix $A\in\Comp^{m\times n}$, the \emph{(even) polynomial transformation} of $A$ associated to $\myp$ is defined as $\myp(\sqrt{A^\dagger A})$. Let us write
\[
\myp(x)=a_0+a_2x^2+a_4x^4+\cdots+ a_{d}x^{d},
\]
where $d$ is the (even) degree of $\myp$ and $a_0,\ldots,a_{d}$ are coefficients. It is easy to check that the following relation holds:
\[
\myp(\sqrt{A^\dagger A})=a_0I_n+a_2A^\dagger A+a_4(A^\dagger A)^2+\cdots+a_{d}(A^\dagger A)^d,
\]
where $I_n$ denotes the identity matrix of dimension $n$. 

We now define the singular value transformation of an arbitrary matrix. This concept will be used in Sections \ref{sub:app4} and \ref{sub:app5}. 
For any matrix $A\in \Comp^{m\times n}$, consider its spectral singular value transformation
\[
A=\sum_{i=1}^{\min(m,n)}\sigma_i u_iv_i^\dagger,
\]
where the $\sigma_i$'s are nonnegative real numbers, $\{u_i\}_i$ are orthonormal vectors in $\Comp^m$ and $\{v_i\}_i$ are orthonormal vectors in $\Comp^n$. For any function $f\colon\Real_{\ge 0}\to\Comp$ such that $f(0)=0$, the matrix 
\[
f(A)=\sum_{i=1}^{\min(m,n)}f(\sigma_i) u_iv_i^\dagger
\]
is called the \emph{singular value transformation} of $A$ associated to $f$. Note that when $A$ is a positive semidefinite Hermitian matrix we recover the above definition.\footnote{For arbitrary matrices, the condition $f(0)=0$ guarantees that the singular value transformation is well defined.}

\paragraph{Probability distributions and matrix-valued random variables.}

Given a probability distribution $p\colon\set{1}{n}\to[0,1]$, we denote its support as $\supp{p}$, i.e., 
\[
\supp{p}=\{i\in\set{1}{n}\:|\:p(i)>0\}.
\]

Given two probability distributions $p,q\colon\set{1}{n}\to[0,1]$ we write their total variation distance as
\[
\stat{p-q}=\frac{1}{2}\sum_{i=1}^n \abs{p(i)-q(i)}.
\]

As already mentioned in the introduction, for any nonzero vector $u\in \Comp^{n}$ we define the probability distribution $p_u\colon\{1,\ldots,n\}\to [0,1]$ as 
\[
p_u(i)=\frac{|u(i)|^2}{\norm{u}^2}
\]
for each $i\in\{1,\ldots,n\}$.

We will often work with matrix-valued random variables. We recall the definition of the expectation and variance of such variables. A matrix-valued random variable $Z\in\Comp^{m\times n}$ is a matrix where each entry $Z(i,j)$ is a complex-valued random variable, i.e., $Z(i,j)=X(i,j)+i\:Y(i,j)$ for two real-valued random variables $X(i,j)$ and $Y(i,j)$. The expectation of $Z$ is the matrix $C\in\Comp^{m\times n}$ such that the $C(i,j)=\ex{X(i,j)}+i\:\ex{Y(i,j)}$. The variance of~$Z$ is 
\[
\var{Z}=\ex{\fnorm{Z-\ex{Z}}^2}=\ex{\fnorm{Z}^2}-\fnorm{\ex{Z}}^2.
\] 
With these definitions, Chebyshev's inequality (which is just Markov's inequality applied to the real-valued random variable $\fnorm{Z-\ex{Z}}^2$) holds for matrix-valued random variables.

We will also need MacDiarmid's bounded difference inequality \cite{McDiarmid89}.
\begin{lemma}[MacDiarmid's bounded difference inequality \cite{McDiarmid89}]\label{lemma:MacDiarmid}
Let $X_1, \ldots, X_r$ be independent random variables, where $X_i$ takes values in a set $A_i$ for each $i\in\set{1}{r}$.  Let $f\colon A_1\times\cdots\times A_r\to \Real$ be a function satisfying the following property: for any $i\in\set{1}{r}$, 
\[
\abs{f(x_1,\ldots,x_r)-f(x'_1,\ldots,x'_r)}\le d_i
\] 
whenever the vectors $(x_1,\ldots,x_r)$ and $(x'_1,\ldots,x'_r)$ differ in just the $i$-th coordinate.
Then for any $t>0$, 
\[
\Prob{\abs{f(X_1,\ldots,X_r)-\ex{f(X_1,\ldots,X_r)}}\ge t}\le 2\exp\left(-\frac{2t^2}{\sum_{i=1}^r d_i^2}\right).
\]
\end{lemma}

\paragraph{Powering lemma.}
In order to amplify the success probability of probabilistic estimators, we will often use the following version of the ``powering lemma'' from \cite{Jerrum+86}. Since our version is slightly more general than the original version (our version applies to complex-valued random variables as well), we include a proof for completeness. 
\begin{lemma}\label{lemma:powering}
Consider a randomized algorithm that produces an estimate $\tilde \mu$ of a complex-valued quantity $\mu$ such that $\abs{\tilde \mu-\mu}\le \varepsilon$ holds with probability at least $3/4$. Then, for any $\delta>0$, it suffices to repeat $O(\log(1/\delta))$ times the algorithm to obtain an estimate $\hat \mu$ such that $\abs{\hat \mu-\mu}\le \sqrt{2}\varepsilon$ holds with probability at least $1-\delta$.
\end{lemma}
\begin{proof}
Let us write $\mu=\mu_1 + i\mu_2$ with $\mu_1,\mu_2\in\Real$.
Each application of the algorithm returns a complex number $\tilde \mu=\tilde \mu_1+i\tilde \mu_2$ with $\tilde \mu_1,\tilde \mu_2\in\Real$ such that both $\abs{\tilde \mu_1-\mu_1}\le \varepsilon$ and $\abs{\tilde \mu_2-\mu_2}\le \varepsilon$ hold with probability at least $3/4$. We repeat the algorithm $m$ times. We compute the median of the real parts of the $m$ estimates and write it $\hat \mu_1$. Similarly, we compute the median of the imaginary parts of the $m$ estimates and write it $\hat \mu_2$. Finally, we output $\hat \mu= \hat \mu_1+i\hat \mu_2$.

By taking $m=O(\log(1/\delta))$, Chernoff's bound guarantees that the inequality $\abs{\hat \mu_1-\mu_1}\le \varepsilon$ holds with probability at least $1-\delta/2$, and the inequality $\abs{\hat \mu_2-\mu_2}\le \varepsilon$ holds with probability at least $1-\delta/2$. The inequality $\abs{\hat \mu -\mu}\le \sqrt{2}\varepsilon$ then holds with probability at least $1-\delta$.
\end{proof}

\subsection{Sampling-and-query access to vectors}\label{sub:def1}
In this subsection we define several versions of sampling access to vectors. The relations between these notions is summarized in Figure \ref{fig1}.


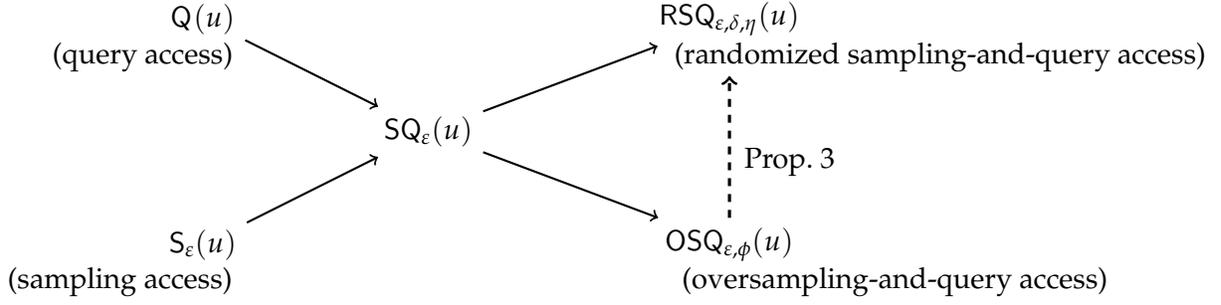
\begin{figure}
\begin{center}
    \begin{tikzpicture}
\node (n1) at (0,3)   {$\Q{}(u)$}; 
\node (t1) at (-0.8,2.5)   {(query access)}; 
\node (n2) at (0,0)  {$\Samp{\varepsilon}(u)$};
\node (t2) at (-1.1,-0.5)   {(sampling access)}; 
\node (n3) at (3,1.5)  {$\SQ{\varepsilon}(u)$};
\node (n4) at (7,3)  {$\RSQ{\varepsilon,\delta,\eta}(u)$};
\node (m4) at (7,2.6)  {$\phantom{\RSQ{\varepsilon,\delta,\eta}(u)}$};
\node (t4) at (9.8,2.5)   {(randomized sampling-and-query access)}; 
\node (n5) at (7,0)  {$\OSQ{\varepsilon,\phi}(u)$};
\node (t5) at (9.2,-0.5)   {(oversampling-and-query access)}; 
\node (t5) at (7.8,1.1)   {Prop.~\ref{prop:oversampling}}; 
\draw [->,thick] (n1) to (n3);
\draw [->,thick] (n2) to (n3);
\draw [->,thick] (n3) to (n4);
\draw [->,thick] (n3) to (n5);
\draw [->,dashed,very thick] (n5) to (m4);
    \end{tikzpicture}  \vspace{-2mm}
\end{center}
\caption{Relation between the notions of sampling access to a vector defined in Section \ref{sub:def1}. A plain arrow shows a trivial generalization. The dotted arrow refers to the implementation proved in Proposition \ref{prop:oversampling} (which modifies the parameter $\varepsilon$).}\label{fig1}
\end{figure}

\paragraph{Basic definitions.}
We first define query access and (approximate) sampling access to a vector.
\begin{definition}\label{def:q-access}
For a nonzero vector $u\in \Comp^{n}$, we say that we have query access to $u$, which we write $\Q{}(u)$, if on input $i\in\{1,\ldots,n\}$ we can query the entry $u(i)$. We denote $\costq{u}$ the cost of implementing one such query. 
\end{definition}

\begin{definition}\label{def:s-access}
For a nonzero vector $u\in \Comp^{n}$, we say that we have $\varepsilon$-approximate sampling access to $u$, which we write $\Samp{\varepsilon}(u)$, if we can sample from a distribution $\tilde p\colon\{1,\ldots,n\}\to [0,1]$ such that $\stat{p_u-\tilde{p}}\le \varepsilon$. We denote $\costs{u}$ the cost of generating one sample. 
\end{definition}

We now define the concept of (approximate) sampling-and-query access, which is the main definition of this paper.

\begin{definition}\label{def:sq-access}
For a nonzero vector $u\in \Comp^{n}$ and a parameter $\varepsilon\ge 0$, we say that we have $\varepsilon$-approximate sampling-and-query access to~$u$, which we write $\SQ{\varepsilon}(u)$, if we have $\Q{}(u)$ and $\Samp{\varepsilon}(u)$ and 
additionally can get the value $\norm{u}$.
We write $\costsq{u}$ the maximum among $\costq{u}$, $\costs{u}$ and the cost of obtaining $\norm{u}$. When $\varepsilon=0$, we simply write $\SQ{}(u)$.
\end{definition}

Knowledge of $\norm{u}$ is required since this norm cannot be easily computed from $\Q{}(u)$ and $\Samp{\varepsilon}(u)$.
The literature on sampling methods for linear algebra, as well as recent works on dequantization \cite{Chia+JACM22,TangSTOC19}, have used $\SQ{}(u)$ (i.e., Definition \ref{def:sq-access} with $\varepsilon=0$).\footnote{We also mention the work by Cotler, Huang and McClean \cite{Cotler+21}, which compares $\SQ{}(u)$ with quantum access to a quantum state proportional to $u$ and shows that $\SQ{}(u)$ can in some cases be more powerful due to the query access allowed in the classical setting (this advantage naturally disappears if we allow quantum algorithms to use $\Q{}(u)$).} We are not aware of any prior work that introduces a similar definition with $\varepsilon>0$.

\paragraph{More technical definitions.}
We define two generalizations of our central definition (Definition \ref{def:sq-access}) that will be useful to state our results. 

The first one is a slight generalization of Definition \ref{def:sq-access}, which only requires to get an approximation of $\norm{u}$ and only requires randomized procedures for sampling and norm estimation. This generalization is needed since in some applications we will not be able to implement sampling with probability one or compute exactly the norm of the vector. 

\begin{definition}\label{def:rand-sample}
For a nonzero vector $u\in \Comp^{n}$ and parameters $\epsilon,\delta,\eta\ge 0$, we say that we have $(\varepsilon,\delta,\eta)$-approximate randomized sampling-and-query access to $u$, which we write $\RSQ{\varepsilon,\delta,\eta}(u)$, if the following three conditions are satisfied:
\begin{itemize}
\item[(i)]
we have $\Q{}(u)$;
\item[(ii)]
we can get an estimate $\estnorm{u}$ such that $\abs{\estnorm{u}-\norm{u}}\le \eta \norm{u}$ holds with probability at least $1-\delta$;
\item[(iii)]
we can generate with probability at least $1-\delta$ a sample from a distribution $\tilde p\colon\{1,\ldots,n\}\to [0,1]$ such that $\stat{p_u-\tilde{p}}\le \varepsilon$.\footnote{This means that we have a Las-Vegas algorithm that samples from $\tilde p$ with probability at least $1-\delta$: the algorithm outputs an error message with probability at most $\delta$, but when it does not output an error message then it returns a sample from the distribution $\tilde p$.}
\end{itemize}
We write $\costrsq{u}$ the maximum among $\costq{u}$, the cost of getting $\estnorm{u}$ and the cost of sampling from $\tilde p$.
\end{definition}

The second one is a more significant generalization based on the concept of oversampling from \cite{Drineas+SICOMP06}. This is a generalization of the definition from \cite{Chia+JACM22}, which corresponds to the case $\varepsilon=0$. This generalization is convenient since (as we will show later) it has several useful closure properties.

\begin{definition}\label{def:over-sample}
For a nonzero vector $u\in \Comp^{n}$ and parameters $\varepsilon\ge 0$ and $\phi\ge 1$, we say that we have $(\varepsilon,\phi)$-approximate oversampling-and-query access to $u$, which we write $\OSQ{\varepsilon,\phi}(u)$, if the following two conditions are satisfied:
\begin{itemize}
\item[(i)]
we have $\Q{}(u)$;
\item[(ii)]
we have $\SQ{\varepsilon}(\bar u)$ for a vector $\bar u\in\Comp^n$ such that $\norm{\bar u}^2=\phi\norm{u}^2$ and $\abs{\bar u(i)}\ge \abs{u(i)}$ for all $i\in\set{1}{n}$.
\end{itemize}
We denote $\costosq{u}=\max\{\costq{u},\costsq{\bar u}\}$.
\end{definition}

\subsection{Sampling-and-query access to matrices}\label{sub:def2}
We now define sampling access to matrices. We first state the definition of query access to a matrix, which is a direct generalization of the concept of query access to a vector (Definition \ref{def:q-access}). 
\begin{definition}\label{def:q-access-mat}
For a nonzero matrix $A\in \Comp^{m\times n}$, we say that we have query access to $A$, which we write $\Q{}(A)$, if on input $(i,j)\in\{1,\ldots,m\}\times\{1,\ldots,n\}$ we can query the entry $A(i,j)$. We denote $\costq{A}$ the cost of implementing such a query.
\end{definition}

For a matrix $A\in \Comp^{m\times n}$, we denote $\row{A}$ the vector in $\Comp^m$  such that the $i$-th coordinate is $\norm{A(i,\cdot)}$, for each $i\in\set{1}{m}$. We now introduce our central definition of approximate sampling-and-query access to a matrix, which corresponds to approximate sampling-and-query access to all rows of $A$ and also to $\row{A}$. This is a generalization of the standard definition from the literature, which corresponds to the case $\varepsilon=0$.

\begin{definition}\label{def:qs-access-mat}
For a nonzero matrix $A\in \Comp^{m\times n}$ and a parameter $\varepsilon\ge 0$, we say that we have $\varepsilon$-approximate sampling-and-query access to $A$, which we write $\SQ{\varepsilon}(A)$, if the following two conditions are satisfied:
\begin{itemize}
\item[(i)]
for each $i\in\{1,\ldots,m\}$ such that $\norm{A(i,\cdot)}\neq 0$, we have $\SQ{\varepsilon}(A(i,\cdot))$;
\item[(ii)]
we have $\SQ{\varepsilon}(\row{A})$.
\end{itemize} 
We denote $\costsq{A}=\max\{\costsq{\row{A}},\costsq{A(1,\cdot)},\dots,\costsq{A(m,\cdot)}\}$. When $\varepsilon=0$, we simply write $\SQ{}(A)$.
\end{definition}
Note that from Condition (i) in Definition \ref{def:qs-access-mat} we automatically get $\Q{}(A)$, and from Condition~(ii)  we automatically get the norm $\fnorm{A}$.

We finally define oversampling access to a matrix, in a way similar to how oversampling access to a vector is defined. The case $\varepsilon=0$ corresponds to the model considered in \cite{Chia+JACM22}.

\begin{definition}\label{def:qos-access-mat}
For a nonzero matrix $A\in \Comp^{m\times n}$ and parameters $\varepsilon\ge 0$ and $\phi\ge 1$, we say that we have $(\varepsilon,\phi)$-approximate oversampling-and-query access to $A$, which we write $\OSQ{\varepsilon,\phi}(A)$, if the following two conditions are satisfied:
\begin{itemize}
\item[(i)]
we have $\Q{}(A)$;
\item[(ii)]
we have $\SQ{\varepsilon}(\bar A)$ for a nonzero matrix $\bar A\in\Comp^{m\times n}$ such that $\abs{\bar A(i,j)}\ge \abs{A(i,j)}$ for all $(i,j)\in\set{1}{m}\times \set{1}{n}$ and $\fnorm{\bar A}=\phi\fnorm{A}$.
\end{itemize} 
We denote $\costosq{A}=\max\{\costq{A},\costsq{\bar A}\}$.
\end{definition}

\section{Warming-Up: Robust Estimation of the Inner Product}\label{sec:warmingup}
As a warming-up, in this section we describe a special case illustrating one of the main ideas of this paper: how to estimate the inner product when given approximate sampling-and-query access to one vector and query access to the other. Here is the statement of the result:

\begin{proposition}\label{th:IP}
For any $\varepsilon\in[0,1]$, assume that  
\begin{itemize}
\item 
we have $\SQ{\varepsilon}(u)$ for a nonzero vector $u\in\Comp^n$;
\item 
we have $\Q{}(v)$ for a nonzero vector $v\in\Comp^n$ and know a value $\estnorm{v}$ such that $\estnorm{v}\ge\norm{v}$.
\end{itemize}
For any $\delta\in(0,1)$ and any $\xi>0$, we can output an estimator $\alpha$ such that 
\[
\abs{\alpha-(u,v)}\le(2\sqrt{2\varepsilon}+\xi)\norm{u}\estnorm{v}
\]
holds with probability at least $1-\delta$ at cost $O\left(\frac{\log(1/\delta)}{\xi^2}(\costsq{u}+\costq{v})\right)$.
\end{proposition}

Proposition \ref{th:IP} already generalizes the inner product estimation algorithm from \cite[Proposition 4.2]{TangSTOC19}, which worked only for the case of perfect sampling-and-query access (i.e., $\varepsilon=0$). The key idea to prove Proposition \ref{th:IP} is to partition the coordinates into two carefully-chosen sets (the sets $\Gamma$ and $\set{1}{n}\setminus \Gamma$ in our proof) and define the estimator differently on these two sets. This idea will also be used to show the matrix multiplication algorithm of Theorem \ref{th:MM1}. Note that by applying Theorem \ref{th:MM1}, we will give in Section \ref{sub:IP} another inner product estimation algorithm (see Proposition~\ref{prop:IP}), which is more general than Proposition \ref{th:IP} since it applies even to the case of approximate oversampling-and-query access.

\begin{proof}[Proof of Proposition \ref{th:IP}]
Let us define the set
\[
\Gamma=\left\{
i\in\set{1}{n}\:|\: \abs{v(i)}\ge \frac{\estnorm{v}}{\sqrt{2\varepsilon}\norm{u}}\abs{u(i)}
\right\}.
\]
Note that $u(i)\neq 0$ and
\[
\frac{\abs{v(i)}}{\abs{u(i)}}\le \frac{\estnorm{v}}{\sqrt{2\varepsilon}\norm{u}}
\]
hold for all $i\notin \Gamma$.
Let $X$ be the random variable representing the output of the following process: sample an index $i\in\set{1}{n}$ from the probability distribution $\tilde p_u$ and output
\[
\begin{cases}
0&\textrm{ if } i\in \Gamma,\\
\frac{v(i)^\ast\norm{u}^2}{u(i)^\ast}& \textrm{ if }i\notin \Gamma.
\end{cases}
\]
We have
\begin{align*}
    \abs{\ex{X} - (u,v)}
    &=
    \abs{\sum_{i\notin \Gamma}\frac{v(i)^\ast\norm{u}^2}{u(i)^\ast} \tilde p_u(i) - \sum_{i=1}^n u(i)v(i)^\ast}\\
    &\le
    \abs{\sum_{i\notin \Gamma}\frac{v(i)^\ast\norm{u}^2}{u(i)^\ast} \tilde p_u(i)-\sum_{i\notin \Gamma}\frac{v(i)^\ast\norm{u}^2}{u(i)^\ast} p_u(i)}
+
\abs{\sum_{i\notin \Gamma}\frac{v(i)^\ast\norm{u}^2}{u(i)^\ast} p_u(i)-\sum_{i=1}^n u(i)v(i)^\ast}
\\
     &\le
\frac{\norm{u}^2\estnorm{v}}{\sqrt{2\varepsilon}\norm{u}}
\sum_{i\notin \Gamma}\abs{\tilde p_u(i)- p_u(i)}
+
      \abs{\sum_{i\in \Gamma}u(i)v(i)^\ast}\\
      &\le
      \frac{\norm{u}\estnorm{v}}{\sqrt{2\varepsilon}}2\varepsilon
      +
      \sqrt{
      \left(
      \sum_{i\in \Gamma}
      \abs{u(i)}^2
     \right)
      \left(
      \sum_{i\in \Gamma}
      \abs{v(i)}^2
      \right)
      }\\
      &\le
      \frac{\norm{u}\estnorm{v}}{\sqrt{2\varepsilon}}2\varepsilon
	+
      \sqrt{
      \left(
      \sum_{i\in \Gamma}
	\frac{2\varepsilon\norm{u}^2}{\estnorm{v}^2}
      \abs{v(i)}^2
      \right)
      \left(
      \sum_{i\in \Gamma}
      \abs{v(i)}^2
     \right)
      }\\
    &\le
\sqrt{2\varepsilon} \norm{u}\estnorm{v}
       + 
       \sqrt{2\varepsilon}\norm{u}\norm{v}\\
    &=
    2\sqrt{2\varepsilon} \norm{u}\estnorm{v}.
\end{align*}

We now compute the variance:
\begin{align*}
    \var{X}&=\ex{\abs{X}^2}-\abs{\ex{X}}^2\\
    &\le \ex{\abs{X}^2}\\
    &=\sum_{i\notin \Gamma} \frac{\abs{v(i)}^2\norm{u}^4}{\abs{u(i)}^2} \tilde p_u(i)\\
    &=\sum_{i\notin \Gamma}  \frac{\abs{v(i)}^2\norm{u}^4}{\abs{u(i)}^2}(\tilde p_u(i)-p_u(i))+
\sum_{i\notin \Gamma}  \frac{\abs{v(i)}^2\norm{u}^4}{\abs{u(i)}^2}p_u(i)\\
    &=\sum_{i\notin \Gamma}  \frac{\abs{v(i)}^2\norm{u}^4}{\abs{u(i)}^2}(\tilde p_u(i)-p_u(i))+
\sum_{i\notin \Gamma}  \abs{v(i)}^2\frac{\norm{u}^4}{\norm{u}^2}\\
    &\le\frac{\norm{u}^2\estnorm{v}^2}{2\varepsilon}2\stat{\tilde p_u-p_u}+\norm{u}^2\norm{v}^2\\
    &\le 2\norm{u}^2\estnorm{v}^2.
\end{align*}

Finally, we explain how to compute our estimate of the inner product $(u,v)$.
Take $r=\ceil{16/\xi^2}$ independent samples from $X$ and denote the corresponding random variables as $X_1,\ldots,X_r$. Consider the random variable $Z=\frac{X_1+\cdots X_r}{r}$. By Chebyshev's inequality we have
\[
\Prob{\abs{Z-E[Z]}>\frac{\xi}{\sqrt{2}}\norm{u}\estnorm{v}}\le \frac{2\var{Z}}{\xi^2\norm{u}^2\estnorm{v}^2}
=
\frac{2 \var{X}}{r\xi^2 \norm{u}^2\estnorm{v}^2}
\le
\frac{1}{4}.
\]
We then use Lemma \ref{lemma:powering} to compute an estimate $\alpha$ such that 
\[
\Prob{\abs{\alpha-E[Z]}\le\xi\norm{u}\estnorm{v}}\ge 1-\delta.
\]
By the triangle inequality we conclude that
\[
\Prob{\abs{\alpha-(u,v)}\le \left(2\sqrt{2\varepsilon}+\xi\right)\norm{u}\estnorm{v}}\le 1-\delta.
\]
The overall complexity is 
\[
O\left(r\log(1/\delta)(\costsq{u}+\costq{v})\right)
=O\left(\frac{\log(1/\delta)}{\xi^2}(\costsq{u}+\costq{v})\right),
\]
as claimed.
\end{proof}


\section{Matrix Multiplication using Importance Matrix Sketches}
In this section we introduce several notions of approximate matrix sketching and show how to use them to approximate matrix products and singular value transformations. 
\subsection{Approximate importance matrix sketches  and matrix multiplication}\label{sub:MM1}
In this subsection we define the basic notion of approximate matrix sketch needed for this work and show how to use it to approximate the matrix product.

Here is the first key definition.
\begin{definition}\label{def:S}
Let $r,m$ be two positive integers such that $r\le m$, and $p, \tilde p$ be two probability distributions over $\set{1}{m}$. The $r\times m$ matrix sampled according to $(p,\tilde p)$ is the matrix $S\in\Real^{r\times m}$ obtained by the following process: for each $i\in\{1,\ldots,r\}$, choose an index $s_i\in\set{1}{m}$ by sampling from the distribution $\tilde p$, and set the $i$-th row of $S$ as 
\[
S(i,\cdot)=
\begin{cases}
\frac{e_{s_i}}{\sqrt{r p(s_i)}}&\textrm{if }p(s_i)> 0,\\
{\mathbf 0}&\textrm{if }p(s_i)=0,
\end{cases}
\]
where $e_{s_i}\in\Comp^{1\times n}$ is the row-vector that has coordinate 1 in the $s_i$-th position and 0 elsewhere, and ${\mathbf 0}$ denotes the all-zero row-vector.\footnote{For our purpose, the definition of $S(i,\cdot)$ for the case $p(s_i)=0$ is actually arbitrary. Indeed, this case will never occur in our calculations since we will work on the support of the distribution $p$.} We call the list $((s_1,\alpha_1),\ldots,(s_r,\alpha_r))$, where $\alpha_i=1/\sqrt{r p(s_i)}$ if $p(s_i)>0$ and $\alpha_i=0$ otherwise, the standard description of $S$.
\end{definition}

We now define approximate importance matrix sketches.
\begin{definition}\label{def:imp}
Given a nonzero matrix $A\in \Comp^{m\times n}$, a positive integer $r\le m$ and two parameters $\varepsilon\in[0,1]$ and $\phi\ge 1$, an $(r, \varepsilon,\phi)$-approximate importance matrix sketch of $A$ is a matrix $S\in\Real^{r\times m}$ that is sampled according to $(p,\tilde p)$ for some probability distributions $p,\tilde p$ satisfying the following two conditions:
\begin{itemize}
\item
$p (i)\ge \frac{1}{\phi}p_{\row{A}}(i)$ for all $i\in\set{1}{m}$;
\item
$\stat{\tilde p-p}\le \varepsilon$.
\end{itemize}
\end{definition}

This notion of approximate importance matrix sketch is motivated by the following definition.
\begin{definition}\label{def:impa}
Consider a nonzero matrix $A\in \Comp^{m\times n}$ with $\OSQ{\varepsilon,\phi}(A)$, and an integer $r\le  m$. Let $\bar A$ denote the matrix from Definition \ref{def:qos-access-mat} and $\tilde p_{\brow{A}}$ denote the distribution such that $\stat{\tilde p_{\brow{A}}-p_{\brow{A}}}\le \varepsilon$ corresponding to $\SQ{\varepsilon,\phi}(\bar A)$.
The $r\times m$ matrix sampled according to $(p_{\brow{A}},\tilde p_{\brow{A}})$
is called the $(r,\varepsilon,\phi)$-approximate importance matrix sketch of $A$ associated with $\OSQ{\varepsilon,\phi}(A)$.
\end{definition}
It is straightforward to check that the matrix sketch of Definition \ref{def:impa} is an $(r,\varepsilon,\phi)$-approximate importance matrix sketch of~$A$ according to Definition \ref{def:imp}.\footnote{We have 
\[
p_{\brow{A}}(i)=\frac{\norm{\bar A(i,\cdot)}^2}{\fnorm{\bar A}^2}\ge \frac{\norm{A(i,\cdot)}^2}{\phi\fnorm{A}^2}=\frac{1}{\phi}p_{\row{A}}(i)
\]
for all $i\in\set{1}{m}$.} Definition \ref{def:impa} is the special case of Definition \ref{def:imp} that is relevant for algorithmic applications when we can access a matrix via approximate oversampling-and-query access. Since all the results of this subsection work for Definition \ref{def:imp}, we nevertheless present our results using this more general (and mathematically more elegant) definition.

The following lemma, which extends Lemma 4.2 in \cite{Chia+JACM22} to the setting of approximate importance matrix sketches, shows that if $S$ is an $(r, \varepsilon,\phi)$-approximate importance matrix sketch of $A$ with~$r$ large enough and $\phi$ small enough, then $\fnorm{SA}$ is a good approximation of $\fnorm{A}$.
\begin{lemma}\label{lemma:norm}
Given a matrix $A\in \Comp^{m\times n}$, a positive integer $r\le m$ and two parameters $\varepsilon\in[0,1]$ and $\phi\ge 1$, consider an $(r, \varepsilon,\phi)$-approximate importance matrix sketch $S$ of $A$. We have 
\begin{equation}\label{eq:up}
\norm{[SA](i,\cdot)}^2\le \frac{\phi\fnorm{A}^2}{r}
\end{equation}
for all $i\in\set{1}{r}$. Moreover,
for any $\delta\in(0,1]$ we have
\[
\Prob{\abs{\fnorm{SA}^2-\fnorm{A}^2}\le \left(2\varepsilon+\sqrt{\frac{\ln{(2/\delta)}}{2r}}\right)\phi\fnorm{A}^2}\ge 1-\delta.
\]
\end{lemma}
\begin{proof}
For any $i\in\set{1}{r}$, we have
\[
\norm{[SA](i,\cdot)}^2=
\frac{\norm{A(s_i,\cdot)}^2}{rp(s_i)}\le 
\frac{\phi \norm{A(s_i,\cdot)}^2}{rp_{\row{A}}(s_i)}= 
\frac{\phi\fnorm{A}^2}{r}
\]
if $p(s_i)>0$. If  $p(s_i)=0$ we have $\norm{[SA](i,\cdot)}=0$, which trivially satisfies the  inequality.

Let us now prove the second part. We have
\begin{align*}
\abs{
\ex{\norm{[SA](i,\cdot)}^2}
-
\frac{\fnorm{A}^2}{r}
}
&=
\abs{
\sum_{s\in\supp{p}} \tilde p(s)
\frac{\norm{A(s,\cdot)}^2}
{rp(s)}
-
\frac{\sum_{s\in\supp{p}}\norm{A(s,\cdot)}^2}
{r}
}\\
&=
\abs{
\sum_{s\in\supp{p}} (\tilde p(s)-p(s))
\frac{\norm{A(s,\cdot)}^2}
{rp(s)}
}\\
&\le
\sum_{s\in\supp{p}}
\abs{
\tilde p(s)-
 p(s)}
\frac{\norm{A(s,\cdot)}^2}
{rp(s)}\\
&\le
\sum_{s\in\supp{p}}
\abs{
\tilde p(s)-
 p(s)}
\frac{\phi\fnorm{A}^2}{r}\\
&\le 
\frac{2\varepsilon\phi\fnorm{A}^2}{r},
\end{align*}
where we used again $p(s)\ge \frac{1}{\phi}p_{\row{A}}(s)$ to derive the second inequality.
We thus have
\[
\abs{
\ex{\norm{SA}^2}
-
\fnorm{A}^2
}
\le 
2\varepsilon\phi\fnorm{A}^2.
\]
Observe that the quantity $\norm{SA}^2$ is a sum of $r$ independent random variables, each in the interval $[0, \phi\fnorm{A}^2/r]$ (the upper bound follows from Equation (\ref{eq:up})). From Hoeffding's inequality, we conclude that
\[
\Prob{\abs{\fnorm{SA}^2-\ex{\fnorm{SA}^2}}\ge \sqrt{\frac{\ln{(2/\delta)}}{2r}}\phi\fnorm{A}^2}\le \delta.
\]
From the triangle inequality we get
\[
\Prob{\abs{\fnorm{SA}^2-\fnorm{A}^2}\le \left(2\varepsilon+\sqrt{\frac{\ln{(2/\delta)}}{2r}}\right)\phi\fnorm{A}^2}\ge 1-\delta,
\]
as claimed.
\end{proof}

Here is the main result of this subsection, which shows how to approximate a matrix product $X^\dagger Y$ by using an approximate matrix sketch of $X$.
\addtocounter{theorem}{-3}
\begin{theorem}[formal statement]
Given two nonzero matrices $X\in \Comp^{m\times n}$ and $Y\in\Comp^{m\times n'}$, a positive integer $r\le m$ and two parameters $\varepsilon\in[0,1]$ and $\phi\ge 1$, define the set 
\[
\Gamma=\left\{
j\in\set{1}{m}\:|\: \norm{Y(j,\cdot)}\ge \frac{\fnorm{Y}}{\sqrt{2\phi\varepsilon}\fnorm{X}}\norm{X(j,\cdot)}
\right\}.
\]
Consider an $(r, \varepsilon,\phi)$-approximate importance matrix sketch $\Sigma$ of $X$. Let $\bar\Sigma\in\Real^{r\times m}$ be the matrix obtained from $\Sigma$ by the following process: for each $i\in\set{1}{r}$, if $i\in \Gamma$ then replace the $i$-th row of $\Sigma$ by an all-zero row.
 Then for any $\delta>0$,
\[
\Prob{
\fnorm{X^\dagger \bar\Sigma^\dagger \bar\Sigma Y-X^\dagger Y}\le \left(2\sqrt{2\varepsilon}+\sqrt{\frac{2}{r\delta}}\right) \sqrt{\phi}\fnorm{X}\fnorm{Y}}\ge 1-\delta.
\]
\end{theorem}
\addtocounter{theorem}{3}
The informal version of Theorem \ref{th:MM1} given in the introduction corresponds to the case $\phi=1$ (the claim that $\bar\Sigma$ can be efficiently constructed follows from Lemma \ref{lemma:createsk} in Section \ref{sub:app1}).
\begin{proof}[Proof of Theorem \ref{th:MM1}]
At a high level, our proof applies a strategy similar to the strategy used in Lemma~4 in~\cite{Drineas+SICOMP06} (and also used in~\cite{Chia+JACM22}). In particular, we compute the expectation of the sketch, its variance and conclude using Chebyshev's inequality. Significant additional care is needed in the analysis since our we are not working with $\Sigma$ but with $\bar \Sigma$, and those sketches are not obtained by sampling from the perfect distribution but only from a distribution close to it, which leads to a biased estimator instead of an unbiased estimator in those prior works.

Let $p$ and $\tilde p$ denote the probability distributions from Definition \ref{def:imp}. 
Observe that for any $j\notin \Gamma$ we necessarily have $\norm{X(j,\cdot)}>0$, and thus $p(j)>0$ since $p(j)\ge \frac{1}{\phi}p_{\row{X}}(j)$. 

We first show that the expectation of $X^\dagger \bar\Sigma^\dagger \bar \Sigma Y$ is close to $X^\dagger Y$:
\begin{align*}
\fnorm{
\ex{X^\dagger \bar\Sigma^\dagger \bar\Sigma Y}-X^\dagger Y
}&=
\fnorm{
\ex{\sum_{i=1}^r [(\bar\Sigma X)^\dagger](\cdot,i) [\bar\Sigma Y](i,\cdot)}
-\sum_{j=1}^mX^\dagger(\cdot,j) Y(j,\cdot)
}\\
&=
\fnorm{
\ex{\sum_{i=1}^r ([\bar\Sigma X](i,\cdot))^\dagger [\bar\Sigma Y](i,\cdot)}
-\sum_{j=1}^mX(j,\cdot)^\dagger Y(j,\cdot)
}\\
&=
\fnorm{
r\ex{([\bar\Sigma X](1,\cdot))^\dagger [\bar\Sigma Y](1,\cdot)}
-\sum_{j=1}^mX(j,\cdot)^\dagger Y(j,\cdot)
}\\
&=
\fnorm{
r\sum_{j\notin \Gamma} \tilde p(j)\frac{X(j,\cdot)^\dagger Y(j,\cdot)}{r p(j)}
-\sum_{j=1}^mX(j,\cdot)^\dagger Y(j,\cdot)
}\\
&=
\fnorm{
\sum_{j\notin \Gamma}\left(\tilde p(j)-p(j)\right)
\frac{X(j,\cdot)^\dagger Y(j,\cdot)}{p(j)}
-
\sum_{j\in \Gamma}X(j,\cdot)^\dagger Y(j,\cdot)
}\\
&\le
\sum_{j\notin \Gamma}\abs{\tilde p(j)-p(j)} \frac{\phi\fnorm{X}^2\norm{Y(j,\cdot)}}{\norm{X(j,\cdot)}}+
\sum_{j\in \Gamma}\norm{X(j,\cdot)}\norm{Y(j,\cdot)}
\\
&\le
\frac{\phi\fnorm{X}\fnorm{Y}}{\sqrt{2\phi\varepsilon}}2\varepsilon+\sum_{j\in \Gamma}\frac{\sqrt{2\phi\varepsilon}\fnorm{X}}{\fnorm{Y}}\norm{Y(j,\cdot)}^2\\
&\le 2\sqrt{2\phi\varepsilon}\fnorm{X}\fnorm{Y}.
\end{align*}

Let us now compute the variance:
\begin{align*}
\var{X^\dagger \bar \Sigma^\dagger \bar \Sigma Y}
&\le \ex{\fnorm{X^\dagger \bar \Sigma^\dagger \bar \Sigma Y}^2}\\
&=  \ex{\sum_{i=1}^n\sum_{j=1}^{n'} \abs{[X^\dagger \bar \Sigma^\dagger \bar \Sigma Y](i,j)}^2}\\
&=  \ex{\sum_{i=1}^n\sum_{j=1}^{n'}\sum_{k=1}^r \abs{[X^\dagger \bar\Sigma^\dagger](i,k) [\bar\Sigma Y](k,j)}^2}\\
&=  r\ex{\sum_{i=1}^n\sum_{j=1}^{n'}\abs{[\bar\Sigma X](1,i) [\bar\Sigma Y](1,j)}^2}\\
&=  r\ex{\norm{[\bar\Sigma X](1,\cdot)}^2\norm{[\bar \Sigma Y](1,\cdot)}^2}\\
&=  r\sum_{j\notin \Gamma} \tilde p(j)\frac{\norm{X(j,\cdot)}^2\norm{Y(j,\cdot)}^2}{(rp(j))^2}\\
&=  r\sum_{j\notin \Gamma} (\tilde p(j)-p(j))\frac{\norm{X(j,\cdot)}^2\norm{Y(j,\cdot)}^2}{(rp(j))^2}+
 r\sum_{j\notin \Gamma} p(j)\frac{\norm{X(j,\cdot)}^2\norm{Y(j,\cdot)}^2}{(rp(j))^2}
\\
&\le  
r\sum_{j\notin \Gamma} (\tilde p(j)-p(j))\frac{\norm{X(j,\cdot)}^4\fnorm{Y}^2}{2\phi\varepsilon(rp(j))^2\fnorm{X}^2}+
\frac{\phi\fnorm{X}^2\fnorm{Y}^2}{r}.
\end{align*}
Using the inequality $p(j)\ge \frac{1}{\phi}p_{\row{X}}(j)$ and $p_{\row{X}}(j)=\frac{\norm{X(j,\cdot)}^2}{\fnorm{X}^2}$, we get
\[
\var{X^\dagger \bar \Sigma^\dagger \bar \Sigma Y}\le
\sum_{j\notin \Gamma} (\tilde p(j)-p(j))\frac{\phi\fnorm{X}^2\fnorm{Y}^2}{2\varepsilon r}+
\frac{\phi\fnorm{X}^2\fnorm{Y}^2}{r}
\le
\frac{2\phi\fnorm{X}^2\fnorm{Y}^2}{r}.
\]

Using Chebyshev's inequality we obtain
\[
\Prob{\fnorm{X^\dagger \bar\Sigma^\dagger \bar\Sigma Y -\ex{X^\dagger \bar\Sigma^\dagger \bar\Sigma Y}}\ge 
\sqrt{\frac{2\phi}{r\delta}}
\fnorm{X}\fnorm{Y}}\le \frac{\var{X^\dagger \bar\Sigma^\dagger \bar\Sigma Y}}{\frac{2\phi}{r\delta}\fnorm{X}^2\fnorm{Y}^2}=\delta.
\]
We conclude the proof by using the triangle inequality:
\[
\Prob{\fnorm{X^\dagger \bar\Sigma^\dagger \bar\Sigma Y -X^\dagger Y}\le \left(2\sqrt{2\varepsilon}+\sqrt{\frac{2}{r\delta}}\right)\sqrt{\phi}\fnorm{X}\fnorm{Y}}\ge
1-\delta,
\]
as claimed.
\end{proof}

\subsection{Approximate joint importance matrix sketches and matrix multiplication}\label{sub:MM}
In order to compute an approximation of the product of two matrices $A$ and $B$ with dependence on the error parameter $\delta$ better than in Theorem \ref{th:MM1}, we will need a slightly different notion of sketch that we call approximate \emph{joint} importance matrix sketch and define as follows. 
\begin{definition}\label{definition:mixed}
Given two nonzero matrices $A\in \Comp^{m\times n}$ and $B\in \Comp^{m\times n'}$, a positive integer $r\le m$ and three parameters $\varepsilon\in(0,1)$ and $\phi,\phi'\ge 1$, an $(r, \varepsilon,\phi,\phi')$-approximate joint importance matrix sketch of $A$ and~$B$ is a matrix~$S\in \Real^{r\times m}$ that is sampled according to $(p,\tilde p)$ for some probability distributions $p,\tilde p$ satisfying the following conditions:
\begin{itemize}
\item[(i)]
$p (i)\ge \frac{1}{2}(\frac{1}{\phi}p_{\row{A}}(i)+\frac{1}{\phi'}p_{\row{B}}(i))$ for all $i\in\set{1}{m}$;
\item[(ii)]
$\stat{\tilde p-p}\le \varepsilon$.
\end{itemize}
\end{definition}

This notion of approximate joint importance matrix sketch is motivated by the following definition.
\begin{definition}\label{definition:mixeda}
Consider two nonzero matrices $A\in \Comp^{m\times n}$, $B\in \Comp^{m\times n'}$ with $\OSQ{\varepsilon,\phi}(A)$, $\OSQ{\varepsilon,\phi'}(B)$ and an integer $r\le m$. Let $\bar A,\bar B$ denote the matrices from Definition \ref{def:qos-access-mat} and $\tilde p_{\brow{A}}, \tilde p_{\brow{B}}$ denote the distributions such that $\stat{\tilde p_{\brow{A}}-p_{\brow{A}}}\le \varepsilon$ and $\stat{\tilde p_{\brow{B}}-p_{\brow{B}}}\le \varepsilon$ corresponding to $\SQ{\varepsilon,\phi}(\bar A)$ and $\SQ{\varepsilon,\phi'}(\bar B)$, respectively. The $r\times m$ matrix sampled according to $(\frac{1}{2}(p_{\brow{A}}+p_{\brow{B}}),\frac{1}{2}(\tilde p_{\brow{A}}+\tilde p_{\brow{B}}))$ is called the $(r,\varepsilon,\phi,\phi')$-approximate joint importance matrix sketch of $A$ and $B$ associated with $\OSQ{\varepsilon,\phi}(A)$ and $\OSQ{\varepsilon,\phi'}(B)$.
\end{definition}
It is again straightforward to check that the matrix sketch of Definition \ref{definition:mixeda} is an $(r,\varepsilon,\phi,\phi')$-approximate joint importance matrix sketch of~$A$ and $B$ according to Definition \ref{definition:mixed}.\footnote{We have 
\[
\frac{1}{2}\left(p_{\brow{A}}(i)+p_{\brow{B}}(i)\right)
=\frac{1}{2}\left(\frac{\norm{\bar A(i,\cdot)}^2}{\fnorm{\bar A}^2}+\frac{\norm{\bar B(i,\cdot)}^2}{\fnorm{\bar B}^2}\right)
\ge \frac{1}{2}\left(\frac{\norm{A(i,\cdot)}^2}{\phi\fnorm{A}^2}+\frac{\norm{B(i,\cdot)}^2}{\phi'\fnorm{B}^2}\right)=
\frac{1}{2}\left(\frac{1}{\phi}p_{\row{A}}(i)+\frac{1}{\phi'}p_{\row{B}}(i)\right)
\]
for all $i\in\set{1}{m}$.}  Definition~\ref{definition:mixeda} is the special case of Definition \ref{definition:mixed} that is relevant for algorithmic applications when we can access the two matrices via approximate oversampling-and-query access. Since all the results of this subsection work for Definition \ref{definition:mixed}, we nevertheless again present our results using this more general (and mathematically more elegant) definition.

The following easy proposition, similar to the statement used in the proof of Lemma~4.6 in \cite{Chia+JACM22} for perfect oversampling-and-query access, shows that a joint importance matrix sketch of $A$ and $B$ is an importance matrix sketch of both $A$ and $B$ (with a slightly worse oversampling parameter), and shows Inequality (\ref{ineq:cruc}), which will be crucial for the matrix multiplication algorithm of Theorem~\ref{th:MM}. 
\begin{proposition}
An $(r, \varepsilon,\phi,\phi')$-approximate joint importance matrix sketch of $A\in \Comp^{m\times n}$ and $B\in \Comp^{m\times n'}$ is an $(r, \varepsilon,2\phi)$-approximate importance matrix sketch of $A$ and an $(r, \varepsilon,2\phi')$-approximate importance matrix sketch of $B$. Moreover, for any probability distribution $p$ satisfying Condition (i) of Definition \ref{definition:mixed} we have 
\begin{equation}\label{ineq:cruc}
p(i)\ge \frac{\norm{A(i,\cdot)}\norm{B(i,\cdot)}}{\sqrt{\phi\phi'} \fnorm{A}\fnorm{B}}
\end{equation}
for all $i\in\set{1}{m}$.
\end{proposition}
\begin{proof}
The first part is trivial since we have $p (i)\ge \frac{1}{2\phi}p_{\row{A}}(i)$ and $p (i)\ge \frac{1}{2\phi'}p_{\row{B}}(i)$ for all $i\in\set{1}{m}$.
The second part follows from the inequality of arithmetic and geometric means: for each $i\in\set{1}{m}$ we get
\begin{align*}
p(i)&
\ge
\frac{1}{2}\left(\frac{1}{\phi}p_{\row{A}}(i)+\frac{1}{\phi'}p_{\row{B}}(i)\right)
\ge\
\sqrt{\frac{p_{\row{A}}(i)p_{\row{B}}(i)}{\phi\phi'}}
= \frac{\norm{A (i,\cdot)}\norm{B (i,\cdot)}}{\sqrt{\phi\phi'} \fnorm{A}\fnorm{B}},
\end{align*}
as claimed.
\end{proof}

We can now state the main theorem of this subsection: the product of two matrices  can be approximated well using a joint importance matrix sketch.
\addtocounter{theorem}{-3}
\begin{theorem}[formal statement]
Given two matrices $X\in \Comp^{m\times n}$ and $Y\in\Comp^{m\times n'}$, a positive integer $r\le m$ and three parameters $\varepsilon\in(0,1)$ and $\phi,\phi'\ge 1$, consider an $(r, \varepsilon,\phi,\phi')$-approximate joint importance matrix sketch $\Sigma$ of $X$ and $Y$. Then for any $\delta\in(0,1]$,
\[
\Prob{
\fnorm{X^\dagger \Sigma^\dagger \Sigma Y-X^\dagger Y}\le \sqrt{\phi\phi'}\left(2\varepsilon+\sqrt{\frac{7\ln(2/\delta)}{r}}\right)\fnorm{X}\fnorm{Y}}\ge 1-\delta.
\]
\end{theorem}
\addtocounter{theorem}{3}
Note that the dependence on $\delta$ in Theorem \ref{th:MM} is exponentially better than in Theorem~\ref{th:MM1}.
The informal version of Theorem \ref{th:MM} given in the introduction corresponds to the case $\phi=\phi'=1$ (the claim that $\Sigma$ can be efficiently constructed follows from Lemma \ref{lemma:createsk} in Section \ref{sub:app1}).
\begin{proof}[Proof of Theorem \ref{th:MM}]
At a high level, our proof follows the strategy from \cite{Drineas+SICOMP06} (and also used in~\cite{Chia+JACM22}) for the case of perfect importance matrix sketches. In particular, we use MacDiarmid's bounded difference inequality. Significant additional care is needed in the analysis to control the impact of not being able to sample from the perfect distribution.

Let $p$ and $\tilde p$ denote the probability distributions from Definition \ref{definition:mixed}. Observe that $p(j)=0$ implies that $p_{\row{X}}(j)=p_{\row{Y}}(j)=0$, which implies $\norm{X(j,\cdot)}=\norm{Y(j,\cdot)}=0$. 

We first show that the expectation of $X^\dagger \Sigma^\dagger \Sigma Y$ is close to $X^\dagger Y$:
\begin{align*}
\fnorm{
\ex{X^\dagger \Sigma^\dagger \Sigma Y}-X^\dagger Y
}&=
\fnorm{
\ex{\sum_{i=1}^r [(\Sigma X)^\dagger](\cdot,i) [\Sigma Y](i,\cdot)}
-\sum_{j=1}^mX^\dagger (\cdot,j) Y(j,\cdot)
}\\
&=
\fnorm{
\ex{\sum_{i=1}^r ([\Sigma X](i,\cdot))^\dagger [\Sigma Y](i,\cdot)}
-\sum_{j\in\supp{p}}X(j,\cdot)^\dagger Y(j,\cdot)
}\\
&=
\fnorm{
r\ex{([\Sigma X](1,\cdot))^\dagger [\Sigma Y](1,\cdot)}
-\sum_{j\in\supp{p}}X(j,\cdot)^\dagger Y(j,\cdot)
}\\
&=
\fnorm{
r\sum_{j\in\supp{p}} \tilde p(j)\frac{X(j,\cdot)^\dagger Y(j,\cdot)}{r p(j)}
-\sum_{j\in\supp{p}}X(j,\cdot)^\dagger Y(j,\cdot)
}\\
&=
\fnorm{
\sum_{j\in\supp{p}}\left(\tilde p(j)-p(j)\right)
\frac{X(j,\cdot)^\dagger Y(j,\cdot)}{p(j)}
}\\
&\le
\sum_{j\in\supp{p}}\abs{\tilde p(j)-p(j)} \sqrt{\phi\phi'}\fnorm{X}\fnorm{Y}\\
&\le
2\varepsilon\sqrt{\phi\phi'}\fnorm{X}\fnorm{Y},
\end{align*}
where we used Inequality (\ref{ineq:cruc}) to derive the first inequality.

Consider the indices $s_1,\ldots,s_r\in\set{1}{m}$ chosen when constructing the matrix sketch $\Sigma$ (see Definition \ref{def:S}). Define the function 
\[
f(s_1,\ldots,s_r)=\fnorm{X^\dagger \Sigma^\dagger \Sigma Y-\ex{X^\dagger \Sigma^\dagger \Sigma Y}}.
\]
We have 
\begin{align*}
\ex{f}&=\ex{\fnorm{X^\dagger \Sigma^\dagger \Sigma Y-\ex{X^\dagger \Sigma^\dagger \Sigma Y}}}\\
&\le \sqrt{\ex{\fnorm{X^\dagger \Sigma^\dagger \Sigma Y-\ex{X^\dagger \Sigma^\dagger \Sigma Y}}^2}}\\
&\le  \sqrt{\sum_{i=1}^n\sum_{j=1}^{n'} \ex{\abs{[X^\dagger \Sigma^\dagger \Sigma Y](i,j)}^2}},
\end{align*}
where we used Jensen's inequality, and then properties of the variance for the second equality. We thus have:
\begin{align*}
\ex{f}&\le   \sqrt{\ex{\sum_{i=1}^n\sum_{j=1}^{n'}\sum_{k=1}^r \abs{[X^\dagger \Sigma^\dagger](i,k) [\Sigma Y](k,j)}^2}}\\
&=  \sqrt{r\ex{\sum_{i=1}^n\sum_{j=1}^{n'}\abs{[\Sigma X](1,i) [\Sigma Y](1,j)}^2}}\\
&=  \sqrt{r\ex{\norm{\Sigma X(1,\cdot)}^2\norm{\Sigma Y(1,\cdot)}^2}}\\
&= \sqrt{r\sum_{j\in\supp{p}} \tilde p(j)\frac{\norm{X(j,\cdot)}^2}{rp(j)}\frac{\norm{Y(j,\cdot)}^2}{rp(j)}}\\
&\le \sqrt{\frac{1}{r}\sum_{j\in\supp{p}} \tilde p(j)\phi\phi'\fnorm{X}^2\fnorm{Y}^2}\\
&\le 
\frac{\sqrt{\phi\phi'}\fnorm{X}\fnorm{Y}}{\sqrt{r}}.
\end{align*}

Let $\vec{i}=(i_1,\ldots,i_r)$ and $\vec{i'}=(i'_1,\ldots,i'_r)$ be two vectors of variables that differ in only one coordinate.
Using the reverse triangle inequality and Inequality (\ref{ineq:cruc}), we obtain 
\[
\abs{f(\vec{i})-f(\vec{i'})}\le 2 \max_{j\in\supp{p}} \fnorm{\frac{X(j,\cdot)^\dagger Y(j,\cdot)^\dagger}{rp(j)}}\le \frac{2\sqrt{\phi\phi'}\fnorm{X}\fnorm{Y}}{r}.
\]

By using MacDiarmid's bounded difference inequality (Lemma \ref{lemma:MacDiarmid}), we conclude that 
\[
\Prob{\abs{f-\ex{f}}\ge \sqrt{\frac{2\ln(2/\delta)\phi\phi'}{r}}\fnorm{X}\fnorm{Y}}\le \delta.
\]
In consequence the following inequalities hold with probability at least $1-\delta$:
\begin{align*}
\fnorm{X^\dagger \Sigma^\dagger \Sigma Y-\ex{X^\dagger \Sigma^\dagger \Sigma Y}}
&\le \ex{f}+\sqrt{\frac{2\ln(2/\delta)\phi\phi'}{r}}\fnorm{X}\fnorm{Y}\\
&\le \sqrt{\phi\phi'}\left(\frac{1}{\sqrt{r}}+\sqrt{\frac{2\ln(2/\delta)}{r}}\right)\fnorm{X}\fnorm{Y},
\end{align*}
and thus 
\begin{align*}
\fnorm{X^\dagger \Sigma^\dagger \Sigma Y-X^\dagger Y}&\le \sqrt{\phi\phi'}\left(2\varepsilon+\frac{1}{\sqrt{r}}+\sqrt{\frac{2\ln(2/\delta)}{r}}\right)\fnorm{X}\fnorm{Y}\\
&\le \sqrt{\phi\phi'}\left(2\varepsilon+\sqrt{\frac{7\ln(2/\delta)}{r}}\right)\fnorm{X}\fnorm{Y},
\end{align*}
as claimed, where we used the inequality $1+\sqrt{2\log(2/\delta)}\le \sqrt{7\log(2/\delta)}$.
\end{proof}

\subsection{Singular value transformation}\label{sub:SVT}
In this subsection we show how matrix multiplication via importance matrix sketches can be used to approximate the singular value transformation of a positive semidefinite Hermitian matrix of the form $A^\dagger A$.

We start with two definitions. 

\begin{definition}
Given a positive semidefinite Hermitian matrix $H\in \Comp^{n\times n}$, for any $\chi\ge 0$ we write
\[
\spec_{\chi}(H)=\left\{
x\in \Real_{\ge 0}\:|\: \abs{x-z}\le \chi \textrm{ for some }z\in \spec(H)
\right\}.
\]
\end{definition}

\begin{definition}
Given two parameters $L,\bar L\ge 0$ and a set $\Lambda\subseteq \Real_{\ge 0}$, a function $f\colon \Real_{\ge 0}\to \Comp$ is $(L,\bar L)$-smooth on $\Lambda$ if it satisfies the following conditions:
\begin{itemize}
\item
$f(0)=0$ and  the limit of the function $\frac{f(x)}{x}$ at $x=0$ exists (we write this limit c); 
\item
$f$ is $L$-Lipschitz over $\Lambda$ ;
\item
the function $\bar f\colon \Real_{\ge 0}\to \Comp$ defined as $\bar f(x)=\frac{f(x)}{x}$ for $x>0$ and $\bar f(0)=c$ is $\bar L$-Lipschitz over $\Lambda$.
\end{itemize}
\end{definition}

We can now state the main result of this subsection.
\addtocounter{theorem}{-3}
\begin{theorem}[formal statement]
Let $A\in\Comp^{m\times n}$ be a matrix and $f\colon\Real_{\ge 0}\to \Comp$ be an $(L,\bar L)$-smooth function on $\spec_{\chi}(A^\dagger A)$,  for some $L,\bar L, \chi> 0$. For any parameters $\delta\in(0,1]$ and $\gamma>0$ and any $\phi,\phi'\ge 1$, consider any $\varepsilon\ge 0$ and any integers $r\in\set{1}{m}$ and $c\in\set{1}{n}$ satisfying the following conditions:
\begin{align}
\varepsilon&\le \min\left(
\frac{1}{8\phi}, \frac{\chi}{24(\phi+\phi')\fnorm{A}^2},
\frac{\gamma}{24 L\phi \fnorm{A}^{2}},\frac{\gamma}{48 \bar L\phi'\fnorm{A}^{4}}\right),\label{cond1}\\
r&\ge \max\left(2\phi^2\ln(6/\delta),112\phi^2\fnorm{A}^4\left(
\frac{\ln(6/\delta)}{\chi^2}+\frac{\ln(6/\delta)L^2}{\gamma^2}
\right)\right)
,\label{cond2}\\
c&\ge 112\phi'^2\fnorm{A}^4\left(
\frac{\ln(6/\delta)}{\chi^2}+\frac{4\ln(6/\delta)\bar L^2\fnorm{A}^4}{\gamma^2}
\right).\label{cond3}
\end{align}
Let $S\in\Real^{r\times m}$ be an $(r,\varepsilon,\phi)$-approximate importance matrix sketch of $A$,  and $T\in\Real^{c\times n}$ be a $(c,\varepsilon,\phi')$-approximate importance matrix sketch of $(SA)^\dagger$.  Then for
$R=SA$ and $C=SAT^\dagger$, the inequality 
\[
\fnorm{R^\dagger \bar{f}(CC^\dagger)R -f(A^\dagger A)}\le \gamma
\]
holds with probability at least $1-\delta$.
\end{theorem}
\addtocounter{theorem}{3}
The informal version of Theorem \ref{th:SVT} given in the introduction corresponds to the case $\phi,L,\bar{L}=O(1)$ and $\chi=\infty$. The claim that $R$ and $C$ can be efficiently (implicitly) constructed follows from Lemmas \ref{lemma:createsk} and \ref{lemma:sample2} in Section \ref{sub:app1}.

Theorem \ref{th:SVT} is a generalization of Theorem 5.1 in \cite{Chia+JACM22}, which focused on perfect sampling-and-query (i.e., $\varepsilon=0$). As the proof in \cite{Chia+JACM22} combined several techniques for perfect importance matrix sketches, our proof combines essentially all the tools we developed so far for approximate importance matrix sketches.

\begin{proof}[Proof of Theorem \ref{th:SVT}]
First observe that from Lemma \ref{lemma:norm}, with probability at least $1-\delta/3$ we have
\[
\fnorm{R}^2\le 2 \fnorm{A}^2.
\]
We assume below that this inequality holds.


From Theorem \ref{th:MM} with $X=Y=A$ and $\Sigma=S$ (note that $S$ is obviously also an $(r,\varepsilon,\phi,\phi)$-approximate joint importance sketch of $A$ and $A$), we know that the inequality 
\[
\fnorm{R^\dagger R -A^\dagger A}\le 
\min\left(
\left(\frac{\chi}{4}+\frac{\chi}{12}\right)
\frac{\phi\fnorm{A}^2}{\phi\fnorm{A}^2},\left(\frac{\gamma}{4}+\frac{\gamma}{12}\right)\frac{\phi\fnorm{A}^2}{\phi L\fnorm{A}^2}\right)
\le  \min\left(\frac{\chi}{3},\frac{\gamma}{3L}\right)
\]
holds with probability at least $1-\delta/3$. By applying again Theorem \ref{th:MM} with $X=Y=R^\dagger$ and $\Sigma=T$, we know that the inequality 
\[
\fnorm{C C^\dagger- RR^\dagger}\le 
\min\left(
\left(\frac{\chi}{4}+\frac{\chi}{12}\right)
\frac{\phi' \fnorm{R^\dagger}^2}{\phi'\fnorm{A}^2},
\left(\frac{\gamma}{8}+\frac{\gamma}{24}\right) \frac{\phi'\fnorm{R}^2}{\phi'\bar L\fnorm{A}^4}
\right)
\le
\min\left(
\frac{2\chi}{3},
\frac{\gamma}{3\bar L\fnorm{A}^2}\right)
\]
holds with probability at least $1-\delta/3$ as well. We assume below that these inequalities hold.

Applying Lemma \ref{lemma:HW} twice we thus obtain
\begin{align}
\sqrt{\sum_{i=1}^{n}\abs{\lambda_i(R^\dagger R)-\lambda_i(A^\dagger A)}^2}&\le
\fnorm{R^\dagger R  -A^\dagger A}\le \frac{\chi}{3},\label{eq:full1}\\
\sqrt{\sum_{i=1}^{r}\abs{\lambda_i(CC^\dagger)-\lambda_i(R R^\dagger)}^2}&\le
\fnorm{CC^\dagger  -RR^\dagger }\le \frac{2\chi}{3}.\label{eq:full2}
\end{align}
These inequalities imply the inclusions\footnote{Observe that Inequalities (\ref{eq:full1}) and (\ref{eq:full2}) are actually significantly stronger bounds than these inclusions since they bound the \emph{sum} of the difference of eigenvalues. These inclusions will nevertheless be enough for our purpose.}
\begin{align*}
\spec(R^\dagger R)&\subseteq \spec_{\chi/3}(A^\dagger A),\\
\spec(CC^\dagger)&\subseteq \spec_{2\chi/3}(RR^\dagger)\subseteq\spec_{2\chi/3}(R^\dagger R) \subseteq \spec_{\chi}(A^\dagger A).
\end{align*}
From Lemma \ref{lemma:Lipschitz}, we thus obtain
\begin{align*}
\fnorm{f(R^\dagger R)-f(A^\dagger A)}&\le L \fnorm{R^\dagger R-A^\dagger A}\le\frac{\gamma}{3},\\
\fnorm{\bar f(CC^\dagger)-\bar f(RR^\dagger )}&\le \bar L \fnorm{C C^\dagger-R R^\dagger}\le \frac{\gamma}{3\fnorm{A}^2}.
\end{align*}

Putting all these inequalities together, we conclude that with probability at least $1-\delta$:
\begin{align*}
\fnorm{R^\dagger \bar{f}(CC^\dagger)R -f(A^\dagger A)}&\le 
\fnorm{R^\dagger \bar{f}(RR^\dagger)R -f(A^\dagger A)}
+
\fnorm{R^\dagger (\bar{f}(CC^\dagger)-\bar{f}(RR^\dagger))R}\\
&=
\fnorm{f(R^\dagger R) -f(A^\dagger A)}
+
\fnorm{R^\dagger (\bar{f}(CC^\dagger)-\bar{f}(RR^\dagger))R}\\
&\le 
\fnorm{f(R^\dagger R) -f(A^\dagger A)}
+
\norm{R}^2\fnorm{\bar{f}(CC^\dagger)-\bar{f}(RR^\dagger)}\\
&\le \frac{\gamma}{3} +\frac{\gamma\norm{R}^2}{3\fnorm{A}^2}\\
&\le \gamma,
\end{align*}
where the identity $R^\dagger \bar{f}(RR^\dagger)R=f(R^\dagger R)$ can be easily shown by considering the singular value decomposition of $R$.
\end{proof}
\section{Technical Tools}
In this section we present three important technical techniques that will be used in the applications we discuss in Section~\ref{sec:app}. In Section \ref{sub:over} we show how to implement randomized sampling from oversampling. In Section~\ref{sub:line} we show how the concept of oversampling-and-query access is closed under taking linear combination of rows.
Finally, in Section \ref{sub:app1} we show how to concretely construct importance matrix sketches, and prove several technical lemmas that will be used in Section \ref{sec:app}.

\subsection{From oversampling to sampling}\label{sub:over}
The goal of this subsection is to prove the following proposition that shows how to implement approximate (randomized) sampling-and-query access to a vector given approximate oversampling-and-query access to the same vector.

\begin{proposition}\label{prop:oversampling}
For a vector $u\in \Comp^{n}$, assume that we have $\OSQ{\varepsilon,\phi}(u)$ with parameters $\phi\ge 1$ and $\varepsilon\in[0,\frac{1}{2\phi})$ and assume that we know an upper bound $\phi_{\textup{max}}\ge \phi$. Then for any parameters $\delta,\eta\in(0,1]$, we can implement $\RSQ{\varepsilon',\delta,\eta}(u)$ with $\varepsilon'\le 3\varepsilon\phi$ at cost 
\[
O\left(\frac{\phi_{\textup{max}}\log(1/\delta)}{\eta^2}\costosq{u}\right).
\]
\end{proposition}

This proposition is a generalization of Lemma 3.5 in \cite{Chia+JACM22} and Proposition 4.3 in \cite{TangSTOC19}, which focused on perfect sampling (i.e., $\varepsilon=0$) and used rejection sampling. To prove Proposition \ref{prop:oversampling}, we first develop a theory of \emph{robust} rejection sampling in Section~\ref{subsec:robustrs}. We then prove Proposition \ref{prop:oversampling} in Section \ref{subsec:proofoversampling}.

\subsubsection{Robust rejection sampling}\label{subsec:robustrs}
\paragraph{Standard rejection sampling.} 
Rejection sampling is a method to generate samples from a (hard) distribution $p_1\colon\set{1}{n}\to [0,1]$ given the ability to generate samples from an (easier) distribution $p_2\colon\set{1}{n}\to [0,1]$. In its standard statement, the following three conditions should be satisfied:
\begin{itemize}
    \item[(i)]
    we can generate samples from $p_2$;
    \item[(ii)]
    we know a value $m$ such that $\frac{p_1(j)}{mp_2(j)}\le 1$ holds for all $j$;
    \item[(iii)] 
    given $j\in\set{1}{n}$, we can compute $p_1(j)$ and $p_2(j)$.
\end{itemize}
The method works as follows: sample $j\in\set{1}{n}$ from $p_2$; output it with probability $\frac{p_1(j)}{mp_2(j)}$ and otherwise report ``failure''. It is not difficult to show that conditioned on the event that the process does not fail, which happens with probability $1/m$, the probability that sample $j$ is output is precisely $p_1(j)$. Repeating the process $\Theta(m)$ times will then output a sample drawn from the distribution $p_1$ with high probability. 

\paragraph{Robust rejection sampling.} 
We consider the setting where we have two probability distributions $p_1,p_2\colon\set{1}{n}\to [0,1]$ satisfying the following conditions:
\begin{itemize}
    \item[(i')]
    we can generate samples from a probability distribution $\tilde p_2\colon\set{1}{n}\to [0,1]$ such that $\stat{\tilde p_2 -p_2}\le \varepsilon$;
    \item[(ii')] 
    there exists a value $m$ such that $\frac{p_1(j)}{mp_2(j)}\le 1$ holds for all $j$;
    \item[(iii')] 
    for each $j\in\set{1}{n}$ we can compute the value $\frac{p_1(j)}{mp_2(j)}$.
\end{itemize}
In particular, we do not assume that we can compute $\tilde p_2(j)$ efficiently. Condition (i') is a clear relaxation of Condition (i).  Conditions (ii') and (iii') are slight relaxations of Conditions (ii) and~(iii) that will be crucial to derive our results (since in our setting we are not able to easily compute $p_1(j)$ and $p_2(j)$). 

The goal is to output a sample drawn from a distribution close to~$p_1$. The basic sampling procedure we use is described in Figure \ref{fig:alg0}. Note that Step 1 can be implemented due to Condition (i'), while Step 3 can be implemented due to Condition (iii').

Here is our main result.

\begin{figure}[h!]
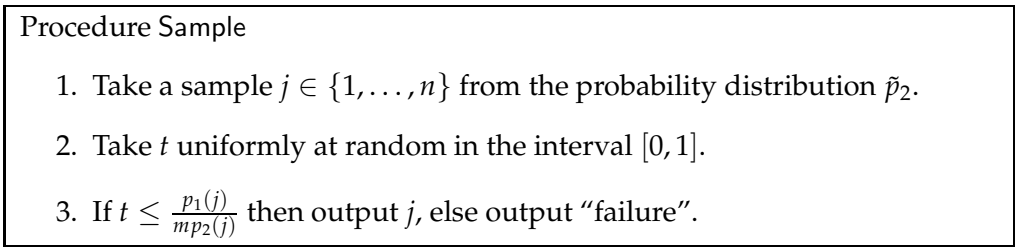

\begin{center}
\fbox{
\begin{minipage}{13 cm} 
Procedure $\sample$
\begin{itemize}
\item[1.]
Take a sample $j\in\set{1}{n}$ from the probability distribution $\tilde p_2$.
\item[2.] 
Take $t$ uniformly at random in the interval $[0,1]$.
\item[3.]
If $t\le \frac{p_1(j)}{mp_2(j)}$ then output $j$, else output ``failure''.
\end{itemize}
\end{minipage}
}
\end{center}
\vspace{-2mm}
\caption{Procedure $\sample$.}\label{fig:alg0}
\end{figure}
The following proposition analyzes the behavior of the procedure $\sample$.
\begin{proposition}\label{prop:rs}
For any $\varepsilon\in[0,\frac{1}{2m}]$, the success probability of the procedure $\sample$ is $\Theta(1/m)$ and conditioned on success,  the total variation distance between the distribution of its output $j$ and $p_1$ is at most $3 \varepsilon m$.
\end{proposition}
\begin{proof}
Let $p_{\mathrm{acc}}$ denote the probability that Procedure $\sample$ does not fail. We have 
\begin{align*}
p_{\mathrm{acc}}=\sum_{j=1}^n \tilde p_2(j) \frac{p_1(j)}{mp_2(j)}
&=
\sum_{j=1}^n p_2(j) \frac{p_1(j)}{mp_2(j)} + 
\sum_{j=1}^n (p_2(j) - \tilde p_2(j))
\frac{p_1(j)}{mp_2(j)}\\
&=
\frac{1}{m} + 
\sum_{j=1}^n (p_2(j) - \tilde p_2(j))
\frac{p_1(j)}{mp_2(j)}
\end{align*} 
and then
\begin{align*}
    \abs{p_{\mathrm{acc}}-\frac{1}{m}}
    &\le
    \sum_{j=1}^n \abs{p_2(j) - \tilde p_2(j)}
\frac{p_1(j)}{mp_2(j)}
\le
    \stat{p_2(j) - \tilde p_2(j)}
    \le \varepsilon,
\end{align*}
where we used the condition $\frac{p_1(j)}{mp_2(j)}\le 1$ for the second inequality.
Note that we thus have 
\begin{align*}
\frac{1}{p_{\mathrm{acc}}}&\in
\left[
\frac{m}{(1+\varepsilon m)},\frac{m}{(1-\varepsilon m)}
\right]\subseteq
\left[
(1-\varepsilon m)m,(1+2\varepsilon m)m
\right]
,
\end{align*}
since we are assuming $\varepsilon\in[0,\frac{1}{2m}]$. In particular we have $p_{\mathrm{acc}}=\Theta(1/m)$.

Let $\hat p\colon\set{1}{n}\to [0,1]$ be the probability distribution corresponding to the output of Procedure $\sample$ conditioned on the event that Procedure $\sample$ stops. By definition we have
\[
\hat p(j)=
\frac{\tilde p_2(j) p_1(j)}
{mp_{\textrm{acc}}p_2(j)}
\]
for all $j\in\set{1}{n}$.

We now show that the distribution $\hat p$ is close to the distribution $p_{1}$.
We have 
\begin{align*}
   \stat{\hat p- p_{1}}&=
   \frac{1}{2}\sum_{j=1}^n
   \abs{
   \frac{\tilde p_2(j) p_1(j)}
{mp_{\textrm{acc}}p_2(j)}
- p_1(j)
}\\
&\le
   \frac{1}{2}\sum_{j=1}^n
   \abs{
   \frac{p_2(j) p_1(j)}
{mp_{\textrm{acc}}p_2(j)}
- p_1(j)
}
+ \frac{1}{2}\sum_{j=1}^n
\abs{
    \frac{(p_2(j)-\tilde p_2(j)) p_1(j)}
{mp_{\textrm{acc}}p_2(j)}
}\\
&=
   \frac{1}{2}\sum_{j=1}^n
   \abs{
   \left(\frac{1}{mp_{\textrm{acc}}}-1\right)
   p_1(j)}
+ \frac{1}{2}\sum_{j=1}^n
\abs{
    \frac{(p_2(j)-\tilde p_2(j)) p_1(j)}
{mp_{\textrm{acc}}p_2(j)}
}\\
&\le 
\frac{1}{2}\abs{
   \frac{1}{mp_{\textrm{acc}}}-1}
   +
   \frac{1}{p_{\textrm{acc}}}
   \stat{p_2(j)-\tilde p_2(j)}
   \\
&\le
\varepsilon m + (1+2\varepsilon m)m\cdot \varepsilon\\
&\le
3\varepsilon m,
 \end{align*}
where we used again the condition $\frac{p_1(j)}{mp_2(j)}\le 1$ for the second inequality and $\varepsilon m\le 1/2$ for the last inequality.
%
\end{proof}

\subsubsection{Proof of Proposition \ref{prop:oversampling}}\label{subsec:proofoversampling}

\begin{proof}[Proof of Proposition \ref{prop:oversampling}]
    We take $p_1=p_u$, $p_2=p_{\bar u}$, $\tilde p_2=\tilde p_{\bar u}$ and $m=\phi$. Note that 
\begin{equation}\label{eq:m}
\frac{p_1(j)}{mp_2(j)}=\frac{\abs{u(j)}^2\norm{\bar u}^2}{\phi\abs{\bar u(j)}^2\norm{u}^2}
=\frac{\abs{u(j)}^2}{\abs{\bar u(j)}^2}
\le 1
\end{equation}
holds for all $j\in\set{1}{n}$, as required.

Consider repeating the procedure $\sample$ until it outputs a sample. Proposition \ref{prop:rs} guarantees that the expected number of repetitions is $O(\phi)$ and conditioned on success, the  total variation distance between the output distribution and $p_u$ is at most $3\varepsilon \phi$. Repeating the procedure $O(\phi_{\textrm{max}}\log(1/\delta))$ times guarantees that with probability at least $1-\delta$ the procedure successfully outputs a sample. For each execution of the procedure $\sample$, Step 1 only requires taking one sample from $\tilde p_{\bar u}$, which has cost $\costsq{\bar u}$, and Step 3 can be implemented via Equation (\ref{eq:m}) by querying $\abs{u(j)}$ and $\abs{\bar u(j)}$, which has cost $\costq{u}+\costq{\bar u}$. To get a sample, the overall complexity is thus $O(\phi_{\textrm{max}}\log(1/\gamma)\cdot\costosq{u})$.

To compute an approximation of $\norm{u}$, observe that $\norm{u}=\norm{\bar u}/\phi$. Since we can get $\norm{\bar u}$, we only need to approximate $\phi^{-1}$. We run $r=\ceil{3\phi_{\textrm{max}}\log(2/\delta)/\eta^2}$ times the procedure $\sample$. Let $X$ be the fraction of these $r$ executions that are successful and observe that $\ex{X}=\phi^{-1}$. From Chernoff's bound we have 
\[
\Prob{\abs{X-\phi^{-1}}\ge \eta \phi^{-1}}\le 2\exp\left(-\frac{\eta^2r }{3\phi}\right)\le  \delta.
\]
We take $X$ as our estimator of $\phi^{-1}$ and output $X\norm{\bar u}$.
Then with probability at least $1-\delta$ we have 
\begin{align*}
\abs{\norm{u} - X\norm{\bar u}}
&
\le 
\eta \norm{u},
\end{align*}
as claimed. The overall cost of computing an approximation of $\norm{u}$ is  thus
\[
O(r\cdot \costosq{u})=O\left(\frac{\phi_{\textrm{max}}\log(1/\delta)}{\eta^2}\costosq{u}\right),
\]
which gives the claimed cost for $\RSQ{\varepsilon',\delta,\eta}(u)$.
\end{proof}

\subsection{Sampling from linear combinations of rows}\label{sub:line}
In this subsection we show that our notion of approximate oversampling-and-query access is robust under taking linear combinations of a small number of vectors. Here is our main result.
\begin{proposition}\label{prop:lin1}
For a matrix $R\in\Comp^{r\times n}$, assume that we have
    $\OSQ{\varepsilon,\phi_i}(R(i,\cdot))$ for each $i\in\set{1}{r}$ with parameters $\varepsilon\ge 0$ and $\phi_1,\ldots,\phi_r\ge 1$. Let $v^\dagger=(\lambda_1,\ldots,\lambda_r)$ be a known row-vector, and assume that the row-vector $v^\dagger R$ is nonzero. We can then implement $\OSQ{\varepsilon,\phi}(v^\dagger R)$  with 
\begin{equation}\label{ineq:combination}
\phi= r
\frac{\sum_{i=1}^r \phi_i\norm{\lambda_iR(i,\cdot)}^2}{\norm{v^\dagger R}^2}
\end{equation}
at cost $O(r+\costosq{R(1,\cdot)}+\cdots+\costosq{R(r,\cdot)})$.
\end{proposition}
This proposition is a generalization of Lemma 3.6 in \cite{Chia+JACM22} and Proposition 4.3 in \cite{TangSTOC19}. The proof strategy is similar, but additional care is needed in the analysis to ensure that the total variation distance between the distribution we construct and the ideal distribution is preserved.


\begin{proof}[Proof of Proposition \ref{prop:lin1}.]
Note that
$
v^\dagger R = \sum_{i=1}^r \lambda_i R(i,\cdot).
$
For any $j\in\set{1}{n}$ we can compute $[v^\dagger R](j)$ by querying $R(1,j),\ldots,R(r,j)$ and then performing arithmetic operations, which means that we can implement one query access to the vector $v^\dagger R$ at cost 
$O(r+\costq{R(1,\cdot)}+\cdots+\costq{R(r,\cdot)})$.

The remaining of the proof explains how to implement approximate sampling-and-query access to a vector $w$ such that $\norm{w}=\phi\norm{v^\dagger R}$ for $\phi$ satisfying Equation (\ref{ineq:combination}), and $\abs{w(j)}\ge \abs{[v^\dagger R](j)}$ for all $j\in\set{1}{n}$.

\paragraph{Notations.}
We first introduce notations related to the vectors we can query and sample. From the assumptions, for each $i\in\set{1}{r}$ we have query access to $R(i,\cdot)$ and approximate sampling-and-query access to a vector, which we denote $u_i\in\Comp^n$, such that $\norm{u_i}=\phi_i \norm{R(i,\cdot)}$ and $\abs{u_i(j)}\ge \abs{R(i,j)}$ for all $j\in\set{1}{n}$. This means that for each $i\in\set{1}{r}$, each of the following operations can be implemented at cost $\costosq{R(i,\cdot)}$:

\begin{itemize}
\item
query access to $R(i,\cdot)$ and $u_i$;
\item
getting $\norm{u_i}$;
\item
getting one sample from a distribution $\tilde p_{u_i}\colon\{1,\ldots,n\}\to [0,1]$ such that $\stat{\tilde p_{u_i}-{p}_{u_i}}\le \varepsilon$.
\end{itemize}

\paragraph{Definition of the vector $\boldsymbol{w}$.}
Define the row-vector $w\in\Real^{1\times n}$ as follows:
\[
w(j)=\sqrt{r\sum_{i=1}^r\abs{\lambda_i u_i(j)}^2}
\]
for all $j\in\set{1}{n}$. For all $j\in\set{1}{n}$, we have 
\[
w(j)\ge \sqrt{r\sum_{i=1}^r\abs{\lambda_i R(i,j)}^2}
\ge \abs{\sum_{i=1}^r\lambda_i R(i,j)} = 
\abs{[v^\dagger R] (j)},
\]
as claimed, by Cauchy-Schwarz inequality. We also have $\norm{w}^2=\phi \norm{v^\dagger R}^2$ for 
\begin{align*} 
\phi=
\frac{r\sum_{i=1}^r\abs{\lambda_i}^2 \norm{u_i}^2}{\norm{\sum_{i=1}^r \lambda_i R(i,\cdot)}^2}
=
\frac{r\sum_{i=1}^r\phi_i\abs{\lambda_i}^2 \norm{R(i,\cdot)}^2}{\norm{\sum_{i=1}^r \lambda_i R(i,\cdot)}^2},
\end{align*}
as claimed.

\paragraph{Query access to $\boldsymbol{w}$.}
For any $j\in\set{1}{n}$ we can compute $w(j)$ by querying $u_1(j),\ldots,u_r(j)$ and then performing arithmetic operations, which means that we can implement one query access to the vector $w$ with cost $O(r+\costosq{R(1,\cdot)}+\cdots+\costosq{R(r,\cdot)})$.

\paragraph{Estimating the norm of $\boldsymbol{w}$.}
Observe that
\[
\norm{w} = \sqrt{r\sum_{i=1}^r \abs{\lambda_i}^2 \norm{u_i}^2}.
\]
This quantity can be computed using $O(r)$ arithmetic operations once we know $\norm{u_1},\ldots,\norm{u_r}$. The cost is $O(r+\costosq{R(1,\cdot)}+\cdots+\costosq{R(r,\cdot)})$.

\paragraph{Sampling from $\boldsymbol{w}$.}
Let $q\colon\set{1}{r}\to[0,1]$ be the probability distribution such that
\[
q(i)=\frac{\abs{\lambda_i}^2\norm{ u_i}^2}{\sum_{\ell=1}^r \abs{\lambda_\ell}^2\norm{ u_\ell}^2}.
\]
Observe that the distribution $q$ can be explicitly constructed by getting the norms $\norm{u_1},\ldots,\norm{u_r}$.
In order to sample from $w$, we execute the following process. \vspace{2mm}

\fbox{
\begin{minipage}{12 cm} 
\vspace{2mm}
1. 
Take a sample $i\in\set{1}{r}$ from the probability distribution $q$.\\

\noindent 2.
Take a sample $j\in\set{1}{n}$ from the probability distribution $\tilde p_{u_i}$.  
\vspace{2mm}
\end{minipage}
}

\vspace{5mm}
The cost of the whole process (including the construction of $q$) is 
\[
O(r+\costosq{R(1,\cdot)}+\cdots+\costosq{R(r,\cdot)}).
\]
Let $P\colon\set{1}{n}\to [0,1]$ be the probability distribution of the sample obtained at Step 2, i.e., 
\[
P(j)=
\sum_{i=1}^r
q(i)
\tilde p_{u_i}(j)
\]
for all $j\in\set{1}{n}$. 
We have

    \begin{align*}
\stat{P-p_w}&=
\frac{1}{2}\sum_{j=1}^n \abs{\sum_{i=1}^r
q(i)
\tilde p_{u_i}(j) -
\frac{\abs{w(j)}^2}{\norm{w}^2}}\\
&=
\frac{1}{2}\sum_{j=1}^n \abs{\sum_{i=1}^r
q(i)
\tilde p_{u_i}(j) -
\frac{r\sum_{i=1}^r\abs{\lambda_i u_i(j)}^2}{r\sum_{i=\ell}^r\abs{\lambda_\ell}^2\norm{u_\ell}^2}}\\
&=
\frac{1}{2}\sum_{j=1}^n \abs{\sum_{i=1}^r
q(i)
\tilde p_{u_i}(j) -
\sum_{i=1}^r
 q(i)
\frac{\abs{u_i(j)}^2}{\norm{u_i}^2}}\\
&\le
\sum_{i=1}^r
q(i)\stat{\tilde p_{u_i}-p_{u_i}}\\
&\le\varepsilon,
\end{align*}
which concludes the proof.
\end{proof}
\subsection{Importance matrix sketches associated with sampling access}\label{sub:app1}
We now discuss how to construct and use importance matrix sketches when given oversampling-and-query access to matrices. 

Remember how we defined a special case of approximate importance matrix sketches in Definition \ref{def:impa} and a special case of approximate joint importance matrix sketches in Definition~\ref{definition:mixeda}. We first observe that the standard representation (as defined in Definition~\ref{def:S}) of such matrix sketches can be computed efficiently:
\begin{lemma}\label{lemma:createsk}
The standard description of the $(r,\varepsilon,\phi)$-approximate importance matrix sketch of $A$ associated with $\OSQ{\varepsilon,\phi}(A)$ can be computed at cost $O(r\cdot \costosq{A})$. The standard description of the $(r,\varepsilon,\phi,\phi')$-approximate joint importance matrix sketch of $A$ and $B$ associated with $\OSQ{\varepsilon,\phi}(A)$ and $\OSQ{\varepsilon,\phi'}(B)$ can be computed at cost $O(r\cdot (\costosq{A}+\costosq{B}))$. 
\end{lemma}
\begin{proof}
The first part follows from the properties of the distributions $\tilde p_{\brow{A}}$ and $p_{\brow{A}}$. Observe in particular that since we have $\SQ{\varepsilon}(\brow{A})$, we can sample from $p_{\brow{A}}$ at cost $\costosq{A}$ and compute $p_{\brow{A}}(i)=\norm{\brow{A}(i)}^2/\norm{\brow{A}}^2$ at cost $O(\costosq{A})$ for any $i\in\set{1}{n}$. We can thus compute the standard description of the $(r,\varepsilon,\phi)$-approximate importance matrix sketch of $A$ associated with $\OSQ{\varepsilon,\phi}(A)$ at cost $O(r\cdot \costosq{A})$. 

The proof of the second part is similar. For sampling from the distribution $\frac{1}{2}(p_{\brow{A}}+p_{\brow{B}})$, we simply take a random bit $b$, and sample from $p_{\brow{A}}$ if $b=0$ and from $p_{\brow{B}}$ if $b=1$. 
\end{proof}


We now show how to get oversampling-and-query access to the matrices $SA$ and $(SA)^\dagger$ when given an approximate importance matrix sketches of $A$ associated with $\OSQ{\varepsilon,\phi}(A)$. The two following lemmas generalize similar results for the case of perfect oversampling-and-query access  proved in \cite{Chia+JACM22,Frieze+JACM04}.

\begin{lemma}\label{lemma:sample1}
For a nonzero matrix $A\in \Comp^{m\times n}$, any parameters $\varepsilon\in[0,1]$, $\phi\ge 1$ and any integer $r\ge 1$, assume that we have $\OSQ{\varepsilon,\phi}(A)$ and know the standard representation of the $(r,\varepsilon,\phi)$-approximate importance matrix sketch of $A$ associated with $\OSQ{\varepsilon,\phi}(A)$.
Assume that the inequality
$
\fnorm{SA}^2\ge \frac{1}{2}\fnorm{A}^2
$
 holds. We can then implement $\OSQ{\varepsilon,\phi'}(SA)$ with $\phi'\le 2\phi$ at cost $O(\costosq{A})$.
\end{lemma}

\begin{proof}
Write $((s_1,\alpha_1),\ldots,(s_r,\alpha_r))$ the standard representation of $S$, where
\[
\alpha_i=\frac{1}{\sqrt{r p_{\brow{A}}(s_i)}}= \frac{\fnorm{\bar A}}{\sqrt{r}\norm{\bar A(s_i,\cdot)}}\le \frac{\sqrt{\phi}\fnorm{A}}{\sqrt{r}\norm{A(s_i,\cdot)}}
\]
if $\norm{\bar A(s_i,\cdot)}\neq 0$, and $\alpha_i=0$ if $\norm{\bar A(s_i,\cdot)}=0$. Define 
\[
\Delta=\left\{i\in\set{1}{r}\:|\: \norm{\bar A(s_i,\cdot)}\neq 0\right\}.
\]

For any $(i,j)\in\set{1}{r}\times\set{1}{n}$, we can compute $[SA](i,j)$ using the identity $[SA](i,j)=\alpha_i A(s_i,j)$, which requires only one query access to $A$ and an arithmetic operation. 

We define the matrix $\olsi{SA}=S\bar{A}$. We have 
\[
\fnorm{\olsi{SA}}^2=
\sum_{i=1}^r\abs{\alpha_i}^2\norm{\bar A(s_i,\cdot)}^2
=\frac{\abs{\Delta}}{r}\fnorm{\bar A}^2
\le
\phi\fnorm{A}^2\le 2\phi\fnorm{SA}^2
\]
 and $\abs{[\olsi{SA}](i,j)}=\abs{\alpha_i}\abs{\bar A (s_i,j)}\ge \abs{\alpha_i}\abs{A (s_i,j)}=\abs{[SA](i,j)}$ for all $(i,j)\in\set{1}{r}\times\set{1}{n}$.
We now explain how to implement $\SQ{\varepsilon}(\olsi{SA})$. \vspace{2mm}

\noindent{\bf Approximate sampling-and-query access to the rows of $\boldsymbol{\olsi{SA}}$.}
For each $i\in\set{1}{r}$, we explain how to implement $\SQ{\varepsilon}([\olsi{SA}](i,\cdot))$. We can implement query access from the identity $[\overline{SA}](i,j)=\alpha_i\bar A(s_i,j)$ using only one query access to $\bar A$ and one arithmetic operation. Similarly, we can compute $\norm{[\olsi{SA}](i,\cdot)}=\abs{\alpha_i}\norm{A(s_i,\cdot)}$ by getting the norm $\norm{\bar A(s_i,\cdot)}$. To sample from a nonzero vector $\olsi{SA}(i,\cdot)$, we simply sample from the distribution $\tilde p_{\bar{A}(s_i,\cdot)}$. Since $p_{\olsi{SA}(i,\cdot)}=p_{\bar{A}(s_i,\cdot)}$, we have 
\[
\stat{\tilde p_{\bar{A}(s_i,\cdot)} - p_{\olsi{SA}(i,\cdot)}} = \stat{\tilde p_{\bar{A}(s_i,\cdot)} - p_{\bar{A}(s_i,\cdot)}}\le \varepsilon.
\]
The cost of each of these operations is $O(\costosq{A})$.\vspace{2mm}

\noindent{\bf Approximate sampling-and-query access to the row norm vector of $\boldsymbol{\olsi{SA}}$.}
We now explain how to implement $\SQ{\varepsilon}(\row{\olsi{SA}})$ for the vector $\row{\olsi{SA}}=(\norm{[\olsi{SA}](1,\cdot)},\ldots, \norm{[\olsi{SA}](r,\cdot)})$. We have $\row{\olsi{SA}}(i)=\abs{\alpha_i}\norm{\bar A(s_i,\cdot)}= \sqrt{1/r}\fnorm{\bar A}$ for all $i\in\Delta$, and $\row{\olsi{SA}}(i)=0$ for all $i\notin\Delta$. We can implement query access to $\row{\olsi{SA}}$ using one query to the norm $\norm{\bar A(s_i,\cdot)}$ (or to the norm $\norm{\brow{A}}=\fnorm{\bar A}$) and one arithmetic operation. Since $\norm{\row{\olsi{SA}}}=\sqrt{\abs{\Delta}/r}\fnorm{\bar A}=\sqrt{\abs{\Delta}/r}\norm{\brow{A}}$, we can compute $\norm{\row{\olsi{SA}}}$ by getting $\norm{\brow{A}}$, which can be done at cost $\costosq{A}$.
Implementing (perfect) sampling access to $\row{\olsi{SA}}$ is trivial since $p_{\row{\olsi{SA}}}$ is the uniform distribution over~$\Delta$.
\end{proof}

\begin{lemma}\label{lemma:sample2}
For a nonzero matrix $A\in \Comp^{m\times n}$, any parameters $\varepsilon\in[0,1]$, $\phi\ge 1$ and any integer $r\ge 1$, assume that we have $\OSQ{\varepsilon,\phi}(A)$ and know the standard representation of the $(r,\varepsilon,\phi)$-approximate importance matrix sketch of $A$ associated with $\OSQ{\varepsilon,\phi}(A)$.
Assume that the inequality
$
\fnorm{SA}^2\ge \frac{1}{2}\fnorm{A}^2
$
 holds. We can then implement $\OSQ{\varepsilon,\phi'}((SA)^\dagger)$ with $\phi'\le 2\phi$ at cost $O(r\cdot \costosq{A})$.
\end{lemma}

\begin{proof}
For conciseness, let us write $M=(SA)^\dagger$. Write $((s_1,\alpha_1),\ldots,(s_r,\alpha_r))$ the standard representation of $S$ and define the set $\Delta$ as in the proof of Lemma \ref{lemma:sample1}.

 For any $(i,j)\in\set{1}{r}\times\set{1}{n}$, we can compute $M(j,i)$ using the identity $M(j,i)=[SA](i,j)=\alpha_i A(s_i,j)$, which requires only one query access to $A$ and an arithmetic operation. 

We define the matrix $\olsi M=(S\bar A)^\dagger$. We have 
\[
\fnorm{\olsi M}^2=
\sum_{i=1}^r\abs{\alpha_i}^2\norm{\bar A(s_i,\cdot)}^2=\frac{\abs{\Delta}}{r} 
\fnorm{\bar A}^2\le \phi\fnorm{A}^2\le 2\phi\fnorm{SA}^2=2\phi\fnorm{M}^2
\]
and $\abs{\olsi M(j,i)}=\abs{\alpha_i}\abs{\bar A(s_i,j)}\ge \abs{\alpha_i}\abs{A(s_i,j)}=\abs{M(j,i)}$ for any $(i,j)\in\set{1}{r}\times\set{1}{n}$.
We now explain how to implement $\SQ{\varepsilon}(\olsi M)$. \vspace{2mm}

\noindent {\bf Approximate sampling-and-query access to the rows of $\boldsymbol{\olsi M}$.}
For each $j\in\set{1}{n}$, we explain how to implement $\SQ{\varepsilon}(\olsi M (j,\cdot))$. For any $i\in\set{1}{r}$, we can implement query access to this vector from the identity $\olsi M(j,i)=\alpha_i\bar A(s_i,j)$ using only one query access to $\bar A$ and one arithmetic operation. We can compute its norm $\norm{\olsi M(j,\cdot)}$ by querying $\bar A(s_1,j),\ldots, \bar A(s_r,j)$ and then performing $O(r)$ arithmetic operations, which has cost $O(r\cdot \costosq{A})$. Similarly, we can (exactly) sample from the distribution $p_{\olsi M(j,\cdot)}$ by first querying $\bar A(s_1,j),\ldots, \bar A(s_r,j)$ and then explicitely computing this distribution. The cost is again $O(r\cdot \costosq{A})$. \vspace{2mm}

\noindent{\bf Approximate sampling-and-query access to the row norm vector of $\boldsymbol{\olsi{M}}$.} 
We now explain how to implement $\SQ{\varepsilon}(\browolsi{M})$ for the vector $\browolsi{M}=(\norm{\olsi M(1,\cdot)},\ldots,\norm{\olsi M(n,\cdot)})$. Query access can be done at cost $O(r\cdot \costosq{A})$ by computing $\norm{\olsi M(j,\cdot)}$ as described above. Since $\norm{\row{\olsi{M}}}=\fnorm{\olsi M}=\sqrt{\abs{\Delta}/r}\fnorm{\bar {A}}$, we can compute $\norm{\row{\olsi{M}}}$ at cost $\costosq{A}$ by getting $\norm{\brow{A}}=\fnorm{\bar A}$. Sampling from $\browolsi{M}$ is more delicate. Here is our strategy: we take an index $i\in\Delta$ uniformly at random,  and then return the index $j\in\set{1}{n}$ sampled from the distribution $\tilde p_{\bar A(s_i,\cdot)}$. Using the fact that $\norm{\bar A(s_i,\cdot)}=0$ for $i\notin \Delta$, the total variation distance between the resulting distribution and the distribution $p_{\browolsi{M}}$ can be evaluated as follows:
\begin{align*}
\frac{1}{2}\sum_{j=1}^n 
\left\vert
\frac{1}{\abs{\Delta}}\sum_{i\in\Delta} \tilde p_{\bar A(s_i,\cdot)} (j)
-
p_{\row{\olsi M}}(j)
\right\vert
=&
\frac{1}{2}\sum_{j=1}^n 
\left\vert
\frac{1}{\abs{\Delta}}\sum_{i\in\Delta} \tilde p_{\bar A(s_i,\cdot)} (j)
-
\frac{\norm{\olsi M(j,\cdot)}^2}{\norm{\browolsi{M}}^2}
\right\vert\\
=&
\frac{1}{2}\sum_{j=1}^n 
\left\vert
\frac{1}{\abs{\Delta}}\sum_{i\in\Delta} \tilde p_{\bar A(s_i,\cdot)} (j)
-
\frac{r}{\abs{\Delta}}
\sum_{i\in\Delta}
\frac{\norm{[S\bar A](i,j)}^2}{\fnorm{\bar A}^2}
\right\vert\\
\le&
\frac{1}{2}\sum_{j=1}^n 
\sum_{i\in\Delta}
\left\vert
\frac{\tilde p_{\bar A(s_i,\cdot)} (j)}{\abs{\Delta}}
-
\frac{r\norm{[S\bar A](i,j)}^2}{\abs{\Delta}\fnorm{\bar A}^2}
\right\vert\\
=&
\frac{1}{2}\sum_{i\in\Delta}
\sum_{j=1}^n 
\left\vert
\frac{\tilde p_{\bar A(s_i,\cdot)} (j)}{\abs{\Delta}}
-
\frac{r\alpha_i^2\norm{\bar A(s_i,j)}^2}
{\abs{\Delta}\fnorm{\bar A}^2}
\right\vert\\
=&
\frac{1}{2}\sum_{i\in\Delta}
\sum_{j=1}^n 
\left\vert
\frac{\tilde p_{\bar A(s_i,\cdot)} (j)}{\abs{\Delta}}
-
\frac{\fnorm{\bar A}^2}{\abs{\Delta}\norm{\bar A(s_i,\cdot)}^2}\frac{\norm{\bar A(s_i,j)}^2}
{\fnorm{\bar A}^2}
\right\vert\\
=&
\frac{1}{\abs{\Delta}}
\sum_{i\in\Delta}
\stat{
\tilde p_{\bar A(s_i,\cdot)}
-
p_{\bar A(s_i,\cdot)}
}\\
\le& \varepsilon.
\end{align*}
The cost to output one sample is $O(\costosq{A})$.
\end{proof}

\section{Applications}\label{sec:app}
In this section we discuss applications of the framework and techniques developed in the previous sections.  In Section \ref{sub:IP} we first show how to apply Theorem~\ref{th:MM1} to estimate the inner product of two vectors. In Sections \ref{sub:app2} and \ref{sub:app8}, we apply our machinery to obtain a robust dequantization of the Quantum Singular Value Transformation in the low-rank setting (Section~\ref{sub:app2}) and the sparse setting (Section \ref{sub:app8}). In Sections \ref{sub:app3}, \ref{sub:app4} and \ref{sub:app5}  we then show how to obtain a robust dequantization of known quantum algorithms for supervised clustering, recommendation systems and low-rank matrix inversion. 

When we consider in this section oversampling-and-query access to the input, for simplicity we always assume that we know the oversampling parameter $\phi$ (if we know instead only an upper bound $\phi_{\textrm{max}}\ge \phi$, all the statements still hold by replacing $\phi$ by $\phi_{\textrm{max}}$).

\subsection{Estimation of the inner product with oversampling access}\label{sub:IP}
In this subsection we describe how to estimate the inner product with oversampling access. Here is our main result, which essentially generalizes Proposition \ref{th:IP} to approximate oversampling-and-query access.
\begin{proposition}\label{prop:IP}
For any $\varepsilon\in[0,1]$, any $\delta\in(0,1]$, any $\xi>0$ and any $\phi\ge 1$,
assume that 
\begin{itemize}
\item
we have $\OSQ{\varepsilon,\phi}(u)$ for a nonzero vector $u\in\Comp^m$,
\item
we have $\Q{}(v)$ for a nonzero vector $v\in\Comp^m$ and know $\norm{v}$.
\end{itemize}
 We can then output an estimator $\alpha$ such that $\abs{\alpha-(u,v)}\le\left(2\sqrt{2\phi\varepsilon}+\xi\right)\norm{u}\norm{v}$ holds with probability at least $1-\delta$ at cost $O\left(\frac{\phi\log(1/\delta)}{\xi^2}(\costosq{u}+\costq{v})\right)$.
\end{proposition}
\begin{proof}
Take 
\[
r=\ceil{\frac{16\phi}{\xi^2}}.
\]
and consider the $(r,\varepsilon,\phi)$-approximate importance matrix sketch $\Sigma$ of $u$ associated with $\OSQ{\varepsilon,\phi}(u)$.
The standard description of this sketch can be computed at cost $O(r\cdot \costosq{u})$ via Lemma \ref{lemma:createsk}.
We apply Theorem~\ref{th:MM1} with $n=n'=1$ and $X=u^\ast$ and $Y=v^\ast$ (remember that $u^\ast$ and $v^\ast$ denote the complex conjugate of $u$ and $v$). Consider the complex number $\beta=X^\dagger \bar\Sigma^\dagger \bar\Sigma Y$, which can be computed at cost $O(r\cdot (\costosq{u}+\costq{v}))$. The first part of the proof of Theorem~\ref{th:MM1} shows that the inequality
\[
\abs{\ex{\beta}-(u,v)}\le 2\sqrt{2\phi\varepsilon}\norm{u}\norm{v}
\]
holds. The second part of the proof of Theorem~\ref{th:MM1} shows that the inequality
\[
\abs{\beta-\ex{\beta}}\le \frac{\xi}{\sqrt{2}}\norm{u}\norm{v}.
\]
holds with probability at least $3/4$.
By using Lemma \ref{lemma:powering}, we can get an estimate $\alpha$ such that $\abs{\alpha-\ex{\beta}}\le \xi\norm{u}\norm{v}$ holds with probability at least $1-\delta$. From the triangle inequality we thus get
\[
\Prob{\abs{\alpha-(u,v)}\le\left(2\sqrt{2\phi\varepsilon}+\xi\right)\norm{u}\norm{v}}\ge 1-\delta,
\]
as claimed.
\end{proof}

As a corollary, we can estimate values of the form $u^\dagger Av$ as well:
\begin{corollary}\label{expest}
For any $\varepsilon\in[0,1]$, any $\delta\in(0,1]$, any $\xi>0$ and any $\phi\ge 1$,
assume that 
\begin{itemize}
\item
we have $\OSQ{\varepsilon,\phi}(A)$ for a nonzero matrix $A\in\Comp^{m\times m}$,
\item
we have $\Q{}(u)$ and $\Q{}(v)$ for two nonzero vectors $u,v\in\Comp^m$ and know $\norm{u}$ and $\norm{v}$.
\end{itemize}
 We can then output an estimator $\alpha$ such that $\abs{\alpha-u^\dagger A v}\le\left(4\sqrt{\varepsilon}+\xi\right)\norm{u}\fnorm{A}\norm{v}$ holds with probability at least $1-\delta$ at cost $O\left(\frac{\phi\log(1/\delta)}{\xi^2}(\costosq{A}+\costq{u}+\costq{v})\right)$.
\end{corollary}
\begin{proof}
Observe that $u^\dagger A v=\Tr(u^\dagger A v)=\Tr(Av u^\dagger)=(x,y)$, where $x$ is the vector of length~$m^2$ obtained by concatenating the entries of matrix $A$ row by row, and $y$ is the vector of length~$m^2$ obtained by concatenating the entries of the matrix $vu^\dagger$ column by column. 

We show how to implement $\OSQ{2\varepsilon,\phi}(x)$ at cost $O(\costosq{A})$. We define $\bar x$ as the vector of length~$m^2$ obtained by concatenating the entries of matrix $\bar A$ row by row, where $\bar A$ represents the matrix in Definition \ref{def:qos-access-mat}. We index the entries of $x$ and $\bar x$ by $(i,j)\in\set{1}{m}\times\set{1}{m}$, so that $x(i,j)=A(i,j)$ and $\bar x(i,j)=\bar A(i,j)$.
Query access to $x$ and $\bar x$, as well as the norms $\norm{x}$ and $\norm{\bar x}$, can trivially be obtained from $\OSQ{\varepsilon,\phi}(A)$. Let us write $\Delta=\{i\in\set{1}{m}\:|\:\norm{\bar A(i,\cdot)}\neq 0\}$
To sample from~$\bar x$, we sample an index $i\in\set{1}{m}$ from $\tilde p_{\brow{A}}$. If $i\in\Delta$, then we sample an index $j\in\set{1}{m}$ from $\tilde p_{\bar A(i,\cdot)}$ and output the index $(i,j)$. If $i\notin\Delta$, then we output the index $(i,1)$.
Let $P\colon\set{1}{m}\times\set{1}{m}\to[0,1]$ denote the resulting probability distribution. 
For conciseness, let us write $p=p_{\row{\bar A}}$, $\tilde p=\tilde p_{\row{\bar A}}$ and $p_i=p_{\bar A(i,\cdot)}$, $\tilde p_i=\tilde p_{\bar A(i,\cdot)}$ for all $i\in\set{1}{m}$.
We have 
\begin{align*}
\stat{P-p_{\bar x}}&=\frac{1}{2}\sum_{i\in\Delta}\sum_{j=1}^m\abs{\tilde p(i)\tilde p_i(j)-\frac{\abs{\bar A(i,j)}^2}{\fnorm{\bar A}^2}}+
\frac{1}{2}\sum_{i\notin\Delta}\abs{\tilde p(i)-0}\\
&=\frac{1}{2}\sum_{i\in\Delta}\sum_{j=1}^m\abs{\tilde p(i)\tilde p_i(j)- p(i)  p_i(j)}+
\frac{1}{2}\sum_{i\notin\Delta}\abs{\tilde p(i)-p(i)}\\
&=\frac{1}{2}\sum_{i\in\Delta}\sum_{j=1}^m\abs{\tilde p(i)\tilde p_i(j)- \tilde p(i)  p_i(j)}+
\frac{1}{2}\sum_{i\in\Delta}\sum_{j=1}^m\abs{\tilde p(i) p_i(j)- p(i)  p_i(j)}+
\frac{1}{2}\sum_{i\notin\Delta}\abs{\tilde p(i)-p(i)}\\
&=\sum_{i\in\Delta}\tilde p(i)\stat{\tilde p_i-p_i}+
\frac{1}{2}\sum_{i\in\Delta}\abs{\tilde p(i)-p(i)}
+
\frac{1}{2}\sum_{i\notin\Delta}\abs{\tilde p(i)-p(i)}\\
&
\le 2\varepsilon.
\end{align*}

We can implement $\Q{}(y)$ at cost $\costq{u}+\costq{v}$. We have $\norm{x}=\fnorm{A}$ and $\norm{y}=\norm{u}\norm{v}$.
We then use Proposition \ref{prop:IP} to estimate the inner product, which gives the claimed approximation factor and complexity.
\end{proof}

\subsection{Dequantizing the Quantum Singular Value Transformation}\label{sub:app2}
In this subsection we consider the following problem.

\begin{center}
\fbox{
\begin{minipage}{15 cm} \vspace{2mm}

{\sf Quantum Singular Value Transformation} 

\begin{description}[listparindent=-10mm]
\item[Input:]a matrix $A\in\Comp^{m\times n}$ with $\fnorm{A}=1$, a vector $b\in\Comp^n$ with $\norm{b}=1$, an even polynomial $\myp$ of degree $d\ge 2$ such that $\abs{\myp(x)}\le 1$ for all $x\in[-1,1]$, and two parameters $\delta\in(0,1]$ and $\eta\in(0,1/2]$.
\item[Goal:]
find with probability at least $1-\delta$ a vector $v\in \Comp^n$ such that 
\begin{equation}\label{cond:qsvt}
\norm{v-\myp(\sqrt{A^\dagger A})b}\le \eta\norm{\myp(\sqrt{A^\dagger A})b}.
\end{equation}
\end{description}
\end{minipage}
}
\end{center}

Quantum algorithms for the Quantum Singular Value Transformation \cite{Gilyen+STOC19} output with probability at least $1-\delta$ a quantum state proportional to the vector $\myp(\sqrt{A^\dagger A})b$ in time 
\[
\tilde O\left(\frac{d\log(1/\delta)}{\norm{p(\sqrt{A^\dagger A})b}}\right)
\]
when $A$ and $b$ are given in QRAM.  The dequantization framework from \cite{Chia+JACM22} assumes that we have perfect sampling-and-query access to $A$ and $b$ (i.e., $\SQ{}(A)$ and $\SQ{}(b)$), and gives a classical algorithm
that after a preprocessing stage running in time 
\[
\tilde O\left(
\frac{d^{16}\log^3(1/\delta)}{\eta^6\norm{\myp(\sqrt{A^\dagger A})b}^6}
\right),
\] 
implements $\RSQ{0,\delta',\eta'}(v)$ for a vector $v$ satisfying Condition~(\ref{cond:qsvt}) with probability at least $1-\delta$
at cost
\[
\tilde O\left( \frac{d^{12}\log^2(1/\delta)\log(1/\delta')}{\eta^4\eta'^2\norm{\myp(\sqrt{A^\dagger A})b}^6}\right),
\]
for any $\delta',\eta'\in (0,1]$.
We show how to achieve a similar running time when given only $\SQ{\varepsilon}(A)$ and $\SQ{\varepsilon}(b)$ for $\varepsilon>0$, when $\varepsilon$ is small enough. More generally, we also consider the case where we have $\OSQ{\varepsilon,\phi}(A)$ and $\OSQ{\varepsilon,\phi}(b)$ with $\phi>1$. 

\begin{theorem}\label{th:deq-evenSVT}
There exists an absolute constant $\tau\in(0,1)$ such that the following holds.
Assume that we have $\OSQ{\varepsilon,\phi}(A)$ and $\OSQ{\varepsilon,\phi}(b)$ at cost $\costosq{A}=\costosq{b}=\tilde O(1)$ for some $\phi\ge 1$ and $\varepsilon\ge 0$. For any $\varepsilon',\delta',\eta'\in (0,1]$, if 
\begin{equation}\label{cond5}
\varepsilon\le \varepsilon'\cdot 
\frac{\tau\eta^2\norm{\myp(\sqrt{A^\dagger A})b}^4}{\phi^3d^8\log(1/\delta)}
\end{equation}
then there is a classical algorithm that after a preprocessing stage running in time 
\[
\tilde O\left(
\frac{\phi^6 d^{16}\log^3(1/\delta)}{\eta^6\norm{\myp(\sqrt{A^\dagger A})b}^6}
\right)
\]
implements $\RSQ{\varepsilon',\delta',\eta'}(v)$ for a vector $v$ satisfying Condition (\ref{cond:qsvt}) with probability at least $1-\delta$ at cost
\[
\tilde O\left( \frac{\phi^5d^{12}\log^2(1/\delta)\log(1/\delta')}{\eta^4\eta'^2\norm{\myp(\sqrt{A^\dagger A})b}^6}\right).
\]
\end{theorem}
\begin{proof}
Our proof follows the same strategy as in \cite{Chia+JACM22}.
Let us consider the polynomial~$\myq$ such that $\myp(x)-\myp(0)=\myq(x^2)$ and define the function $f\colon\Real\to \Comp$ as follows:
\[
f(x)=
\begin{cases}
\myq(-1)&\textrm{if }x\le -1,\\
\myq(x)&\textrm{if }x\in[-1,1],\\
\myq(1)&\textrm{if }x\ge 1.
\end{cases}
\]

As shown in Lemma 6.3 of \cite{Chia+JACM22}, the function $f$ is $L$-Lipschitz on $\Real$ with $L=O(d^2)$, the function $\bar f$ is $\bar L$-Lipschitz on $\Real$ with $\bar L=O(d^4)$, and we have 
\[
L_{\textrm{max}}:=\max_{x\in \Real}\abs{\bar f(x)}=O(d^2).
\]

Let us set 
\begin{align*}
r&= \ceil{\max\left(
\frac{9}{2}\phi^2\ln\left(6/\delta\right),
\frac{112\phi^2\ln(6/\delta)L^2}{\left(\frac{\eta}{2}\norm{\myp(\sqrt{A^\dagger A})b}\right)^2}\right)
}
=O\left(\frac{\phi^2d^4\log(1/\delta)}{\eta^2\norm{\myp(\sqrt{A^\dagger A})b}^2}\right),\\
c&= 
\ceil{\frac{448\phi^2\ln(6/\delta)\bar L^2}{\left(\frac{\eta}{2}\norm{\myp(\sqrt{A^\dagger A})b}\right)^2}}=O\left(\frac{\phi^2d^8\log(1/\delta)}{\eta^2\norm{\myp(\sqrt{A^\dagger A})b}^2}\right).
\end{align*}

We first consider the $(r,\varepsilon,\phi)$-approximate importance sketch $S\in\Real^{r\times m}$ of $A$ associated with $\OSQ{\varepsilon,\phi}(A)$, and construct its standard representation via Lemma \ref{lemma:createsk} at cost $\tilde O(r)$. Let us write $R=SA$. Observe that by taking the constant $\tau$ small enough, any $\varepsilon$ satisfying Condition~(\ref{cond5}) also satisfies the condition $\varepsilon\le \phi/6$. From Lemma \ref{lemma:norm}, with probability at least $1-\delta/3$ we thus have
\begin{equation}\label{eq:SVT1-k}
\fnorm{R}^2\in\left[\frac{1}{2} \fnorm{A}^2,\frac{3}{2}\fnorm{A}^2\right].
\end{equation}
We assume below that this inequality holds. 
From Lemma \ref{lemma:sample2}, we can then implement $\OSQ{\varepsilon,\phi'}(R^\dagger)$ with $\phi'\le 2\phi$ at cost $\tilde O(r)$.

We now consider the $(c,\varepsilon,\phi')$-approximate importance sketch $T\in\Real^{c\times n}$ of $R^\dagger$ associated with $\OSQ{\varepsilon,\phi'}(R^\dagger)$, and construct its standard representation via Lemma \ref{lemma:createsk} at cost $O(c\cdot \costosq{R^\dagger})=\tilde O(cr)$.
Let us write $C=SAT^\dagger$. We now apply Theorem \ref{th:SVT} (with $\chi=\infty$ since our bounds on the Lipschitz constants $L$ and $\bar L$ hold over $\Real$). Observe that by taking the constant $\tau$ small enough, any $\varepsilon$ satisfying Condition (\ref{cond5}) also satisfies Condition (\ref{cond1}) with $\gamma=\frac{\eta}{2}\norm{\myp(\sqrt{A^\dagger A})b}$.
Theorem~\ref{th:SVT} thus shows that 
\[
\fnorm{R^\dagger \bar{f}(CC^\dagger)R -f(A^\dagger A)}\le \frac{\eta}{2}\norm{\myp(\sqrt{A^\dagger A})b}
\]
holds with probability at least $1-\delta/3$. Note that we then have 
\[
\norm{R^\dagger \bar{f}(CC^\dagger)Rb +\myp(0)b-\myp(\sqrt{A^\dagger A})b}\le \fnorm{R^\dagger \bar{f}(CC^\dagger)R -f(A^\dagger A)}\norm{b}\le \frac{\eta}{2}\norm{\myp(\sqrt{A^\dagger A})b}.
\]
We now explain how to implement access to a good approximation $v$ of the vector $R^\dagger \bar{f}(CC^\dagger)Rb +\myp(0)b$. 

We first explain how to explicitly compute a vector $u\in\Comp^{r}$ close to $Rb$. First, we take
\[
r'=\ceil{\frac{7\ln(6/\delta)\phi^2}
{\left(\eta\norm{\myp(\sqrt{A^\dagger A})b}/(16L_{\textrm{max}})\right)^{2}}
}
=
O\left(\frac{\phi^2d^4\log(1/\delta)}{\eta^2\norm{\myp(\sqrt{A^\dagger A})b}^2}\right)
\]
and consider the $(r',\varepsilon,\phi',\phi)$-approximate joint importance sketch $\Sigma\in\Real^{r'\times n}$ of $R^\dagger$ and $b$ associated with $\OSQ{\varepsilon,\phi'}(R^\dagger)$ and $\OSQ{\varepsilon,\phi}(b)$, and construct its standard representation, which can be done at cost $O(r'\cdot (\costosq{R^\dagger}+\costosq{b}))=\tilde O(rr')$ via Lemma \ref{lemma:createsk}. We set $u=R\Sigma^\dagger \Sigma b$. Note that $u$ can be computed explicitly at cost $O(rr')$ by computing explicitly $R\Sigma^\dagger\in \Comp^{r\times r'}$ and $\Sigma b\in\Comp^{r'}$, and then performing explicitly the matrix-vector product. By taking the constant $\tau$ small enough, we can ensure that any $\varepsilon$ satisfying Condition (\ref{cond5}) also satisfies the inequality $\varepsilon\le \frac{\eta\norm{p(\sqrt{A^\dagger A})b}}{32L_{\textrm{max}}\sqrt{2}\phi}$. Theorem \ref{th:MM} thus guarantees that 
\begin{align*}
\norm{u-Rb}\le 
\left(\frac{1}{16L_{\textrm{max}}}+\frac{1}{16L_{\textrm{max}}}\right)
\eta\norm{\myp(\sqrt{A^\dagger A})b}\fnorm{R}\norm{b}\le
\frac{\eta}{4L_{\textrm{max}}}\norm{\myp(\sqrt{A^\dagger A})b}
\end{align*}
with probability at least $1-\delta/3$. We assume below that this is the case. Note that 
\begin{equation}\label{eq:SVT2}
\norm{u}\le \norm{Rb}+\frac{\eta}{4L_{\textrm{max}}}\norm{\myp(\sqrt{A^\dagger A})b}\le 2+\frac{\eta}{4L_{\textrm{max}}}=O(1).
\end{equation}
Since $u\in\Comp^r$ and $C\in\Comp^{r\times c}$, we can compute explicitly the vector $w=\bar{f}(CC^\dagger) u$ using $O(r^2c+dr+r^2)$ arithmetic operations by first computing explicitly $\bar{f}(CC^\dagger)$ and then performing the matrix-vector product. 

All the steps we have described so far corresponds to the preprocessing stage. The overall complexity of this stage is 
\[
\tilde O\left( rr'+r^2c+dr\right)=
\tilde O\left(r^2c\right)=
\tilde O\left(
\frac{\phi^6 d^{16}\log^3(1/\delta)}{\eta^6\norm{\myp(\sqrt{A^\dagger A})b}^6}
\right).
\]

Let us write $v=R^\dagger \bar f(CC^\dagger) u+\myp(0)b$. Putting everything together, the following upper bound holds with probability at least $1-\delta$:
\begin{align*}
\norm{v-\myp(\sqrt{A^\dagger A})b}&=
\norm{R^\dagger \bar f(CC^\dagger) u-f(A^\dagger A)b}&\\
&\le 
\norm{R^\dagger \bar{f}(CC^\dagger)(u-Rb)}+
\norm{R^\dagger \bar{f}(CC^\dagger)Rb -f(A^\dagger A)b}\\
&\le \norm{R^\dagger}\norm{\bar{f}(CC^\dagger)}\norm{u-Rb}+\norm{R^\dagger \bar{f}(CC^\dagger)Rb -f(A^\dagger A)b}\\
&\le 2L_{\textrm{max}}\norm{u-Rb}+\fnorm{R^\dagger \bar{f}(CC^\dagger)R -f(A^\dagger A)}\\
&\le \eta \norm{\myp(\sqrt{A^\dagger A})b}.
\end{align*}

Finally, we explain how to implement approximate randomized sampling-and-query access from $v$.
Observe that $v$ is a linear combination of $r+1$ vectors: a linear combination of the~$r$ columns of $R^\dagger$ (i.e., rows of $R$) and $b$, where the coefficients of the linear combination are given by $w$ and $\myp(0)$. Via Lemma \ref{lemma:sample1}, we can implement $\OSQ{\varepsilon,\phi''}(R)$ with $\phi''\le 2\phi$ at cost $\tilde O(1)$. Proposition \ref{prop:lin1} guarantees that we can implement $\OSQ{\varepsilon,\varphi}(v)$ at cost $\tilde O(r)$, with 
\[
\varphi=O\left((r+1)\frac{\phi''\sum_{i=1}^r\abs{w(i)}^2\norm{R(i,\cdot)}^2+\phi\abs{\myp(0)}^2\norm{b}^2}{\norm{v}^2}\right).
\]
Using $\norm{w}\le L_{\textrm{max}}\norm{u}$, Inequalities (\ref{eq:SVT1-k}) and (\ref{eq:SVT2}), the assumption $\norm{b}=1$ and the fact $\norm{v}=\Theta(\smallnorm{\myp(\sqrt{A^\dagger A})b})$, which holds since we are assuming $\eta\le 1/2$, we get 
\[
\varphi=O\left(r\frac{\phi L^2_{\textrm{max}}\fnorm{R}^2+\phi}{\norm{v}^2}\right)
=O\left( \frac{\phi^3 d^8\log(1/\delta)}{\eta^2\norm{\myp(\sqrt{A^\dagger A})b}^4}\right).
\]
 By taking the constant~$\tau$ small enough, any $\varepsilon$ satisfying Condition (\ref{cond5}) also satisfies the inequality $3\varepsilon \varphi\le \varepsilon'$.
We then use Proposition \ref{prop:oversampling} to get $\RSQ{\varepsilon',\delta',\eta'}(v)$ 
at cost
\[
O\left(\frac{\varphi\log(1/\delta')}{\eta'^2}\costosq{v}\right)=
\tilde O\left(\frac{\varphi\log(1/\delta')}{\eta'^2}r\right)=
\tilde O\left( \frac{d^{12}\phi^5\log^2(1/\delta)\log(1/\delta')}{\eta^4\eta'^2\norm{\myp(\sqrt{A^\dagger A})b}^6}\right),
\]
as claimed.
\end{proof}

\subsection{Dequantizing the sparse Quantum Singular Value Transformation}\label{sub:app8}
In this subsection, we discuss the main problem considered by Gharibian and Le Gall \cite{GG22}. 

For any integer $s\ge 1$, we say that a matrix is $s$-sparse if it contains at most $s$ nonzero entries per row and column. The problem considered in \cite[Section 4]{GG22} is as follows.

\vspace{4mm}
\begin{center}
\fbox{
\begin{minipage}{0.96\textwidth} \vspace{2mm}

{\sf Sparse Quantum Singular Value Transformation (SparseQSVT)} 
\begin{description}[listparindent=-10mm]
\item[Input:]
an $s$-sparse matrix $A\in\Comp^{m\times n}$ for $s\ge 2$ such that $\norm{A}\le 1$, two vectors $u,v\in \Comp^{n}$ such that $\norm{u}\le 1$ and $\norm{v}\le 1$, an
even polynomial $\myp\in\Real[x]$ of degree $d$ with $|\myp(x)|\le 1$ for all $x\in[-1,1]$, a parameter $\eta\in(0,1]$.
\vspace{1mm}

\item[Goal:] compute an estimate $\hat{z}\in \Comp$ such that $|\hat z-v^\dagger \myp(\sqrt{A^\dagger A})u|\le \eta$.
\end{description}
\end{minipage}
}
\end{center}
\vspace{4mm}

The main difference with the problem considered in Section \ref{sub:app2} is that the normalization on $A$ is now $\norm{A}\le 1$ instead of $\fnorm{A}=1$. Note that this is a significantly weaker assumption since the inequality $\norm{A}\le\fnorm{A}$ always holds and $\fnorm{A}$ can be as large as $\sqrt{n}\norm{A}$. On the other hand, with this weaker assumption the problem can only be solved efficiently in the classical setting for special types of matrices, such as $s$-sparse matrices. The main dequantization result from \cite{GG22} indeed shows that when given $\Q{}(A)$, $\Q{}(u)$ and $\SQ{}(v)$  with costs $\costq{A}=\costq{u}=\costsq{v}=\tilde O(1)$, the problem {\sf SparseQSVT} can be solved classically with probability at least $1-1/\poly(n)$ in $\tilde O(s^{d}/\eta^2)$ time.\footnote{The result from \cite{GG22} is actually slightly more general: it also works when $v$ can be accessed by a slight generalization of sampling-and-query access (called ``$\zeta$-sampling-access'' in \cite{GG22}), where the sampling is done with small multiplicative error (see Definition 2.4 in \cite{GG22}).} This result is obtained by showing how to implement efficiently query access to the vector $\myp(\sqrt{A^\dagger A})u$ when $A$ is $s$-sparse, and then using the inner product estimation technique from \cite{TangSTOC19}. By using instead our technique for robust estimation of the inner product, we show that this dequantization result even holds when considering $\SQ{\varepsilon}(v)$ for $\varepsilon>0$, when $\varepsilon$ is small enough.
\begin{theorem}\label{th:sparseQSVT}
Given $\Q{}(A)$, $\Q{}(u)$ and $\SQ{\varepsilon}(v)$ with $\varepsilon\le \eta/9$ and costs $\costq{A}=\costq{u}=\costsq{v}=\tilde O(1)$, the problem {\sf SparseQSVT} can be solved classically with probability at least $1-1/\poly(n)$ in $\tilde O(s^{d}/\eta^2)$ time.
\end{theorem}
\begin{proof}
Lemma 8 in \cite{GG22} shows that given $\Q{}(A)$ and $\Q{}(u)$ at cost $\costq{A}=\costq{u}=\tilde O(1)$, we can implement query access to $\myp(\sqrt{A^\dagger A})u$ at cost $\tilde O(s^{d})$. Let $w\in \Comp^n$ denote the complex conjugate of the vector $\myp(\sqrt{A^\dagger A})u$.  We immediately get $\Q{}(w)$ from $\Q{}(\myp(\sqrt{A^\dagger A})u)$. 

Observe that $\norm{w}=\smallnorm{\myp(\sqrt{A^\dagger A})u}\le\norm{u}\le 1$ from our assumptions on $\myp$ and $\norm{u}$. We have 
\[
(v,w)=v^\dagger \myp(\sqrt{A^\dagger A})u.
\]
We can then use Proposition \ref{th:IP} and output at cost 
\[
O\left(\frac{\log(n)}{\eta^2}(\costsq{v}+\costq{w})\right)=
O\left(\frac{s^d\log(n)}{\eta^2}\right),
\]
an estimator $\hat z$ such that 
\[
\abs{\hat z-v^\dagger \myp(\sqrt{A^\dagger A})u}\le(2\sqrt{2\varepsilon}+\frac{\eta}{100})\norm{v}\le \eta
\]
holds with probability at least $1-1/\poly(n)$.
\end{proof}

As mentioned in the introduction, using Theorem \ref{th:sparseQSVT} we can generalize the other dequantization results in \cite[Section 4]{GG22}, and in particular obtain an efficient classical algorithm computing a  constant-precision estimation of the ground state of a local Hamiltonians when given $\SQ{\varepsilon}(u)$ for a guiding vector $u$ that has constant overlap with the ground state (Theorem 1 in~\cite{GG22} showed this result only for the case $\varepsilon=0$). Similarly, the implications to the Quantum PCP conjecture and the No Low-Energy Samplable States (NLSS) conjecture discussed in \cite[Section~5]{GG22} also generalize to the case of approximate sampling-and-query access to the witness states (i.e., $\SQ{\varepsilon}(\ket{\varphi})$ for a witness $\ket{\varphi}$ instead of $\SQ{}(\ket{\varphi})$).
\subsection{Dequantizing quantum algorithms for supervised clustering}\label{sub:app3}
In the supervised clustering problem considered by Lloyd, Mohseni, and Rebentrost \cite{Lloyd+13}, the main computational task reduces to the following problem. 

\begin{center}
\fbox{
\begin{minipage}{14 cm} \vspace{2mm}

{\sf Supervised Clustering}
\begin{description}
\item[Input:]
a matrix $M\in\Real^{n\times d}$ with no all-zero row, a row-vector $w\in\Real^{1\times n}$ and parameters $\eta,\delta\in(0,1)$.
\item[Goal:] output an estimation $\alpha$ of $\norm{wM}^2$ such that 
\[
\abs{\alpha-\norm{wM}^2}\le\eta\fnorm{M}^2\norm{w}^2
\]
holds with probability at least $1-\delta$.
\end{description}
\end{minipage}
}
\end{center}

Ref.~\cite{Lloyd+13} proposed a quantum algorithm for supervised clustering by solving the above problem using the SWAP test in time $\tilde O\left(\frac{\log(1/\delta)}{\eta}\right)$, assuming that the data is stored in QRAM. Prior dequantization works \cite{Chia+JACM22,Tang21} showed how to construct classical algorithms solving this problem in time $O\left(\frac{\log(1/\delta)}{\eta^2}(\costsq{M}+\costq{w})\right)$ when given $\SQ{}(M)$ and $\Q{}(w)$.
 Here we show how that dequantization is possible even when given only $\SQ{\varepsilon}(M)$ for $\varepsilon$ small enough.
\begin{theorem}\label{th:clus}
Given $\SQ{\varepsilon}(M)$ and $\Q{}(w)$ for $\varepsilon\le \frac{\eta^2}{25}$, there is a classical algorithm that solves the problem {\sf Supervised Clustering} at cost $ O\left(\frac{\log(1/\delta)}{\eta^2}(\costsq{M}+\costq{w})\right)$.
\end{theorem}
\begin{proof}
We first observe that $\norm{wM}^2=(u,v)$ for the vectors $u,v\in\Real^{dn^2}$ defined as follows. The coordinates of $u$ and $v$ are indexed by triples $(i,j,k)\in\set{1}{n}\times\set{1}{d}\times\set{1}{n}$.  We set
\begin{align*}
u_{i,j,k}&=M(i,j)\norm{M(k,\cdot)},\\
v_{i,j,k}&=\frac{w_iw_kM(k,j)}{\norm{M(k,\cdot)}}
\end{align*}
for each $(i,j,k)\in\set{1}{n}\times\set{1}{d}\times\set{1}{n}$.
We have 
\begin{align*}
(u,v)&=\sum_{j=1}^d\left(\sum_{i=1}^n  w_iM(i,j)\right)\left(\sum_{k=1}^n  w_kM(k,j)\right)=\norm{wM}^2=(u,v),\\
\norm{u}&=\sqrt{\sum_{i=1}^n\sum_{j=1}^d\sum_{k=1}^n \abs{M(i,j)}^2\norm{M(k,\cdot)}^2}=\fnorm{M}^2\\
\norm{v}&=\sqrt{\sum_{i=1}^n\sum_{j=1}^d\sum_{k=1}^n\frac{
\abs{w_i}^2\abs{w_k}^2\abs{M(k,j)}^2}
{\norm{M(k,\cdot)}^2}}=\norm{w}^2.
\end{align*}

Note that from $\SQ{\varepsilon}(M)$ and $\Q{}(w)$ we can implement $\Q{}(u)$ and $\Q{}(v)$ at cost $O(\costsq{M}+\costq{w})$. In order to implement sample access to $u$, we sample $i$ and $k$ from $\tilde p_{\row{M}}$, then sample $j$ from $\tilde p_{M(i,\cdot)}$. Let $P(i,j,k)$ denote the resulting probability distribution. 
For conciseness, let us write $p=p_{\row{M}}$, $\tilde p=\tilde p_{\row{M}}$ and $p_i=p_{M(i,\cdot)}$, $\tilde p_i=\tilde p_{M(i,\cdot)}$ for each $i\in\set{1}{n}$.
We have
\begin{align*}
\stat{P-p_u}&=\frac{1}{2}\sum_{i=1}^n\sum_{j=1}^d\sum_{k=1}^n \abs{P(i,j,k)-p_u(i,j,k)},\\
&=\frac{1}{2}\sum_{i=1}^n\sum_{j=1}^d\sum_{k=1}^n \abs{
\tilde p(i)\tilde p_{i}(j)\tilde p(k)-
\frac{\norm{M(i,\cdot)}^2}{\fnorm{M}^2}\frac{\abs{M(i,j)}^2}{\norm{M(i,\cdot)}^2}\frac{\norm{M(k,\cdot)}^2}{\fnorm{M}^2}
}\\
&=\frac{1}{2}\sum_{i=1}^n\sum_{j=1}^d\sum_{k=1}^n \abs{
\tilde p(i)\tilde p_{i}(j)\tilde p(k)-
p(i)p_{i}(j)p(k)
}\\
&\le\frac{1}{2}\sum_{i=1}^n\sum_{j=1}^d\sum_{k=1}^n 
\abs{\tilde p(i)\tilde p_{i}(j)\tilde p(k)-p(i)\tilde p_{i}(j)\tilde p(k)}+
\abs{ p(i)\tilde p_{i}(j)\tilde p(k) - p(i) p_{i}(j)\tilde p(k)}\\
&\hspace{70mm}+
\abs{p(i) p_{i}(j)\tilde p(k)-p(i)p_{i}(j)p(k)}\\
&=\frac{1}{2}\sum_{i=1}^n\abs{\tilde p(i)-p(i)}\sum_{j=1}^d\tilde p_{i}(j)\sum_{k=1}^n \tilde p(k)+
\frac{1}{2}\sum_{i=1}^n p(i) \sum_{j=1}^d\abs{\tilde p_{i}(j)-p_i(j)}\sum_{k=1}^n \tilde p(k)
\\
&\hspace{70mm}
+\frac{1}{2}\sum_{i=1}^np(i)\sum_{j=1}^d p_{i}(j)\sum_{k=1}^n \abs{\tilde p(k)-p(k)}\\
&= \stat{\tilde p -p}+\sum_{i=1}^np(i)\stat{\tilde p_i-p_i}+\stat{\tilde p -p}\\
&\le
3\varepsilon.
\end{align*}
 This means that we can implement $\SQ{3\varepsilon}(u)$ at cost $\tilde O(\costsq{M})$. 

We can now use Proposition \ref{prop:IP}: we can output an estimator $\alpha$ such that 
\[
\abs{\alpha-(u,v)}\le \left(2\sqrt{6\varepsilon}+\frac{\eta}{100}\right)\norm{u}\norm{v}=
\left(2\sqrt{6\varepsilon}+\frac{\eta}{100}\right)\fnorm{M}^2\norm{w}^2\le 
\eta\fnorm{M}^2\norm{w}^2
\]
 holds with probability at least $1-\delta$ at cost $\tilde O\left(\frac{\log(1/\delta)}{\eta^2}(\costsq{M}+\costq{w})\right)$.
\end{proof}
\subsection{Dequantizing quantum algorithms for recommendation systems}\label{sub:app4}
In recommendation systems, the main computational task reduces to sampling from a row of a low-rank matrix close to the input matrix (which corresponds to the users' preference matrix). In particular, quantum algorithms \cite{Kerenidis+ITCS17} and their dequantization \cite{TangSTOC19,Gilyen+STOC19} consider and solve the following version.

\begin{center}
\fbox{
\begin{minipage}{15 cm} \vspace{2mm}

{\sf Recommendation Systems}

\begin{description}
\item[Input:]
a matrix $A\in\Real^{m\times n}$, an index $i\in\set{1}{m}$, and four parameters $\sigma,\nu\in(0,1]$ and $\xi,\tau\in(0,1/2]$.
\item[Goal:] 
sample from the probability distribution $p_{\hat{A}(i,\cdot)}$ up to $\delta$ error in total variation distance,
for a matrix $\hat{A}\in\Real^{m\times n}$ such that $\fnorm{\hat{A}-A_{\sigma,\xi}}\le\nu\fnorm{A}$, where $A_{\sigma,\xi}=f(A)$ for some arbitrary function $f$ such that 
\begin{equation}\label{eq:cond-rec}
\begin{cases}
f(x)=x&\textrm{ if }x\ge \sigma(1+\xi)\\
f(x)=0&\textrm{ if }x< \sigma(1-\xi)\\
f(x)\in[0,x]& \textrm{ if } x\in[\sigma(1-\xi),\sigma(1+\xi))
\end{cases}
\end{equation}
holds.
\end{description}
\end{minipage}
}
\end{center}\vspace{2mm}

In applications to recommendation systems, the parameter $\xi$ is taken as a constant (e.g., $\xi=1/6$). 
The matrix $A_{\sigma,\xi}$ represents some kind of low-rank approximation of $A$. The intuition is that if the singular values of the matrix $A$ ``concentrate'' on a few high values, then by choosing $\sigma$ large enough the matrix $A_{\sigma,\xi}$ will have low rank (since $A_{\sigma,\xi}$ has only a few nonzero singular values) and will also be a good approximation of $A$ (since the spectrum concentrates on these high singular values). This typically corresponds to the case where the ratio $\fnorm{A}/\sigma$ is small. 

This parameter $K=\frac{\fnorm{A}^2}{\sigma^2}$ is indeed crucial for analyzing the performance of quantum algorithms for recommendation systems and their dequantization (these algorithms are efficient only when~$K$ is small). We assume that $K\ge 1$, since this is the interesting regime. Kerenidis and Prakash \cite{Kerenidis+ITCS17} have shown how to solve {\sf Recommendation Systems} in time $\tilde O(\sqrt{K})$ when~$A$ is given in QRAM. Chia, Gily\'en, Li, Lin, Tang and Wang~\cite{Chia+JACM22} gave a classical algorithm solving the problem in time
$
\tilde O\left(
\frac{K^{8}\log^3(1/\delta)}{\xi^6\nu^6}+
\frac{K^3\log^2(1/\delta)}{\xi^2\nu^2}\frac{\norm{A(i,\cdot)}^2}{\norm{\hat A(i,\cdot)}^2}
\right)
$ given $\SQ{}(A)$ with $\costsq{A}=\tilde O(1)$, which improved the complexity of the first dequantization algorithm by Tang \cite{TangSTOC19}. We show how to obtain the same complexity when given only $\SQ{\varepsilon}(A)$ for $\varepsilon$ small enough, or more generally even $\OSQ{\varepsilon,\phi}(A)$ with $\phi>1$.

\begin{theorem}\label{th:Recom}
Assume that we have $\OSQ{\varepsilon,\phi}(A)$ with $\costsq{A}=\tilde O(1)$, for any $\phi\ge 1$ and any
\[ 
\varepsilon\le \min\left(\frac{\nu\xi}{48 K^2\phi},
\frac{\delta\norm{\hat A(i,\cdot)}^2}{45K\phi^2\norm{A(i,\cdot)}^2}\right).
\]
There is a classical algorithm that solves the problem {\sf Recommendation Systems} in time 
\[
\tilde O\left(
\frac{K^{8}\phi^6\log^3(1/\delta)}{\xi^6\nu^6}+
\frac{K^3\phi^6\log^2(1/\delta)}{\xi^2\nu^2}\frac{\norm{A(i,\cdot)}^2}{\norm{\hat A(i,\cdot)}^2}
\right).
\]
\end{theorem}
\begin{proof}
Our proof follows the same approach as in \cite{Chia+JACM22}.
Let us define the function
\[
\begin{cases}
t(x)=1&\textrm{ if } x\ge (1+\xi)^2\sigma^2,\\
t(x)=0&\textrm{ if } x<(1-\xi)^2\sigma^2,\\
t(x)=\frac{x-(1-\xi)^2\sigma^2}{4\xi\sigma^2}&\textrm{ if } x\in[(1-\xi)^2\sigma^2,(1+\xi)^2\sigma^2),
\end{cases}
\]
and set $A_{\sigma,\xi}=A \:t(A^\dagger A)$. Note that this matrix satisfies the conditions of Equations (\ref{eq:cond-rec}).
We have 
\[
\begin{cases}
\bar t(x)=1/x&\textrm{ if } x\ge (1+\xi)^2\sigma^2,\\
\bar t(x)=0&\textrm{ if } x<(1-\xi)^2\sigma^2,\\
\bar t(x)=\frac{1}{4\xi\sigma^2}-\frac{(1-\xi)^2\sigma^2}{4\xi\sigma^2x}&\textrm{ if } x\in[(1-\xi)^2\sigma^2,(1+\xi)^2\sigma^2).\\
\end{cases}
\]
Observe that $t(x)$ is $\frac{1}{4\xi\sigma^2}$-Lipschitz over $\Real$ and $\bar t(x)$ is $\frac{1}{4\xi(1-\xi)^2\sigma^4}$-Lipschitz over $\Real$. Also observe that 
\begin{equation}\label{equp}
\max_{x\in\Real}\{\bar t(x)\}= \frac{1}{(1+\xi)^2\sigma^2}\le \frac{1}{\sigma^2}.
\end{equation}

We choose appropriate integers $r\in\set{1}{m}$ and $c\in\set{1}{n}$ such that 
\begin{align*}
r&= \Theta\left(
\phi^2\log(1/\delta)\left(1+\frac{\fnorm{A}^4}{\xi^2\sigma^4\nu^2}\right)
\right)=\Theta\left(\frac{K^2\phi^2\log(1/\delta)}{\xi^2\nu^2}\right),\\
c&= \Theta\left(
\frac{\phi^2\log(1/\delta)\fnorm{A}^{8}}{\xi^2(1-\xi^2)^2\sigma^8\nu^2}
\right)=
\Theta\left(\frac{K^4\phi^2\log(1/\delta)}{\xi^2\nu^2}\right).
\end{align*} 
Consider the $(r,\varepsilon,\phi)$-approximate importance sketch $S\in\Real^{r\times m}$ of $A$ associated with $\OSQ{\varepsilon,\phi}(A)$. 
From Lemma \ref{lemma:norm}, with probability at least $1-\delta/5$ we have\footnote{Note that the inequality $\varepsilon\le \frac{\nu\xi}{48\phi K^2}$ implies that $2\varepsilon\phi\le 1/3$, which is enough to get the bounds on $\fnorm{R}^2$.}
\[
\fnorm{R}^2\in\left[\frac{1}{2} \fnorm{A}^2,\frac{3}{2}\fnorm{A}^2\right].
\]
We assume below that this inequality holds. Let us write $R=SA$. From Lemma \ref{lemma:sample2}, we can implement $\OSQ{\varepsilon,\phi'}(R^\dagger)$ with $\phi'=2\phi$ at cost $\tilde O(r)$. We now consider the $(c,\varepsilon,\phi')$-approximate importance sketch $T\in\Real^{c\times n}$ of $R^\dagger$ associated with $\OSQ{\varepsilon,\phi'}(R^\dagger)$.
Let us write $C=SAT^\dagger$. Theorem~\ref{th:SVT} (with $\chi=\infty$ since our Lipschitz constants for $t(x)$ and $\bar t(x)$ hold over $\Real$) guarantees that the inequality
\begin{equation}\label{eq:ss1}
\fnorm{R^\dagger \bar t(CC^\dagger) R- t(A^tA)}\le \frac{\nu}{2}
\end{equation}
holds with probability at least $1-\delta/5$.\footnote{Note that the inequality $\varepsilon\le \frac{\nu\xi}{48\phi K^2}$ implies that $\varepsilon$ satisfies Condition (\ref{cond1}) of the statement of Theorem \ref{th:SVT} with $\gamma=\nu/2$.} Assume below that this is the case.
Additionally, we have 
\begin{equation}\label{eq:cond2}
\norm{\sqrt{\bar t(CC^\dagger)}R}=\sqrt{\norm{R^\dagger \bar t(CC^\dagger) R}}\le 
\sqrt{\norm{t(A^\dagger A)}+\nu/2}\le \sqrt{1+\nu/2}.
\end{equation}

We set $\hat A=A'R^\dagger \bar t(CC^\dagger) R$, where $A'\in\Real^{m\times n}$ is defined below. First, observe that since we have $\OSQ{\varepsilon,\phi}(A)$, for each $j\in\set{1}{m}$ we have $\OSQ{\varepsilon,\phi}(A(j,\cdot))$ and thus $\OSQ{\varepsilon,\phi}(A(j,\cdot)^\dagger)$ as well. For some integer $r'\in\set{1}{m}$, we consider the $(r',\varepsilon,\phi,\phi')$-approximate joint matrix sketch $\Sigma_j\in\Real^{r'\times n}$ of $A(j,\cdot)^\dagger$ and $R^\dagger$ associated with $\OSQ{\varepsilon,\phi}(A(j,\cdot)^\dagger)$ and $\OSQ{\varepsilon,\phi'}(R^\dagger)$. We then set the $j$-th row of the matrix $A'$ as $A'(j,\cdot)=A(j,\cdot)\Sigma_j^\dagger \Sigma_j$. By taking
\begin{align*}
r'&= \Theta\left(
\frac{K\phi^2\log(m/\delta)}{\nu^2}
\right),
\end{align*} 
Theorem \ref{th:MM} guarantees that for each $j\in\set{1}{m}$, the inequality 
\begin{equation}\label{eq:123}
\norm{A'(j,\cdot)R^\dagger-A(j,\cdot)R^\dagger}\le 
\left(
2\varepsilon\sqrt{\phi\phi'}+\frac{\sigma\nu}{16\fnorm{A}}
\right)\norm{A(j,\cdot)^\dagger}\fnorm{R^\dagger}
\le \frac{\sigma\nu\norm{A(j,\cdot)}}{4}
\end{equation}
holds with probability at least $1-\delta/(5m)$, since $2\varepsilon\sqrt{\phi\phi'}\le \frac{\sqrt{2}\phi\nu\xi}{24\phi K^2}\le \frac{\sigma\nu}{16\fnorm{A}}$.
The inequality 
\begin{equation}\label{eq:11}
\fnorm{A'R^\dagger-AR^\dagger}=\sqrt{\sum_{j=1}^m \norm{A'(j,\cdot)R^\dagger-A(j,\cdot)R^\dagger}^2}\le \frac{\sigma\nu\fnorm{A}}{4}.
\end{equation}
thus holds with probability at least $1-\delta/5$. 

Putting everything together, the following holds with probability at least $1-3\delta/5$:
\begin{align*}
\fnorm{\hat A- A_{\sigma,\eta}}&=
\fnorm{(A'R^\dagger -AR^\dagger)\bar t(CC^\dagger) R
+AR^\dagger \bar t(CC^\dagger) R- A\: t(A^\dagger A)}\\
&\le \fnorm{(A'R^\dagger -AR^\dagger)\bar t(CC^\dagger) R}
+\fnorm{AR^\dagger \bar t(CC^\dagger) R- A\: t(A^\dagger A)}\\
&\le
\fnorm{A'R^\dagger -AR}\norm{\sqrt{\bar t(CC^\dagger)}}\norm{\sqrt{\bar t(CC^\dagger)}R}
+\norm{A}\fnorm{R^\dagger \bar t(CC^\dagger) R- t(A^\dagger A)}\\
&\le \frac{\sigma\nu\fnorm{A}}{4}\frac{1}{\sigma}\sqrt{1+\nu/2}+\norm{A}\frac{\nu}{2}\\
&\le \nu \fnorm{A},
\end{align*}
where we used Inequalities (\ref{equp}), (\ref{eq:ss1}), (\ref{eq:cond2}) and  (\ref{eq:11})  to derive the third inequality.
We assume below that this upper bound holds.

Finally, we explain how to sample from the row-vector $\hat A(i,\cdot)$.
We first compute explicitly the row-vector $x=A'(i,\cdot) R^\dagger \bar t(CC^\dagger)=A(i,\cdot)\Sigma_i^\dagger \Sigma_i R^\dagger \bar t(CC^\dagger)\in\Real^{1\times r}$ as follows. We compute the matrix $C$ explicitly, then compute the singular value decomposition of $CC^\dagger$, and then the matrix $\bar t(CC^\dagger)$. This can be done in $O(r^2c)$ time. Then, we compute the vector $A(i,\cdot)\Sigma_i^\dagger \Sigma_i R^\dagger$ explicitely in $O(r'+rr')$ time by computing the row-vector $A(i,\cdot)\Sigma_i^\dagger\in \Real^{1\times r'}$ and the matrix $\Sigma_i R^\dagger\in \Real^{r'\times r}$, and then computing their product. We compute $x$ by computing the product of this vector with the matrix $\bar t(CC^\dagger)$, which can be done in $O(r^2)$ time. Finally, in order to sample from $\hat A(i,\cdot)=xR$, we use Proposition \ref{prop:lin1}. We can get $\OSQ{\varepsilon,\varphi}(xR)$ with 
\begin{align}
\varphi&= 
r
\frac{\sum_{i=1}^r \phi\abs{x_i}^2 \norm{R(i,\cdot)}^2}{\norm{\hat A(i,\cdot)}^2}
\le
\frac{\sum_{i=1}^r \phi^2  \abs{x_i}^2 \fnorm{A}^2}{\norm{\hat A(i,\cdot)}^2}
=
\frac{\phi^2\norm{x}^2 \fnorm{A}^2}{\norm{\hat A(i,\cdot)}^2}
\label{ineq:phi}
\end{align}
at cost $O(r)$, where we used Inequality (\ref{eq:up}). We can upper bound $\varphi$ as follows:
\begin{align*}
\varphi&\le\left(\norm{A(i,\cdot)R^\dagger \bar t(CC^\dagger)}+\norm{(A'(i,\cdot)-A(i,\cdot))R^\dagger \bar t(CC^\dagger)}\right)^2
\frac{\phi^2\fnorm{A}^2}{\norm{\hat A(i,\cdot)}^2}\\
&\le
\left(
\norm{A(i,\cdot)}\norm{R^\dagger \sqrt{\bar t(CC^\dagger)}}\norm{\sqrt{\bar t(CC^\dagger)}}
+
\fnorm{(A'(i,\cdot)-A(i,\cdot))R^\dagger} \norm{\bar t(CC^\dagger)}
\right)^2
\frac{\phi^2\fnorm{A}^2}{\norm{\hat A(i,\cdot)}^2}.
\end{align*}
Using (\ref{equp}), (\ref{eq:cond2}) and (\ref{eq:123}), we obtain
\begin{align*}
\varphi
&\le
\left(
\norm{A(i,\cdot)}\frac{\sqrt{1+\nu/2}}{\sigma}
+\frac{\sigma\nu\norm{A(i,\cdot)}}{4}\frac{1}{\sigma^2}
\right)^2
\frac{\phi^2\fnorm{A}^2}{\norm{\hat A(i,\cdot)}^2}\\
&\le
\left(
\sqrt{1+\nu/2}
+\frac{\nu}{4}
\right)^2
\frac{\phi^2\norm{A(i,\cdot)}^2\fnorm{A}^2}{\sigma^2\norm{\hat A(i,\cdot)}^2}\\
&\le
3\phi^2
\frac{\norm{A(i,\cdot)}^2\fnorm{A}^2}{\sigma^2\norm{\hat A(i,\cdot)}^2}\\
&=
3 K\phi^2
\frac{\norm{A(i,\cdot)}^2}{\norm{\hat A(i,\cdot)}^2}.
\end{align*}

Finally, we use Proposition \ref{prop:oversampling} to implement $\RSQ{\varepsilon',\frac{\delta}{5},1}(xR)$ with $\varepsilon'\le 3\varepsilon\varphi$ at cost $O(\phi r\log(1/\delta))$. In particular, we can sample with probability at least $1-\delta/5$ (which gives overall success probability at least $1-4\delta/5$) from a distribution $ p$ such that 
\[
\stat{p -p_{xR}}\le 3\varepsilon\varphi\le \delta/5.
\]
Putting the failure probability into the total variation distance, we get overall total variation distance at most $\delta$, as desired (remember that $xR=\hat A(i,\cdot)$). The overall complexity is 
\[
O\left(
r^2c+rr'+r^2+\varphi r \log(1/\delta)
\right)=
O\left(
\frac{K^{8}\phi^6\log^3(1/\delta)}{\xi^6\nu^6}+
\frac{K^3\phi^6\log^2(1/\delta)}{\xi^2\nu^2}\frac{\norm{A(i,\cdot)}^2}{\norm{\hat A(i,\cdot)}^2}
\right),
\]
as claimed.
\end{proof}
\subsection{Dequantizing quantum algorithms for low-rank matrix inversion}\label{sub:app5}
In this subsection we consider the following problem. We again write $K=\fnorm{A}^2/\sigma^2$ and assume that $K\ge 1$.
\begin{center}
\fbox{
\begin{minipage}{15 cm} \vspace{2mm}

{\sf Matrix Inversion}

\begin{description}
\item[Input:]
a nonzero matrix $A\in\Real^{m\times n}$, a vector $b\in \Comp^m$, and four parameters $\sigma,\delta\in(0,1]$ and $\eta,\xi\in(0,1/2]$.
\item[Assumption:]
$b\in W_\sigma(A)$, where $W_\sigma(A)$ denotes the subspace of $\Comp^m$ corresponding to the right singular vectors of $A$ associated with singular values at least $\sigma$.
\item[Goal:] 
find a vector $\hat x$ such that 
\begin{equation}\label{eq:condinv}
\norm{\hat x-x^\ast}\le \eta\norm{x^\ast}
\end{equation}
holds for $x^\ast=A^+_{\sigma,\xi}b$ with probability at least $1-\delta$,
where $A^+_{\sigma,\xi}=f(A)$ for some arbitrary function $f$ such that 
\begin{equation}\label{eq:condf}
\begin{cases}
f(x)=1/x&\textrm{ if }x\ge \sigma\\
f(x)=0&\textrm{ if }x< \sigma(1-\xi)\\
f(x)\in[0,1/\sigma]& \textrm{ if } x\in[\sigma(1-\xi),\sigma)
\end{cases}
\end{equation}
holds.
\end{description}
\end{minipage}
}
\end{center}\vspace{2mm}

Note that the assumption $b\in W_\sigma(A)$ is only needed to ensure that 
the inequality 
\begin{equation}\label{condb}
\frac{\norm{b}}{\norm{A}}\le \norm{x^\ast}
\end{equation}
holds. This inequality makes easier to state the complexity of our algorithm and has been used (implicitly) in \cite{Chia+JACM22} as well.

The problem {\sf Matrix Inversion} can be considered as a low-rank version of the matrix inversion problem solved by the HHL quantum algorithm \cite{HHL09}, which has been studied in several dequantization works \cite{Chia+20,Chia+JACM22,Gilyen+18}. Ref.~\cite{Chia+JACM22}, in particular, showed that given perfect oversampling-and-query access to $A$ (i.e., $\OSQ{0,\phi}(A)$) and $\Q{}(b)$ at cost $\costosq{A}=\costq{b}=\tilde O(1)$, after a preprocessing step running in
$
\tilde O\left(
\frac{K^{14}\phi^6\log^3(1/\delta)}{\xi^6\eta^6}\right)
$
time it is possible to implement $\RSQ{0,\delta',\eta'}(\hat x)$ for a vector $\hat x$ satisfying Condition (\ref{eq:condinv}) with probability at least $1-\delta$ at cost
\[
\tilde O\left(
\frac{K^{7}\phi^4\log(1/\delta)\log(1/\delta')}{\xi^2\eta^2\eta'^2}
\frac{\norm{x^\ast}^2}{\norm{\hat x}^2}
\right),
\]
for any $\delta',\eta'\in (0,1)$.
We show how to obtain the same complexity when given only $\OSQ{\varepsilon,\phi}(A)$ for $\varepsilon>0$, when $\varepsilon$ is small enough.

\begin{theorem}\label{th:inv}
There exists an absolute constant $\tau\in(0,1]$ such that the following holds.
Assume that we have $\OSQ{\varepsilon,\phi}(A)$ and $\Q{}(b)$ at cost $\costosq{A}=\costq{b}=\tilde O(1)$, for some $\phi\ge 1$ and $\varepsilon\ge 0$. For any $\varepsilon',\delta',\eta'\in (0,1)$, if 
\begin{equation}\label{cond6}
\varepsilon\le 
\frac{\tau\eta^2}{K^3}
\hspace{4mm}\textrm{ and }\hspace{4mm}
\varepsilon\le \varepsilon'\cdot \frac{\tau}{\phi^2K^3}
\end{equation}
then there is a classical algorithm that after a preprocessing stage running in time 
\[
O\left(
\frac{K^{14}\phi^6\log^3(1/\delta)}{\xi^6\eta^6}\right)
\]
implements $\RSQ{\varepsilon',\delta',\eta'}(\hat x)$ for a vector $\hat x$ satisfying Condition (\ref{eq:condinv}) with probability at least $1-\delta$ at cost
\[
O\left(
\frac{K^{7}\phi^4\log(1/\delta)\log(1/\delta')}{\xi^2\eta^2\eta'^2}
\frac{\norm{x^\ast}^2}{\norm{\hat x}^2}
\right).
\]
\end{theorem}
\begin{proof}
Our proof follows the same approach as in \cite{Chia+JACM22}.
We define the function
\[
\begin{cases}
t(x)=1/x&\textrm{ if } x\ge \sigma^2,\\
t(x)=\frac{x-(1-\xi)^2\sigma^2}{\xi(2-\xi)\sigma^4}&\textrm{ if } x\in[(1-\xi)^2\sigma^2,\sigma^2),\\
t(x)=0&\textrm{ if } x<(1-\xi)^2\sigma^2,
\end{cases}
\]
and set $A^+_{\sigma,\xi}=t(A^\dagger A)\: A^\dagger b$. Note that this matrix satisfies the conditions of Equation (\ref{eq:condf}).
Observe that $t(x)$ is $L$-Lipschitz over $\Real$ and $\bar t(x)$ is $\bar L$-Lipschitz over $\Real$ with
\[
L=\frac{1}{\xi(2-\xi)\sigma^4}
\hspace{5mm}\textrm{ and }  \hspace{5mm}
\bar L=\frac{1}{\xi(2-\xi)(1-\xi)^2\sigma^6}.
\]
 Also observe that 
\begin{equation}\label{equp2}
\max_{x\in\Real}\{ t(x)\}= \frac{1}{\sigma^2} \hspace{5mm}\textrm{ and }  \hspace{5mm}\max_{x\in\Real}\{ \bar t(x)\}= \frac{1}{\sigma^4}
\end{equation}

We will set \begin{align*}
r&= \Theta\left(
\frac{\phi^2\fnorm{A}^8L^2\log(1/\delta)}{\eta^2}\right)=
\Theta\left(
\frac{\phi^2K^4\log(1/\delta)}{\xi^2\eta^2}\right)
,\\
c&= \Theta\left(
\frac{\phi^2\fnorm{A}^{12}\bar L^2\log(1/\delta)}{\eta^2}\right)=
\Theta\left(
\frac{\phi^2K^6\log(1/\delta)}{\xi^2\eta^2}\right).
\end{align*} 
Consider the $(r,\varepsilon,\phi)$-approximate importance sketch $S\in\Real^{r\times m}$ of $A$ associated with $\SQ{\varepsilon}(A)$ and write $R=SA$. Observe that by taking the constant $\tau$ small enough, any $\varepsilon$ satisfying Condition~(\ref{cond6}) also satisfies the condition $\varepsilon\le \phi/3$.  From Lemma \ref{lemma:norm}, we can then guarantee that with probability at least $1-\delta/3$
\begin{equation}\label{eq:cond8}
\fnorm{R}^2\in\left[\frac{1}{2} \fnorm{A}^2,\frac{3}{2}\fnorm{A}^2\right]
\end{equation}
holds.
We assume below that this is the case. 
From Lemma \ref{lemma:sample2}, we can implement $\OSQ{\varepsilon,\phi'}(R^\dagger)$ with $\phi'\le 2\phi$ at cost $\tilde O(r)$. We now consider the $(c,\varepsilon,\phi')$-approximate importance sketch $T\in\Real^{c\times n}$ of $R^\dagger$ associated with $\OSQ{\varepsilon,\phi'}(R^\dagger)$
and write $C=SAT^\dagger$. Observe that by taking the constant $\tau$ small enough, any $\varepsilon$ satisfying Condition~(\ref{cond6}) also satisfies Condition~(\ref{cond1}) of Theorem~\ref{th:SVT} with $\gamma=\frac{\eta}{2\fnorm{A}^2}$ and $\chi=\infty$. Theorem \ref{th:SVT} thus guarantees that the inequality
\begin{equation}\label{eq:deriv3}
\fnorm{R^\dagger \bar t(CC^\dagger) R- t(A^tA)}\le \frac{\eta}{2\fnorm{A}^2}
\end{equation}
holds with probability at least $1-\delta/3$.
Assume below that this is the case.

Observe that 
\begin{equation}\label{eq:com}
\norm{R^\dagger \bar t(CC^\dagger)}\le 
\norm{R^\dagger \sqrt{\bar t(CC^\dagger)}}\norm{\sqrt{\bar t(CC^\dagger)}}\le
\sqrt{\norm{t(AA^\dagger)}+\frac{\eta}{2\fnorm{A}^2}}\cdot \frac{1}{\sigma^2}\le \frac{2}{\sigma^3},
\end{equation}
where we used the assumption $\eta\le 1/2$ (which implies $\eta\le K/2$ since $K\ge 1$).

We now explain how to get an approximation $u$ of the vector $RA^\dagger b\in\Comp^{r}$. For each $i\in\set{1}{r}$, observe that $b^\dagger A R(i,\cdot)^\dagger =R(i,\cdot) A^\dagger b$. Observe that by taking the constant $\tau$ small enough, we can guarantee that any $\varepsilon$ satisfying Condition~(\ref{cond6}) also satisfies the inequality $\sqrt{\varepsilon}\le \frac{\eta}{64K^{3/2}}$.
Since we have $\OSQ{\varepsilon,\phi}(A)$ and both $\Q{}(R(i,\cdot)^\dagger)$ and $Q{}(b)$, we use Corollary \ref{expest} to compute at cost
$O\left(\frac{\phi K^3\log(r/\delta)}{\eta^2}\right)$
an estimate $\alpha_i$ such that 
\[
\abs{\alpha_i-R(i,\cdot) A^\dagger b}\le \left(4\sqrt{\varepsilon}+\frac{\eta}{16K^{3/2}}\right)\norm{R(i,\cdot)}\fnorm{A}\norm{b}\le \frac{\eta}{8K^{3/2}}\norm{R(i,\cdot)}\fnorm{A}\norm{b}
\]
holds with probability at least $1-\delta/(3r)$. We set $u(i)=\alpha_i$. The total time complexity of computing the whole vector $u$ is 
\begin{equation}\label{eq:31}
O\left(r \frac{\phi K^3\log(r/\delta)}{\eta^2}\right)=
\tilde O\left(\frac{\phi^3K^7\log^2(1/\delta)}{\xi^2\eta^4}\right).
\end{equation}
Then with probability at least $1-\delta/3$ we have 
\begin{equation}\label{eq:deriv1}
\norm{RA^\dagger b- u}= \sqrt{\sum_{i=1}^r \abs{\alpha_i-R(i,\cdot) A^\dagger b}^2}
\le \frac{\eta\sigma^3}{8\fnorm{A}^3}\fnorm{R}\fnorm{A}\norm{b}
\le \frac{\eta\sigma^3}{4}\norm{x^\ast},
\end{equation}
where the last inequality uses (\ref{condb}) and (\ref{eq:cond8}).

We set $\hat x =R^\dagger \bar t(CC^\dagger) u$.
Putting everything together, the following inequalities hold with probability at least $1-\delta$: 
\begin{align*}
\norm{\hat x- x^\ast}&\le 
\norm{R^\dagger \bar t(CC^\dagger) (RA^\dagger b- u)}
+
\norm{R^\dagger \bar t(CC^\dagger) RA^\dagger b- x^\ast}\\
&\le 
\norm{R^\dagger \bar t(CC^\dagger)} \norm{RA^\dagger b- u}
+
\norm{R^\dagger \bar t(CC^\dagger) R- t(A^tA)}\norm{A^\dagger b}\\
&\le 
\frac{\eta}{2}\norm{x^\ast}
+
\frac{\eta}{2\fnorm{A}^2}\norm{A} \norm{b}\\
&\le \eta\norm{x^\ast}, 
\end{align*}
where we used (\ref{eq:deriv3}), (\ref{eq:com}) and (\ref{eq:deriv1}) for the third inequality, and (\ref{condb}) for the last inequality.

Finally, we explain how to implement sampling-and-query access to $\hat x$. We can compute explicitly the matrix $\bar t(CC^\dagger)$ in $O(r^2c)$ time, and then compute the vector $\bar t(CC^\dagger)u\in\Comp^r$ in $O(r^2)$ time.  The overall complexity of this step is 
\[
O\left(
\frac{\phi^6K^{14}\log^3(1/\delta)}{\xi^6\eta^6}\right),
\]
which dominates the complexity of the preprocessing stage (compare with Equation (\ref{eq:31})). 

Since $\hat x$ is a linear combination of the columns of $R^\dagger$ (i.e., the rows of $R$), we can then use Proposition \ref{prop:lin1} to implement $\OSQ{\varepsilon,\varphi}(\hat x)$ at cost $O(r)$ with 
\begin{align*}
\varphi= r
\frac{\sum_{i=1}^r 2\phi|[\bar t(CC^\dagger) u](i)|^2\norm{R(i,\cdot)^\dagger}^2}{\norm{\hat x}^2}
\le
\frac{4\phi^2\norm{\bar t(CC^\dagger) u}^2\fnorm{A}^2}{\norm{\hat x}^2},
\end{align*}
where we used Equation (\ref{eq:up}). Using (\ref{equp2}), (\ref{eq:com}) and (\ref{eq:deriv1}) we further have:
\begin{align*}
\varphi
&=O\left(
\frac{\phi^2(\norm{\bar t(CC^\dagger) RA^\dagger b}+\norm{\bar t(CC^\dagger)(RA^\dagger b-u)})^2\fnorm{A}^2}
{\norm{\hat x}^2}
\right)\\
&=O\left(
\frac{\phi^2(\norm{\bar t(CC^\dagger) R}\norm{A^\dagger} \norm{b}+\norm{\bar t(CC^\dagger)}\norm{RA^\dagger b-u})^2\fnorm{A}^2}
{\norm{\hat x}^2}\right)\\
&=O\left(
\frac{\phi^2\left(\frac{2}{\sigma^3}\norm{A}\norm{b}+\frac{1}{\sigma^4}\frac{\eta\sigma^3}{4}\norm{x^\ast}\right)^2\fnorm{A}^2}
{\norm{\hat x}^2}\right)\\
&=O\left(
\frac{\phi^2\left(\frac{2}{\sigma^3}\fnorm{A}^2\norm{x^\ast}+\frac{1}{\sigma^4}\frac{\eta\sigma^3}{4}\norm{x^\ast}\right)^2\fnorm{A}^2}
{\norm{\hat x}^2}\right)
\end{align*}
Using the assumption $\eta\le 1/2$ (which again implies $\eta\le K/2$ since $K\ge 1$), we obtain
\[
\varphi
=O\left(
\frac{\phi^2K^3 \norm{x^\ast}^2}
{\norm{\hat x}^2}\right).
\]
When the inequality $\norm{\hat x- x^\ast}\le\eta\norm{x^\ast}$ holds (which happens with probability at least $1-\delta$), we get the upper bound $\varphi=O(\phi^2K^3)$ since by assumption we have $\eta\le 1/2$. By taking the constant~$\tau$ small enough, any $\varepsilon$ satisfying Condition (\ref{cond6}) also satisfies the inequality $3\varepsilon \varphi\le \varepsilon'$. We then use Proposition \ref{prop:oversampling} to implement $\RSQ{\varepsilon',\delta',\eta'}(\hat x)$ at cost 
\[
O\left(\frac{\varphi \log(1/\delta')}{\eta'^2}r\right)
=
O\left(
\frac{\phi^4K^{7}\log(1/\delta)\log(1/\delta')}{\xi^2\eta^2\eta'^2}
\frac{\norm{x^\ast}^2}{\norm{\hat x}^2}
\right),
\]
as claimed.
\end{proof}

\section*{Acknowledgments}
The author is grateful to Andr{\'{a}}s Gily{\'{e}}n, Will Rosenbaum, Ewin Tang and Santosh Vempala for helpful correspondence. He also thanks Ansis Rosmanis for valuable discussions. The author was supported by JSPS KAKENHI grants Nos.~JP19H04066, JP20H05966, JP20H00579, JP20H04139, JP21H04879, JP24H00071 and MEXT Quantum Leap Flagship Program (MEXT Q-LEAP) grant No.~JPMXS0120319794.

\bibliographystyle{plain}
\bibliography{dequantization}
\end{document}